\pdfoutput=1
\documentclass[12pt]{article}

\usepackage{amsfonts, amsmath, amssymb, amsthm, cancel, comment, fullpage, latexsym, verbatim}
\usepackage{chngcntr}
\usepackage[title]{appendix}
\usepackage{booktabs}
\usepackage[font=small,labelfont=sc]{caption}
\usepackage[usenames,dvipsnames]{color, xcolor}
\usepackage{enumitem}
\usepackage{graphicx}
\usepackage[
  pdfauthor={Federico Echenique, Taisuke Imai, Kota Saito},
  bookmarks=true, 
  bookmarksopen=false, 
  breaklinks=true, 
  colorlinks=true, 
  citecolor=OliveGreen, 
  filecolor=OliveGreen, 
  linkcolor=OliveGreen, 
  urlcolor=OliveGreen, 
  pdfstartview=FitH, 
]{hyperref}
\usepackage{kbordermatrix}
\usepackage{natbib}
\usepackage{pgfplots}
\usepackage{tikz}

\usepackage{apptools}

\usepackage[T1]{fontenc}
\usepackage[osf,lining]{libertine}
\usepackage[libertine]{newtxmath}
\usepackage[scaled=0.85]{beramono}

\newtheorem{axiom}{Axiom}
\newtheorem{definition}{Definition}
\newtheorem{lemma}{Lemma}

\newtheorem{remark}{Remark}
\newtheorem{theorem}{Theorem}
\AtAppendix{\counterwithin{remark}{section}}

\newcommand{\te}{\text}
\newcommand{\ol}{\overline}

\def\F{\mathbf{F}}
\def\N{\mathbf{N}}
\def\Q{\mathbf{Q}}
\def\Re{\mathbf{R}}
\def\ul{\underline}
\def\epsilon{\varepsilon}
\def\ep{\varepsilon}
\def\Var{\text{Var}}

\newcommand{\qu}{\quad }

\newcommand{\sa}{\sigma}
\newcommand{\da}{\delta}
\newcommand{\X}{\mathcal{X}}

\newcommand{\la}{\lambda}

\newcommand{\pr}{\operatorname{Pr}}

\usepackage{multibib}
\newcites{SOMpapers}{References}

\newcommand*{\thisdraft}{This draft: \today}
\newcommand*{\firstdraft}{First draft: June 22, 2018}

\renewcommand{\baselinestretch}{1.2}
\title{Approximate Expected Utility Rationalization}

\author{
Federico Echenique \and Taisuke Imai \and Kota Saito 
\thanks{Echenique: Division of the Humanities and Social Sciences, California Institute of Technology, \href{mailto:fede@hss.caltech.edu}{\texttt{fede@hss.caltech.edu}}. 
Imai: Department of Economics, LMU Munich, \href{mailto:taisuke.imai@econ.lmu.de}{\texttt{taisuke.imai@econ.lmu.de}}.
Saito: Division of the Humanities and Social Sciences, California Institute of Technology, \href{mailto:saito@caltech.edu}{\texttt{saito@caltech.edu}}. 
We are very grateful to Nicola Persico, who posed questions to us that led to some of the results in this paper, and to Jose Apestegu\'{\i}a, Miguel Ballester, 
Geoffroy de~Clippel, Dan Friedman, Yves Le~Yaouanq, Pietro Ortoleva, Matthew Polisson, John Quah and Kareen Rozen for helpful comments. 
We are also grateful for the feedback provided by numerous audience at 
FUR 2018, 
ESA World Meetings 2018, 
4th Hitotsubashi Summer Institute: Microeconomic Theory, 
2018 European Summer Meeting of the Econometric Society, 
CESifo Area Conference on Behavioural Economics 2018, 
Measuring Individual Well-Being Workshop, and 
2019 European Summer Symposium in Economic Theory. 
This research is supported by Grant SES-1558757 from the National Science Foundation. The authors also acknowledge financial support by the NSF through the grants CNS-1518941 (Echenique) and SES-1919263 (Saito), and the Deutsche Forschungsgemeinschaft through CRC TRR 190 (Imai).} 
}

\date{\thisdraft \\ \firstdraft}

\begin{document}

\maketitle 

\begin{abstract}
We propose a new measure of deviations from expected utility theory. For any positive number~$e$, we give a characterization of the datasets with a rationalization that is within~$e$ (in beliefs, utility, or perceived prices) of expected utility theory. The number~$e$ can then be used as a measure of how far the data is to expected utility theory. We apply our methodology to data from three large-scale experiments. Many subjects in those experiments are consistent with utility maximization, but not with expected utility maximization. Our measure of distance to expected utility is correlated with subjects' demographic characteristics.
\end{abstract}

\clearpage 
\section{Introduction}

Revealed preference theory has traditionally, through its 80-year history, dealt with the empirical content of {\em general} utility maximization. Recent research has, in contrast, turned to the empirical content of {\em specific} utility theories. Mostly the focus has been on expected utility (EU): recent theoretical work seeks to characterize the observable choice behaviors that are consistent with expected utility maximization. At the same time, a number of recent empirical revealed-preference studies use data on choices under risk and uncertainty, in which participants make a series of choices from budget sets. We seek to bridge the gap between the theoretical understanding of expected utility theory, and the machinery needed to analyze experimental data on choices under risk and uncertainty.

Imagine an agent making economic decisions, choosing contingent consumption given market prices and income. Revealed preference theory studies the consistency of such choices with utility maximization. Consistency, however, is a black or white question. The choices are either consistent with EU or they are not. Our contribution is to provide a way to describe the degree to which choices are consistent with EU. We propose a {\em measure} of the degree of a dataset's consistency with EU.

Revealed preference theory has developed measures of consistency with general utility maximization. The most widely used measure is the Critical Cost Efficiency Index (CCEI) proposed by \cite{afriat1972efficiency}. The basic idea in the CCEI is to fictitiously decrease an agent's budget so that fewer options are revealed preferred to a given choice. The CCEI has been widely used to analyze experimental data on choices from budget sets. See, for example, \cite{choi2007}, \cite{ahn2014}, \cite{choi2014more}, \cite{carvalho2016poverty}, \cite{carvalho2017complexity}, and \cite{halevy2018parametric}. All of these experimental studies involve subjects making decisions under risk or uncertainty, and CCEI was proposed as a measure of consistency with general utility maximization, not EU, the most commonly-used theory to explain choices under risk or uncertainty. 

Of course, there is nothing wrong with studying general utility maximization in environments with risk and uncertainty, but the data is ideally suited to studying theories of choice under risk and uncertainty, and it should be of great interest to evaluate EU using this data. We shall argue (on both theoretical and empirical grounds) that our method provides a more accurate and intuitive measure  of consistency with EU than using CCEI.

Our main contribution is to propose a measure of how far a dataset is from being consistent with EU. The measure is different from CCEI: we explain theoretically why our measure, and not CCEI, best captures the distance of a dataset to EU theory. We also argue on empirical grounds that our measure passes ``smell tests'' that CCEI fails. For example, CCEI ignores the manifest violations of EU where subjects make first-order stochastically dominated choices. And CCEI does not correlate well with the property of downward-sloping demand, a property that is implied by EU maximization.\footnote{Roughly speaking, it says that prices and quantities must be inversely related, subject to certain qualifications.} 
We also provide a revealed preference axiomatization of the measure based on observed prices and consumption.

In the sequel, we first lay out the implications of EU that cannot be captured by CCEI, and give an overview of our approach. After a theoretical discussion of our measure of consistency (with objective EU discussed in Section~\ref{section:robustoeu} and subjective EU in Section~\ref{section:robustseu}), we present an empirical application using data from experiments on choices under risk (Section~\ref{section:oeu_application}). 

Our empirical application has two purposes. The first is to illustrate how our method can be applied and to argue that our measure of distance to EU is useful and sensible. 
The second is to offer new insights into existing data. We use data from three large-scale experiments \citep{choi2014more, carvalho2016poverty, carvalho2017complexity}, each with over 1,000 subjects, that involve choices under risk. 
Consistency with general utility maximization is well understood in these studies using CCEI. We test for EU theory using our methodology. 

There are two main take-away messages from our empirical application. 
First, the data shows that there is a gap between consistency with general utility maximization (measured with CCEI) and EU maximization (quantified with our measure). 
Subjects with CCEI close to one, who are largely consistent with utility maximization, exhibit diverse degrees of consistency with EU. Our measure detects violations of a basic property of EU that we term downward-sloping demand, and violations of monotonicity with respect to first-order stochastic dominance. CCEI, on the other hand, is less sensitive to these features in choice data. 
Second, the correlation between closeness to EU and demographic characteristics yields intuitive results. We find that younger subjects, those who have high cognitive abilities, and those who are working, are closer to EU behavior than older, low cognitive ability, or non-working, subjects. For some of the three experiments, we also find that highly educated, high-income, and male subjects, are closer to EU. 
These observations suggest that our measure complements CCEI as an empirical toolkit and provides additional insights on datasets that had been analyzed primarily with CCEI.

\subsection{How to Measure Deviations from EU}
\label{sec:EUCCEI}

The CCEI is meant to test deviations from general utility maximization. If an agent's behavior is not consistent with utility maximization, then it cannot possibly be consistent with EU maximization. Thus it stands to reason that if an agent's behavior is far from being rationalizable as measured by CCEI, then it is also far from being rationalizable with an EU function. The problem is, of course, that an agent's behavior may be rationalizable with a general utility function but not with EU. 

\begin{figure}[t]
\centering
\begin{minipage}{0.3\textwidth}
\centering
\resizebox{\textwidth}{!}{
\begin{tikzpicture}[x=1pt,y=1pt,scale=0.6,]
\definecolor{fillColor}{RGB}{255,255,255}
\path[use as bounding box,fill=fillColor,fill opacity=0.00] (0,0) rectangle (289.08,289.08);
\begin{scope}
\path[clip] (  0.00,  0.00) rectangle (289.08,289.08);

\path[] (  0.00,  0.00) rectangle (289.08,289.08);
\end{scope}
\begin{scope}
\path[clip] ( 30.52, 30.52) rectangle (282.08,282.08);
\definecolor{drawColor}{RGB}{0,0,0}

\path[draw=drawColor,draw opacity=0.80,line width= 1.5pt,line join=round] ( 30.52,247.78) -- (154.41,  0.00);

\path[draw=drawColor,draw opacity=0.80,line width= 1.5pt,line join=round] ( 30.52,187.74) -- (280.84,  0.00);
\definecolor{drawColor}{RGB}{0,100,0}

\path[draw=drawColor,draw opacity=0.80,line width= 1.5pt,dash pattern=on 4pt off 4pt ,line join=round] ( 30.52,224.91) -- (142.97,  0.00);
\definecolor{drawColor}{RGB}{0,0,255}

\path[draw=drawColor,draw opacity=0.80,line width= 1.5pt,dash pattern=on 4pt off 4pt ,line join=round] ( 30.52,147.72) -- (227.48,  0.00);
\definecolor{drawColor}{RGB}{255,0,0}
\definecolor{fillColor}{RGB}{255,0,0}

\path[draw=drawColor,draw opacity=0.80,line width= 0.4pt,line join=round,line cap=round,fill=fillColor,fill opacity=0.80] ( 60.25,165.45) circle (  3.57);

\path[draw=drawColor,draw opacity=0.80,line width= 0.4pt,line join=round,line cap=round,fill=fillColor,fill opacity=0.80] (110.56, 87.69) circle (  3.57);
\definecolor{drawColor}{RGB}{0,0,0}

\node[text=drawColor,anchor=base,inner sep=0pt, outer sep=0pt,fill=white] at ( 73.68,175.08) {$x^a$};

\node[text=drawColor,anchor=base,inner sep=0pt, outer sep=0pt,fill=white] at (124.00, 97.32) {$x^b$};
\end{scope}
\begin{scope}
\path[clip] (  0.00,  0.00) rectangle (289.08,289.08);
\definecolor{drawColor}{RGB}{0,0,0}

\path[->,draw=drawColor,line width= 1.5pt,line join=round] ( 30.52, 30.52) --
	( 30.52,282.08);
\end{scope}
\begin{scope}
\path[clip] (  0.00,  0.00) rectangle (289.08,289.08);
\definecolor{drawColor}{RGB}{0,0,0}

\path[->,draw=drawColor,line width= 1.5pt,line join=round] ( 30.52, 30.52) --
	(282.08, 30.52);
\end{scope}
\begin{scope}
\path[clip] (  0.00,  0.00) rectangle (289.08,289.08);
\definecolor{drawColor}{RGB}{0,0,0}

\node[text=drawColor,anchor=base,inner sep=0pt, outer sep=0pt] at (282.08, 10.00) {$x_1$};
\end{scope}
\begin{scope}
\path[clip] (  0.00,  0.00) rectangle (289.08,289.08);
\definecolor{drawColor}{RGB}{0,0,0}

\node[text=drawColor,anchor=base,inner sep=0pt, outer sep=0pt] at ( 49.02,270.54) {$x_2$};
\end{scope}
\begin{scope}
\path[clip] (  0.00,  0.00) rectangle (289.08,289.08);
\definecolor{drawColor}{RGB}{0,0,0}

\node[text=drawColor,anchor=base west,inner sep=0pt, outer sep=0pt] at (  6.11,274.54) {\textbf{\textsf{A}}};
\end{scope}
\end{tikzpicture}}
\end{minipage}
\begin{minipage}{0.3\textwidth}
\centering
\resizebox{\textwidth}{!}{
\begin{tikzpicture}[x=1pt,y=1pt,scale=0.6]
\definecolor{fillColor}{RGB}{255,255,255}
\path[use as bounding box,fill=fillColor,fill opacity=0.00] (0,0) rectangle (289.08,289.08);
\begin{scope}
\path[clip] (  0.00,  0.00) rectangle (289.08,289.08);

\path[] (  0.00,  0.00) rectangle (289.08,289.08);
\end{scope}
\begin{scope}
\path[clip] ( 30.52, 30.52) rectangle (282.08,282.08);
\definecolor{drawColor}{RGB}{0,0,0}

\path[draw=drawColor,draw opacity=0.80,line width= 1pt,dash pattern=on 4pt off 4pt ,line join=round] ( 30.52, 30.52) -- (282.08,282.08);
\definecolor{drawColor}{RGB}{0,100,0}

\path[draw=drawColor,draw opacity=0.80,line width= 1.5pt,dash pattern=on 4pt off 4pt ,line join=round] ( 30.52,182.59) -- (282.08, 98.74);

\path[draw=drawColor,draw opacity=0.80,line width= 1.5pt,dash pattern=on 4pt off 4pt ,line join=round] ( 41.11,289.08) -- (137.47,  0.00);
\definecolor{drawColor}{RGB}{0,0,0}

\path[draw=drawColor,draw opacity=0.80,line width= 1.5pt,line join=round] ( 30.52,235.20) -- (282.08,  8.79);

\path[draw=drawColor,draw opacity=0.80,line width= 1.5pt,line join=round] ( 30.52,196.32) -- (161.40,  0.00);
\definecolor{drawColor}{RGB}{0,0,255}

\path[draw=drawColor,draw opacity=0.80,line width= 1pt,line join=round] ( 76.26,255.59) --
	( 78.54,239.23) --
	( 80.83,226.33) --
	( 83.12,215.90) --
	( 85.40,207.29) --
	( 87.69,200.07) --
	( 89.98,193.92) --
	( 92.27,188.62) --
	( 94.55,184.01) --
	( 96.84,179.95) --
	( 99.13,176.36) --
	(101.41,173.15) --
	(103.70,170.27) --
	(105.99,167.67) --
	(108.27,165.31) --
	(110.56,163.16) --
	(112.85,161.19) --
	(115.13,159.38) --
	(117.42,157.70) --
	(119.71,156.15) --
	(122.00,154.72) --
	(124.28,153.38) --
	(126.57,152.13) --
	(128.86,150.96) --
	(131.14,149.86) --
	(133.43,148.83) --
	(135.72,147.86) --
	(138.00,146.94) --
	(140.29,146.08) --
	(142.58,145.26) --
	(144.86,144.48) --
	(147.15,143.74) --
	(149.44,143.04) --
	(151.73,142.37) --
	(154.01,141.73) --
	(156.30,141.12) --
	(158.59,140.54) --
	(160.87,139.99) --
	(163.16,139.45) --
	(165.45,138.94) --
	(167.73,138.45) --
	(170.02,137.98) --
	(172.31,137.53) --
	(174.59,137.10) --
	(176.88,136.68) --
	(179.17,136.27) --
	(181.46,135.88) --
	(183.74,135.51) --
	(186.03,135.15) --
	(188.32,134.80) --
	(190.60,134.46) --
	(192.89,134.13) --
	(195.18,133.81) --
	(197.46,133.50) --
	(199.75,133.21) --
	(202.04,132.92) --
	(204.32,132.64) --
	(206.61,132.37) --
	(208.90,132.10) --
	(211.19,131.85) --
	(213.47,131.60) --
	(215.76,131.36) --
	(218.05,131.12) --
	(220.33,130.89) --
	(222.62,130.67) --
	(224.91,130.45) --
	(227.19,130.24) --
	(229.48,130.03) --
	(231.77,129.83) --
	(234.05,129.63) --
	(236.34,129.44) --
	(238.63,129.26) --
	(240.92,129.07) --
	(243.20,128.89) --
	(245.49,128.72) --
	(247.78,128.55) --
	(250.06,128.38) --
	(252.35,128.22) --
	(254.64,128.06) --
	(256.92,127.91) --
	(259.21,127.75) --
	(261.50,127.60) --
	(263.78,127.46) --
	(266.07,127.32) --
	(268.36,127.18) --
	(270.65,127.04);

\path[draw=drawColor,draw opacity=0.80,line width= 1pt,line join=round] ( 69.40,257.41) --
	( 71.68,242.62) --
	( 73.97,228.66) --
	( 76.26,215.46) --
	( 78.54,202.95) --
	( 80.83,191.10) --
	( 83.12,179.86) --
	( 85.40,169.19) --
	( 87.69,159.06) --
	( 89.98,149.44) --
	( 92.27,140.31) --
	( 94.55,131.65) --
	( 96.84,123.44) --
	( 99.13,115.65) --
	(101.41,108.29) --
	(103.70,101.34) --
	(105.99, 94.77) --
	(108.27, 88.60) --
	(110.56, 82.80) --
	(112.85, 77.37) --
	(115.13, 72.30) --
	(117.42, 67.60) --
	(119.71, 63.27) --
	(122.00, 59.29) --
	(124.28, 55.67) --
	(126.57, 52.43) --
	(128.86, 49.56) --
	(131.14, 47.09) --
	(133.43, 45.02) --
	(135.72, 43.41) --
	(138.00, 42.32);
\definecolor{drawColor}{RGB}{255,0,0}
\definecolor{fillColor}{RGB}{255,0,0}

\path[draw=drawColor,draw opacity=0.80,line width= 0.4pt,line join=round,line cap=round,fill=fillColor,fill opacity=0.80] (110.56,163.16) circle (  3.57);

\path[draw=drawColor,draw opacity=0.80,line width= 0.4pt,line join=round,line cap=round,fill=fillColor,fill opacity=0.80] (124.28, 55.67) circle (  3.57);
\definecolor{drawColor}{RGB}{0,0,0}

\node[text=drawColor,anchor=base,inner sep=0pt, outer sep=0pt,fill=white] at (124.00,172.79) {$x^a$};

\node[text=drawColor,anchor=base,inner sep=0pt, outer sep=0pt,fill=white] at (137.72, 65.31) {$x^b$};
\end{scope}
\begin{scope}
\path[clip] (  0.00,  0.00) rectangle (289.08,289.08);
\definecolor{drawColor}{RGB}{0,0,0}

\path[->,draw=drawColor,line width= 1.4pt,line join=round] ( 30.52, 30.52) --
	( 30.52,282.08);
\end{scope}
\begin{scope}
\path[clip] (  0.00,  0.00) rectangle (289.08,289.08);
\definecolor{drawColor}{RGB}{0,0,0}

\path[->,draw=drawColor,line width= 1.4pt,line join=round] ( 30.52, 30.52) --
	(282.08, 30.52);
\end{scope}
\begin{scope}
\path[clip] (  0.00,  0.00) rectangle (289.08,289.08);
\definecolor{drawColor}{RGB}{0,0,0}

\node[text=drawColor,anchor=base,inner sep=0pt, outer sep=0pt] at (282.08, 10.00) {$x_1$};
\end{scope}
\begin{scope}
\path[clip] (  0.00,  0.00) rectangle (289.08,289.08);
\definecolor{drawColor}{RGB}{0,0,0}

\node[text=drawColor,anchor=base,inner sep=0pt, outer sep=0pt,fill=white] at ( 49.02,270.54) {$x_2$};
\end{scope}
\begin{scope}
\path[clip] (  0.00,  0.00) rectangle (289.08,289.08);
\definecolor{drawColor}{RGB}{0,0,0}

\node[text=drawColor,anchor=base west,inner sep=0pt, outer sep=0pt] at (  5.75,274.54) {\textbf{\textsf{B}}};

\node[text=drawColor,anchor=base,inner sep=0pt, outer sep=0pt,fill=white] at (190, 250) {\footnotesize$\text{MRS} = \dfrac{\mu_1 \cancel{u' (x_1^k)}}{\mu_2 \cancel{u' (x_2^k)}} = \dfrac{\mu_1}{\mu_2}$};
\end{scope}
\end{tikzpicture}}
\end{minipage}
\begin{minipage}{0.3\textwidth}
\centering
\resizebox{\textwidth}{!}{
\begin{tikzpicture}[x=1pt,y=1pt,scale=0.6]
\definecolor{fillColor}{RGB}{255,255,255}
\path[use as bounding box,fill=fillColor,fill opacity=0.00] (0,0) rectangle (289.08,289.08);
\begin{scope}
\path[clip] ( 32.02, 32.02) rectangle (282.08,282.08);
\definecolor{drawColor}{RGB}{0,0,0}

\path[draw=drawColor,draw opacity=0.80,line width= 1pt,dash pattern=on 4pt off 4pt ,line join=round] ( 32.02, 32.02) -- (282.08,282.08);
\definecolor{drawColor}{RGB}{0,100,0}

\path[draw=drawColor,draw opacity=0.80,line width= 1.5pt,dash pattern=on 4pt off 4pt ,line join=round] ( 32.02,246.46) -- (278.48,  0.00);

\path[draw=drawColor,draw opacity=0.80,line width= 1.5pt,dash pattern=on 4pt off 4pt ,line join=round] ( 32.02,166.14) -- (198.17,  0.00);
\definecolor{drawColor}{RGB}{0,0,0}

\path[draw=drawColor,draw opacity=0.80,line width= 1.5pt,line join=round] ( 32.02,235.48) -- (282.08, 10.43);

\path[draw=drawColor,draw opacity=0.80,line width= 1.5pt,line join=round] ( 32.02,196.83) -- (163.25,  0.00);
\definecolor{drawColor}{RGB}{0,0,255}

\path[draw=drawColor,draw opacity=0.80,line width= 1pt,line join=round] ( 84.31,266.31) --
	( 86.58,254.76) --
	( 88.85,244.35) --
	( 91.13,234.91) --
	( 93.40,226.32) --
	( 95.67,218.46) --
	( 97.95,211.25) --
	(100.22,204.61) --
	(102.49,198.47) --
	(104.77,192.78) --
	(107.04,187.49) --
	(109.31,182.56) --
	(111.59,177.95) --
	(113.86,173.64) --
	(116.13,169.59) --
	(118.41,165.78) --
	(120.68,162.20) --
	(122.95,158.82) --
	(125.23,155.62) --
	(127.50,152.60) --
	(129.77,149.73) --
	(132.05,147.01) --
	(134.32,144.42) --
	(136.59,141.96) --
	(138.87,139.61) --
	(141.14,137.38) --
	(143.41,135.24) --
	(145.69,133.20) --
	(147.96,131.24) --
	(150.23,129.37) --
	(152.51,127.57) --
	(154.78,125.85) --
	(157.05,124.19) --
	(159.33,122.60) --
	(161.60,121.08) --
	(163.87,119.60) --
	(166.14,118.19) --
	(168.42,116.82) --
	(170.69,115.50) --
	(172.96,114.23) --
	(175.24,113.00) --
	(177.51,111.81) --
	(179.78,110.66) --
	(182.06,109.55) --
	(184.33,108.48) --
	(186.60,107.44) --
	(188.88,106.43) --
	(191.15,105.45) --
	(193.42,104.50) --
	(195.70,103.58) --
	(197.97,102.69) --
	(200.24,101.82) --
	(202.52,100.98) --
	(204.79,100.16) --
	(207.06, 99.36) --
	(209.34, 98.58) --
	(211.61, 97.83) --
	(213.88, 97.10) --
	(216.16, 96.38) --
	(218.43, 95.68) --
	(220.70, 95.01) --
	(222.98, 94.34) --
	(225.25, 93.70) --
	(227.52, 93.07) --
	(229.80, 92.46) --
	(232.07, 91.86) --
	(234.34, 91.27) --
	(236.62, 90.70) --
	(238.89, 90.14) --
	(241.16, 89.60) --
	(243.43, 89.07) --
	(245.71, 88.54) --
	(247.98, 88.04) --
	(250.25, 87.54) --
	(252.53, 87.05) --
	(254.80, 86.57) --
	(257.07, 86.11) --
	(259.35, 85.65) --
	(261.62, 85.20) --
	(263.89, 84.76) --
	(266.17, 84.33) --
	(268.44, 83.91) --
	(270.71, 83.50);

\path[draw=drawColor,draw opacity=0.80,line width= 1pt,line join=round] ( 57.03,266.41) --
	( 59.30,235.32) --
	( 61.58,211.85) --
	( 63.85,193.50) --
	( 66.12,178.76) --
	( 68.40,166.65) --
	( 70.67,156.53) --
	( 72.94,147.93) --
	( 75.22,140.55) --
	( 77.49,134.14) --
	( 79.76,128.52) --
	( 82.03,123.55) --
	( 84.31,119.12) --
	( 86.58,115.16) --
	( 88.85,111.59) --
	( 91.13,108.35) --
	( 93.40,105.41) --
	( 95.67,102.71) --
	( 97.95,100.24) --
	(100.22, 97.97) --
	(102.49, 95.87) --
	(104.77, 93.92) --
	(107.04, 92.11) --
	(109.31, 90.43) --
	(111.59, 88.85) --
	(113.86, 87.38) --
	(116.13, 86.00) --
	(118.41, 84.70) --
	(120.68, 83.48) --
	(122.95, 82.33) --
	(125.23, 81.24) --
	(127.50, 80.21) --
	(129.77, 79.23) --
	(132.05, 78.30) --
	(134.32, 77.42) --
	(136.59, 76.58) --
	(138.87, 75.78) --
	(141.14, 75.02) --
	(143.41, 74.29) --
	(145.69, 73.60) --
	(147.96, 72.93) --
	(150.23, 72.30) --
	(152.51, 71.69) --
	(154.78, 71.10) --
	(157.05, 70.54) --
	(159.33, 70.00) --
	(161.60, 69.48) --
	(163.87, 68.98) --
	(166.14, 68.50) --
	(168.42, 68.04) --
	(170.69, 67.59) --
	(172.96, 67.16) --
	(175.24, 66.74) --
	(177.51, 66.34) --
	(179.78, 65.95) --
	(182.06, 65.57) --
	(184.33, 65.21) --
	(186.60, 64.85) --
	(188.88, 64.51) --
	(191.15, 64.18) --
	(193.42, 63.86) --
	(195.70, 63.55) --
	(197.97, 63.25) --
	(200.24, 62.95) --
	(202.52, 62.67) --
	(204.79, 62.39) --
	(207.06, 62.12) --
	(209.34, 61.86) --
	(211.61, 61.60) --
	(213.88, 61.35) --
	(216.16, 61.11) --
	(218.43, 60.88) --
	(220.70, 60.65) --
	(222.98, 60.42) --
	(225.25, 60.21) --
	(227.52, 59.99) --
	(229.80, 59.79) --
	(232.07, 59.58) --
	(234.34, 59.39) --
	(236.62, 59.19) --
	(238.89, 59.00) --
	(241.16, 58.82) --
	(243.43, 58.64) --
	(245.71, 58.47) --
	(247.98, 58.29) --
	(250.25, 58.13) --
	(252.53, 57.96) --
	(254.80, 57.80) --
	(257.07, 57.64) --
	(259.35, 57.49) --
	(261.62, 57.34) --
	(263.89, 57.19) --
	(266.17, 57.04) --
	(268.44, 56.90) --
	(270.71, 56.76);
\definecolor{drawColor}{RGB}{255,0,0}
\definecolor{fillColor}{RGB}{255,0,0}

\path[draw=drawColor,draw opacity=0.80,line width= 0.4pt,line join=round,line cap=round,fill=fillColor,fill opacity=0.80] (144.48,134.27) circle (  3.57);

\path[draw=drawColor,draw opacity=0.80,line width= 0.4pt,line join=round,line cap=round,fill=fillColor,fill opacity=0.80] ( 88.76,111.73) circle (  3.57);
\definecolor{drawColor}{RGB}{0,0,0}

\node[text=drawColor,anchor=base,inner sep=0pt, outer sep=0pt,fill=white] at (155.84,141.84) {$x^a$};

\node[text=drawColor,anchor=base,inner sep=0pt, outer sep=0pt,fill=white] at (100.13,119.29) {$x^b$};
\end{scope}
\begin{scope}
\path[clip] (  0.00,  0.00) rectangle (289.08,289.08);
\definecolor{drawColor}{RGB}{0,0,0}

\path[->,draw=drawColor,line width= 1.4pt,line join=round,line cap=rect] ( 32.02, 32.02) --
	( 32.02,282.08);
\end{scope}
\begin{scope}
\path[clip] (  0.00,  0.00) rectangle (289.08,289.08);
\definecolor{drawColor}{RGB}{0,0,0}

\path[->,draw=drawColor,line width= 1.4pt,line join=round,line cap=rect] ( 32.02, 32.02) --
	(282.08, 32.02);
\end{scope}
\begin{scope}
\path[clip] (  0.00,  0.00) rectangle (289.08,289.08);
\definecolor{drawColor}{RGB}{0,0,0}

\node[text=drawColor,anchor=base,inner sep=0pt, outer sep=0pt] at (282.08, 10.00) {$x_1$};
\end{scope}
\begin{scope}
\path[clip] (  0.00,  0.00) rectangle (289.08,289.08);
\definecolor{drawColor}{RGB}{0,0,0}

\node[text=drawColor,anchor=base,inner sep=0pt, outer sep=0pt,fill=white] at ( 49.02,270.54) {$x_2$};
\end{scope}
\begin{scope}
\path[clip] (  0.00,  0.00) rectangle (289.08,289.08);
\definecolor{drawColor}{RGB}{0,0,0}

\node[text=drawColor,anchor=base west,inner sep=0pt, outer sep=0pt] at (  5.81,274.59) {\textbf{\textsf{C}}};
\end{scope}
\end{tikzpicture}}
\end{minipage}
\caption{(A) A violation of WARP. (B) A violation of EU: $x_2^a > x_1^a$, $x_1^b > x_2^b$, and $p_1^b/p_2^b < p_1^a/p_2^a$. (C) A choice pattern consistent with EU.}
\label{fig:violation}
\end{figure}
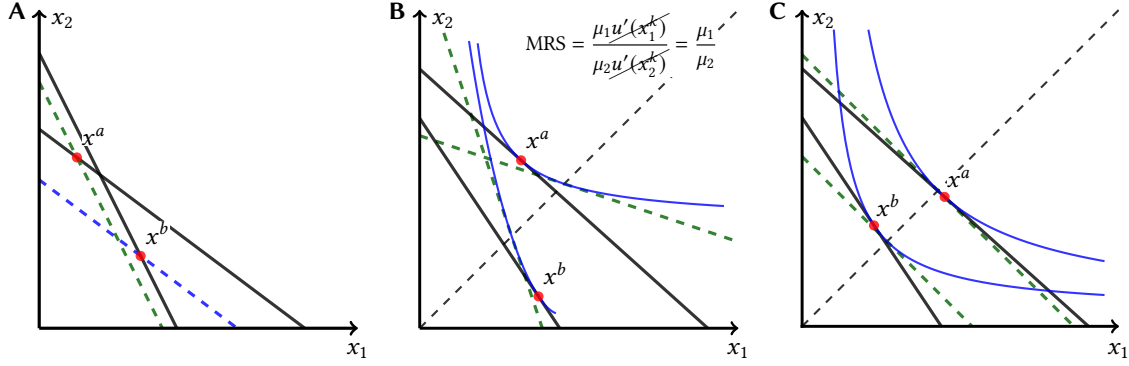

Broadly speaking, the CCEI proceeds by ``amending'' inconsistent choices through the device of changing income. This works for general utility maximization, but it is the wrong way to amend choices that are inconsistent with EU. Since EU is about getting marginal rates of substitution right, prices, not incomes, need to be changed. The problem is illustrated with a simple example in Figure~\ref{fig:violation}. 

Suppose that there are two states of the world, labeled~1 and~2. An agent purchases a state-contingent asset $x=(x_1,x_2)$, given Arrow-Debreu prices $p=(p_1,p_2)$ and her income. Prices and income define a budget set. In Figure~\ref{fig:violation}A, we are given two choices for the agent, $x^a$ and $x^b$, for two different budgets. The choices in Figure~\ref{fig:violation}A are inconsistent with utility maximization: they violate the weak axiom of revealed preference (WARP). When $x^b$ ($x^a$) was chosen, $x^a$ ($x^b$, respectively) was strictly inside of the budget set. This violation of WARP can be resolved by shifting down the budget line associated with choice $x^b$ to the dashed green line passing through $x^a$. Alternatively, the violation can be resolved by shifting down the budget line associated with choice $x^a$ to the dashed blue line passing through $x^b$. CCEI is the smallest of the two shifts that are needed: the smallest proportion of shifting down a budget line to resolve WARP violation. Therefore, the CCEI of this dataset is given by the dashed green line passing through $x^a$. That is, the CCEI is $(p^b \cdot x^a )/(p^b \cdot x^b)$.

Now consider the example in Figure~\ref{fig:violation}B. There are again two choices, $x^a$ and $x^b$, for two different budgets. These choices do not violate WARP, and comply with the theory of utility maximization with $\text{CCEI} = 1$. The choices in the panel are {\em not}, however, compatible with EU. 
To see why, assume that the dataset were rationalized by an expected utility: $\mu_1 u(x^k_1)+\mu_2 u(x^k_2)$, where $(\mu_1, \mu_2)$ are the probabilities of the two states, and $u$ is a (smooth) concave utility function over money. Note that the slope of a tangent line to the indifference curve at a point $x^k$ is equal to the marginal rate of substitution (MRS): $\mu_1 u'(x^k_{1})/ \mu_2 u'(x^k_{2})$. Moreover, at the 45-degree line (i.e., when $x^k_1=x^k_2$), the slope must be equal to $\mu_1\cancel{u'(x^k_{1})}/\mu_2 \cancel{u'(x^k_{2})}= \mu_1/\mu_2$. This is a contradiction because in Figure~\ref{fig:violation}B, the two tangent lines (green dashed lines) associated with $x^a$ and $x^b$ cross each other. Figure~\ref{fig:violation}C shows an example of choices that are consistent with EU. Note that tangent lines at the 45-degree line are parallel in this case. 

Importantly, the violation in Figure~\ref{fig:violation}B cannot be resolved by shifting budget lines up or down, or more generally by adjusting agents' expenditures. The reason is that {\em the empirical content of expected utility is captured by the relation between prices and marginal rates of substitution. The slope, not the level, of the budget line, is what matters.} The basic insight comes from the equality of marginal rates of substitution and relative prices: 
\begin{equation}
\frac{\mu_1 u'(x^k_{1})}{\mu_2 u'(x^k_{2})} = \frac{p^k_1}{p^k_2}. \label{eq:MRSeqprices}
\end{equation} 
Since marginal utility is decreasing, equation~\eqref{eq:MRSeqprices} imposes a negative relation between prices and quantities. The distance to EU is directly related to how far the data is to complying with such a negative relation between prices and quantities. The formal connection is established in Theorem~\ref{theorem:robustoeu}. Empirically,
as we shall see, the degree of compliance of a subject's choices with this ``downward sloping demand'' property, goes a long way to capturing the degree of compliance of the subject's choices with EU.

We propose a measure of how close the data is to being consistent with EU maximization. Our measure is based on the idea that marginal rates of substitution have to conform to EU maximization: whether data conform to equation~\eqref{eq:MRSeqprices}. If one ``perturbs'' marginal utility enough, then a dataset is always consistent with expected utility. Our measure is simply a measure of how large of a perturbation is needed to rationalize the data. Perturbations of marginal utility can be interpreted in three different, but equivalent, ways: as measurement error on prices, as random shocks to marginal utility in the fashion of random utility theory \citep{mcfadden1974frontiers}, or as perturbations to agents' beliefs. For example, if the data in Figure~\ref{fig:violation}B is ``$e$~away'' from being consistent with expected utility given a positive number~$e$, then one can find beliefs $\mu^a$ and $\mu^b$, one for each observation so that EU is maximized for these observation-specific beliefs, and the degree of perturbation of beliefs is bounded by~$e$. 

Our measure can be applied in settings where probabilities are known and objective, for which we develop a theory in Section~\ref{section:robustoeu}, and an application to experimental data in Section~\ref{section:oeu_application}. It can also be applied to settings where probabilities are not known, and therefore subjective (Section~\ref{section:robustseu}).

Finally, we propose a statistical methodology for testing the null hypothesis of consistency with EU (Section~\ref{sec:implementation_perturbation}). Our test relies on a set of auxiliary assumptions. 
The test indicates moderate levels of rejection of the EU hypothesis.

\subsection{Related Literature}
Revealed preference theory has developed tests for consistency with general utility maximization. The seminal papers include \cite{samuelson1938note}, \cite{afria67}, and \cite{varian1982nonparametric}. See \cite{chambers2016revealed} for an exposition of the basic theory. 

More recent work has explored the testable implications of EU theory. This work includes \cite{green1986}, \cite{chambers2016test}, \cite{kubler2014asset}, \cite{echenique2015savage}, and \cite{polissonquahrenou17}. The first four papers focus, as we do here, on rationalizability for risk-averse agents. \cite{green1986} and \cite{chambers2016test} allow for many goods in each state, which our methodology cannot accommodate. \cite{polissonquahrenou17} present a general approach to testing that allows for a test of EU in isolation, not jointly with risk aversion. Our assumptions are the same as in \cite{kubler2014asset} and \cite{echenique2015savage}. 

Compared to most of the existing revealed preference literature on EU, our focus is on measuring consistency with EU, not on providing a test. Our assumption of monetary payoffs and risk aversion is restrictive but consistent with how EU theory is used in economics. Many economic models assume EU together with risk aversion. Our results speak directly to the empirical relevance of such models. A further motivation for focusing on risk aversion is empirical: in the data we have looked at, corner choices are very rare. This would rule out risk-seeking behavior in the context of EU. Thus, arguably, EU and risk-loving behavior would not be a candidate explanation of the experimental data we examine in this paper.

As mentioned, the CCEI was proposed by \cite{afriat1972efficiency}. \cite{varian1990goodness} proposes a modification, and \cite{echenique2011money} and \cite{dean2016measuring} propose alternative measures. \cite{dziewulski2020just} provides a foundation for CCEI based on the model in \cite{dziewulski2016eliciting}, which seeks to rationalize violations of utility-maximizing behavior with a model of just-noticeable differences. Compared to the literature based on the CCEI, we present an explicit model of the errors that would explain the deviation from EU. As a consequence, our measure of consistency with EU is based on a ``story'' for why choices are inconsistent with EU. And, as we have explained above, the nature of EU-consistent choices is poorly reflected in the CCEI's budget adjustments.

\cite{apesteguia2015measure} propose a general method to measure the distance between theory and data in revealed preference settings. For each possible preference relation, they calculate the {\it swaps index}, which counts the number of alternatives that must be swapped with the chosen alternative in order for the preference relation to rationalize the data. Then, \cite{apesteguia2015measure} consider the preference relation that minimizes the total number of swaps in all the observations, weighted by their relative occurrence in the data. \cite{apesteguia2015measure} assume that there is a finite number of alternatives, and thus a finite number of preference relations over the set of alternatives. Because of the finiteness, they can calculate the swaps index for each preference relation and find the preference relation that minimizes the swaps index. This method by \cite{apesteguia2015measure} is not directly applicable to our setup because in our setup, a set of alternatives is a budget set and contains infinitely many elements;  moreover, the number of expected utility preferences relation is infinite.\footnote{In Appendix D.1 of \cite{apesteguia2015measure}, they consider the swaps index for expected utility preferences while assuming the finiteness of the set of alternatives. In their Appendix D.3, without axiomatization, they consider the swaps index for an infinite set of alternatives using the Lebesgue measure to ``count''  the number of swaps. However, they do not study the case where the number of alternatives is infinite and the preference relations are expected utility.} 

There are many other studies of revealed preference that are based on a notion of distance between the theory and the data. For example, \cite{halevy2018parametric} uses such distances as a guide in estimating parametric functional forms for the utility function. 

\cite{polissonquahrenou17} develop a general method called the Generalized Restriction of Infinite Domain (GRID) for testing consistency with models of choice under risk and uncertainty. Using GRID, they provide a way to calculate CCEI for departures from EU. Importantly, and in contrast with our measure, their approach does not rely on risk aversion. They present measures of departure from EU and risk-averse EU. We compare empirically our measure to theirs in Section~\ref{section:oeu_application_results} (the Online Appendix has additional details). 
Suffice it to say here that the measures are similar, but distinct, when applied to the data, and that the differences cannot be attributed to risk aversion. Theoretically, our approach has the advantage of modeling a specific source of deviations from EU, and our results connect the measure to certain observable behavioral patterns. These include exact behavioral patterns described by the theorems, but also an empirically motivated observation that our measure captures compliance with downward-sloping demand.

Finally, \cite{declippel2020relaxed} measure consistency with utility maximization by way of departures from first-order conditions, an approach similar to ours. Their FOC-Departure Index (FDI) can be computed for different classes of utility functions. In particular, their FDI measure for risk-averse expected utility is equivalent to our measure, except for the use of different scaling (their measure $\varepsilon \in [0, 1]$ is the same as a transformation of our measure $e \geq 0$, with $\varepsilon = e / (1+e)$). 
Their axiomatization is different from ours in that their primitives are weak orderings on pairs of price and utility gradient (derivatives of utility function).\footnote{In their paper, $(p, g) \succeq (p_0, g_0)$ means that ``the utility gradient $g$ is farther apart from the price vector $p$ than $g_0$ is from $p_0$.''} On the other hand, we provide an axiomatization based on the observed prices and chosen allocations.
The result in their Proposition~8 is perhaps closest in spirit to our exercise, where they show that computing the measure reduces to checking a set of inequalities. See Remarks~\ref{remark:e-oeu_linear_problem} and~\ref{remark:e-oeu_linear_problem_existence} in Online Appendix~\ref{sec:compute_e} of our paper.  
De~Clippel and Rozen's work is independent and contemporaneous to ours.

\section{Model}
\label{model}

Let $S$ be a finite set of {\em states}. We occasionally use $S$ to denote the number $|S|$ of states. Let $\Delta_{++}(S) = \{ \mu \in \Re^S_{++} \mid \sum_{s=1}^S \mu_s = 1 \}$ denote the set of strictly positive probability measures on $S$.  In our model, the objects of choice are state-contingent monetary payoffs, or {\em monetary acts}. A monetary act is a vector in  $\Re^S_+$.  

\begin{definition}
A {\em dataset} is a finite collection of pairs $(x,p) \in \Re^S_{+}\times \Re^S_{++}$. 
\end{definition}

The interpretation of a dataset $(x^k,p^k)_{k=1}^K$ is that it describes $K$ purchases of a state-contingent payoff $x^k$ at some given vector of prices $p^k$, and income $p^k \cdot x^k = \sum_{s \in S} p^k_s x^k_s$. We sometimes use $K$ to denote the set $\{ 1, \dots, K \}$. 
For any prices $p\in \Re^S_{++}$ and positive number $I>0$, the set $B (p, I) = \{ y \in \Re^S_+ \mid p \cdot y \leq I \}$ is the {\em budget set} defined by $p$ and $I$. 

Expected utility theory requires a decision maker to solve the problem 
\begin{equation}\label{eq:maxEU}
\max_{x \in B(p,I)} \sum_{s\in S} \mu_s u(x_s) , 
\end{equation}
when faced with prices $p \in \Re^S_{++}$ and income $I>0$, where $\mu \in \Delta_{++}(S)$ is a belief and $u$ is a concave utility function over money. We are interested in concave $u$; an assumption that corresponds to risk aversion. 

The belief $\mu$ will have two interpretations in our model. First, in Section~\ref{section:robustoeu}, we shall focus on decisions taken under {\em risk}. The belief $\mu$ will be a known ``objective'' probability measure $\mu^* \in \Delta_{++}(S)$. Then, in Section~\ref{section:robustseu}, we study choice under {\em uncertainty}. Consequently, The belief $\mu$ will be a subjective beliefs, which is unobservable to us as outside observers. 

The following definition formalizes the concept of as-if choices \citep{echenique2015savage}. 
\begin{definition}
A dataset $(x^k,p^k)_{k=1}^K$ is
{\em Objective Expected Utility rational} if there exists a concave and strictly
increasing function $u : \Re_+ \rightarrow \Re$ such that, for all $k$, 
\begin{equation*}
y \in B (p^k, p^k \cdot x^k) \implies 
\sum_{s \in S} \mu_s^* u(y_s) \leq \sum_{s \in S} \mu_s^* u(x^k_s) ,
\end{equation*}
where $\mu^* \in \Delta_{++}(S)$ is an objective probability. 
A dataset $(x^k,p^k)_{k=1}^K$ is
{\em Subjective Expected Utility rational} if there exist
$\mu \in \Delta_{++}(S)$ and a concave and strictly
increasing function $u : \Re_+ \rightarrow \Re$ such that, for all $k$, 
\begin{equation*}
y \in B (p^k, p^k \cdot x^k) \implies 
\sum_{s \in S} \mu_s u(y_s) \leq \sum_{s \in S} \mu_s u(x^k_s) .
\end{equation*}
\end{definition}

When imposed on a dataset, expected utility maximization~\eqref{eq:maxEU} may be too demanding. We are interested in situations where the model in~\eqref{eq:maxEU} holds {\em approximately}. As a result, we shall relax~\eqref{eq:maxEU} by ``perturbing'' some elements of the model. The exercise will be to see if a dataset is consistent with the model in which some elements have been perturbed. Specifically, we shall perturb beliefs, utilities, or prices.

First, consider a perturbation of utility $u$. We allow $u$ to depend on the choice problem $k$ and the realization of the state $s$. We suppose that the utility of consumption $x_s$ in state $s$ is given by $\ep^k_s u(x_s),$ with $\ep^k_s$ being a (multiplicative) perturbation in utility. To sum up, given price $p$ and income $I$, a decision maker solves the problem 
\begin{equation*}
\max_{x \in B(p,I)} \sum_{s \in S} \mu_s \ep^k_s u(x_s) , 
\end{equation*}
when faced with prices $p \in \Re^S_{++}$ and income $I>0$. Here $\{ \ep^k_s \}_{s \in S, k \in K}$ is a set of perturbations, and $u$ is, as before, a concave utility function over money. 

In the second place, consider a perturbation of beliefs. We allow $\mu$ to be different for each choice problem $k$. That is, given price $p$ and income $I$, a decision maker solves the problem 
\begin{equation}\label{eq:maxEUrobust}
\max_{x \in B(p,I)} \sum_{s \in S} \mu^k_s u(x_s) , 
\end{equation}
when faced with prices $p \in \Re^S_{++}$ and income $I > 0$, where $\{ \mu^k \}_{k \in K} \subset \Delta_{++}(S)$ is a set of beliefs and $u$ is a concave utility function over money. 

Finally, consider a perturbation of prices. Our consumer faces perturbed prices  $\tilde{p}^k_s = \ep^k_s p^k_s$, with a perturbation $\ep^k_s$ that depends on the choice problem $k$ and the state $s$. Given price $p$ and income $I$, a decision maker solves the problem 
\begin{equation*}
\max_{x \in B(\tilde{p},I)} \sum_{s\in S} \mu_s u(x_s) , 
\end{equation*}
when faced with income $I>0$ and the perturbed prices $\tilde{p}^k_s = \ep^k_s p^k_s$  for each $k \in K$ and $s \in S$. 

Observe that our three sources of perturbations have different interpretations, each can be traced back to a long-standing tradition for how errors are introduced in economic models. Perturbed prices can be thought of a prices subject to measurement error, measurement error being a very common source of perturbations in econometrics \citep{griliches1986economic}. Perturbed utility is an instance of random utility models \citep{mcfadden1974frontiers}. Finally, perturbations of beliefs can be thought of as a kind of random utility, or as an inability to exactly use probabilities. 
Note that we perturb one source at a time and do not consider combinations of perturbations.

\section{Perturbed Objective Expected Utility}
\label{section:robustoeu}

In this section, we discuss choice under risk: there exists a known ``objective'' belief $\mu^* \in \Delta_{++}(S)$ that determines the realization of states. The experiments we discuss in Section~\ref{section:oeu_application} are all on choice under risk. 

As mentioned above, we go through each of the sources of perturbation: beliefs, utility, and prices. We seek to understand how large a perturbation has to be in order to rationalize a dataset. It turns out that, for this purpose, all sources of perturbations are equivalent.

\subsection{Belief Perturbation}

Deviations from EU are accommodated by allowing a different belief at each observation. So we assume a belief $\mu^k$ for each choice $k$, and allow $\mu^k$ to differ from the objective $\mu^*$. We seek to understand how much the belief $\mu^k$ deviates from the objective belief $\mu^*$ by evaluating how far the ratio, 
\[ \frac{\mu^k_s/\mu^k_t}{\mu^*_s/\mu^*_t} , \] 
where $s \neq t$, differs from~1. If the ratio is larger (smaller) than one, then it means that in choice $k$, the decision maker believes the relative likelihood of state $s$ with respect to state $t$ is larger (smaller, respectively) than what he should believe, given the objective belief $\mu^*$.

Given a non-negative number $e$, we say that a dataset is $e$-belief-perturbed objective expected utility (OEU) rational, if it can be rationalized using expected utility with perturbed beliefs for which the relative likelihood ratios do not differ by more than $e$ from their objective equivalents. Formally:

\begin{definition}
Let $e \in \Re_+$. A dataset $(x^k,p^k)_{k=1}^K$ is
{\em $e$-belief-perturbed OEU rational} if there exist
$\mu^k \in \Delta_{++}(S)$ for each $k\in K$, and a concave and strictly
increasing function $u:\Re_+ \rightarrow \Re$, such that, for all $k$, 
\begin{equation*}
y\in B(p^k,p^k\cdot x^k) \implies 
\sum_{s\in S}\mu^k_su(y_s) \leq \sum_{s\in S}\mu^k_s u(x^k_s),
\end{equation*}
and for each $k \in K$ and $s,t \in S$,
\begin{equation}
\label{eq:oeubound}
\frac{1}{1+e} \le \frac{\mu^k_s/\mu^k_t}{\mu^*_s/\mu^*_t}\le 1+e.
\end{equation}
\end{definition}

When $e = 0$, $e$-belief-perturbed OEU rationality requires that $\mu^k_s = \mu^*_s$ for all $s$ and $k$, so the case of exact consistency with expected utility is obtained with a zero bound of belief perturbations. Moreover, it is easy to see that by taking $e$ to be large enough, any dataset can be $e$-belief-perturbed rationalizable. 

We should note that $e$ bounds belief perturbations for all states and observations. As such, it can be sensitive to extreme observations and outliers \citep[the CCEI is also subject to this critique: see][]{echenique2011money}. In our empirical application, we carry out a robustness analysis to account for such sensitivity (see Online Appendix~\ref{appendix:sensitivity}). 

Finally, we mention a potential relationship with models of nonexpected utility. One could think of rank-dependent utility, for example, as a way of allowing agent's beliefs to adapt to his observed choices. However, unlike $e$-belief-perturbed OEU, the nonexpected utility theory requires some consistencies on the dependency. For example, for the case of rank-dependent utility, the agent's belief over the states is affected by the ranking of the outcomes across states.

\subsection{Price Perturbation}
\label{sec:priceperturb}

We now turn to perturbed prices: think of them as prices measured with error. The perturbation is a multiplicative noise term $\ep^k_s$ to the Arrow-Debreu state price $p^k_s$. Thus, perturbed state prices are $\ep^k_sp^k_s $. Note that if $\ep^k_s=\ep^k_t$ for all $s,t$, then introducing the noise does not affect anything because it only changes the scale of prices. In other words, what matters is how perturbations affect relative prices, that is  $\ep^k_s/\ep^k_t$.  

We can measure how much the noise $\ep^k$ perturbs relative prices by evaluating how much the ratio, 
\[ \frac{\ep^k_s}{\ep^k_t} , \]
where $s \neq t$, differs from~1.

\begin{definition}\label{def:e-price-perturbed_oeu}
Let $e \in \Re_+$. A dataset $(x^k,p^k)_{k=1}^K$ is {\em $e$-price-perturbed OEU rational} if there exists a concave and strictly increasing function $u:\Re_+ \rightarrow \Re$, and $\ep^k \in \Re^S_{+}$ for each $k \in K$ such that, for all $k$, 
\begin{equation*}
y\in B(\tilde{p}^k ,\tilde{p}^k \cdot x^k) \implies \sum_{s\in S}\mu^*_s u(y_s)  \leq \sum_{s\in S}\mu^*_s u(x^k_s),
\end{equation*}
where for each $k\in K$ and $s \in S$
\begin{equation*}
\tilde{p}^k_s= p^k_s \ep^k_s 
\end{equation*}
and for each $k \in K$ and $s,t \in S$
\begin{equation}
\label{eq:oeubound:price}
\frac{1}{1+e} \le \frac{\ep^k_s}{\ep^k_t} \le 1+e .
\end{equation}
\end{definition}
It is without loss of generality to add an additional restriction that $\tilde{p}^k \cdot x^k = p^k \cdot x^k$ for each $k \in K$ because what matters are the relative prices. 

The idea is illustrated in Figure~\ref{fig:perturbation_rationalization_example}. The figure shows how the perturbations to relative prices affect budget lines, under the assumption that $|S|=2$.  For each value of $e\in \{ 0.1, 0.25, 1 \}$ and $k \in K$, the blue area represents the set 
\[ \left\{ x \in \Re_{+}^S \;\middle|\; \tilde{p} \cdot x = p^k \cdot x^k \te{ for some } \tilde{p} \in \Re_{++}^S\text{ such that } \forall s, t \in S , \frac{1}{1+e} \leq \frac{\tilde{p}_s/p_s^k}{\tilde{p}_t/p_t^k} \leq 1+e \right\} \]
of perturbed budget lines. The dataset in the figure is the same as in Figure~\ref{fig:violation}B, which is not rationalizable with any expected utility function as we discussed.

\begin{figure}[t]
\centering 
\includegraphics[width=0.95\textwidth]{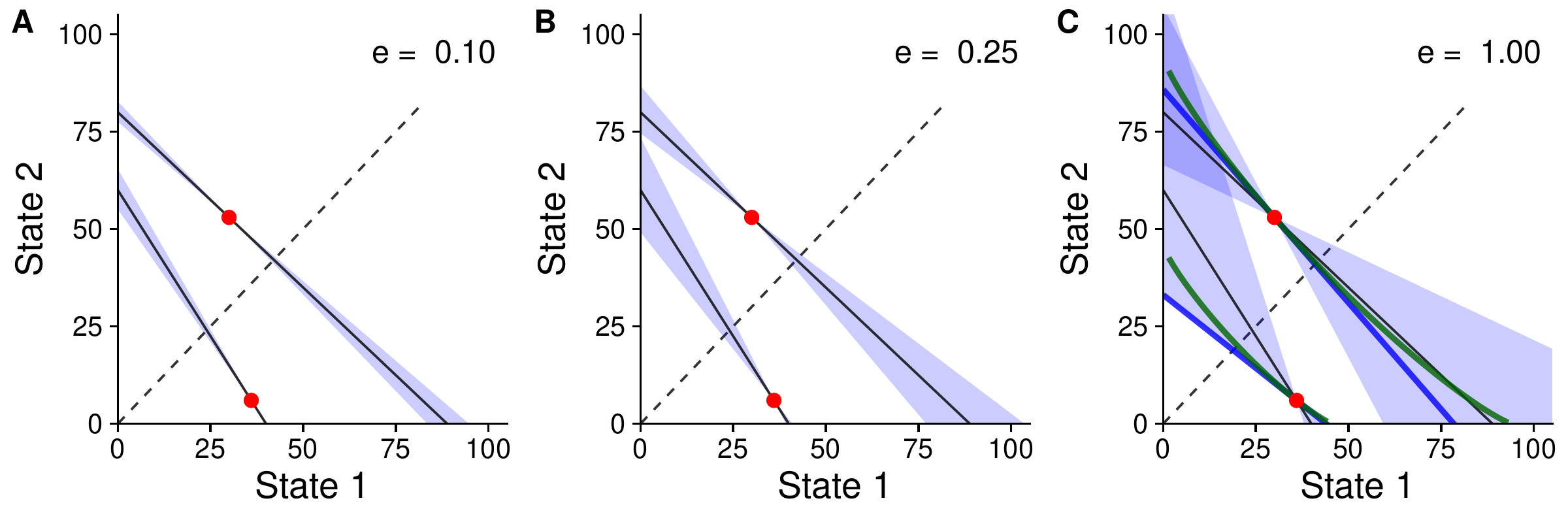}
\caption{Illustration of the set of possible perturbed budget sets with $e \in \{ 0.1, 0.25, 1 \}$. {\em Notes}: Panel~C presents an example of price-perturbed OEU rationalization. The solid blue line represents the perturbed budget set and the green line represents the indifference curve.} 
\label{fig:perturbation_rationalization_example}
\end{figure}

Figure~\ref{fig:perturbation_rationalization_example}C illustrates how we rationalize the dataset in Figure~\ref{fig:violation}B. The blue bold lines are perturbed budget lines and the green bold curves are (fixed) indifference curves passing through each of the $x^k$ in the data. 
The blue shaded areas are the sets of perturbed budget lines bounded by $e = 1$. Perturbed budget lines needed to rationalize the choices are indicated with blue bold lines. Since they are inside the shaded areas, the dataset is price-perturbed OEU rational with $e = 1$.

\subsection{Utility Perturbation}

Finally, we turn to perturbed utility. As explained above, perturbations are multiplicative and take the form $\ep^k_su(x^k_s)$. It is easy to see that this method is equivalent to belief perturbation.\footnote{We consider state-contingent perturbations. As such, perturbed utilities fall outside of the domain of EU theory. We thank Jose Apestegu\'{\i}a and Miguel Ballester for pointing this out to us.} As for price  perturbations, we seek to measure how much the $\ep^k$ perturbs  utilities at choice problem $k$ by evaluating how much the ratio, 
\[ \frac{\ep^k_s}{\ep^k_t} , \]
where $s \neq t$, differs from~1. 

\begin{definition}
Let $e\in \Re_+$. A dataset $(x^k,p^k)_{k=1}^K$ is {\em $e$-utility-perturbed OEU rational} if there exists a concave and strictly increasing function $u:\Re_+ \rightarrow \Re$ and $\ep^k \in \Re^S_{+}$ for each $k \in K$ such that, for all $k$, 
\begin{equation*}
y\in B(p^k ,p^k \cdot x^k) \implies \sum_{s\in S}\mu^*_s \ep^k_s u(y_s)  \leq \sum_{s\in S}\mu^*_s  \ep^k_su(x^k_s),
\end{equation*}
and for each $k \in K$ and $s,t \in S$
\begin{equation}
\label{eq:oeubound:utility}
\frac{1}{1+e} \le \frac{\ep^k_s}{\ep^k_t} \le 1+e.
\end{equation}
\end{definition}

\subsection{Equivalence of Belief, Price, and Utility Perturbations}

The first observation we make is that the three sources of perturbations are equivalent, in the sense that for any $e$ a dataset is $e$-perturbed rationalizable according to one of the sources if and only if it is also rationalizable according to any of the other sources with the same~$e$. By virtue of this result, we can interpret our measure of deviations from OEU in any of the ways we have introduced. 

\begin{theorem}
\label{theorem:robustoeu0}
Let $e \in \Re_+$, and $D$ be a dataset. The following are equivalent: 
\begin{itemize}
\item $D$ is $e$-belief-perturbed OEU rational; 
\item $D$ is $e$-price-perturbed OEU rational;
\item $D$ is $e$-utility-perturbed OEU rational.
\end{itemize}
\end{theorem}

The proof appears in Appendix~\ref{proof}. 
In light of Theorem~\ref{theorem:robustoeu0}, we shall simply say that a dataset is {\em $e$-perturbed OEU rational} if it is $e$-belief-perturbed OEU rational, and this will be equivalent to being $e$-price-perturbed OEU rational, and $e$-utility-perturbed OEU rational.

\subsection{Characterizations}
We proceed to give a characterization of the dataset that are $e$-perturbed OEU rational. Specifically, given $e \in \Re_+$, we propose a revealed preference axiom and prove that a dataset satisfies the axiom if and only if it is $e$-perturbed OEU rational. 

Before we state the axiom, we need to introduce some additional notation. In the current model, where $\mu^*$ is known and objective, what matters to an expected utility maximizer is not the state price itself, but instead the {\em risk-neutral} price. 

\begin{definition}
For any dataset $(p^k,x^k)_{k=1}^K$, the {\em  risk neutral price} $\rho^k_s \in \Re^S_{++}$ in choice problem $k$ at state $s$ is defined by 
\[ \rho^k_s = \frac{p^k_s}{\mu^*_s} . \]
\end{definition}

As in \cite{echenique2015savage}, the axiom we propose involves a sequence $(x^{k_i}_{s_i}, x^{k'_i}_{s'_i})_{i=1}^n$ of pairs satisfying certain conditions. 

\begin{definition} 
\label{def:testsequence}
A sequence of pairs $(x^{k_i}_{s_i}, x^{k'_i}_{s'_i})_{i=1}^n$ is called a {\em test sequence} if  
\begin{enumerate}[label=(\roman*),ref=(\roman*)]
\item\label{1it:sarseuone} $x^{k_i}_{s_i} >  x^{k'_i}_{s'_i}$ for all $i$;
\item\label{1it:sarseuthree} each $k$ appears as $k_i$ (on the left of the pair) the same  number of times it appears as $k'_i$ (on the right).
\end{enumerate}
\end{definition}

\cite{echenique2015savage} provide an axiom for OEU rationalization, termed the Strong Axiom for Revealed Objective Expected Utility (SAROEU), which states that for any test sequence $(x^{k_i}_{s_i}, x^{k'_i}_{s'_i})_{i=1}^n$, we have
\begin{equation}\label{eq:conclusion}
\prod_{i=1}^n \frac{\rho^{k_i}_{s_i}}{\rho^{k'_i}_{s'_i}} \leq 1.
\end{equation}
SAROEU is equivalent to the axiom provided by \cite{kubler2014asset}. 

It is easy to see why SAROEU is necessary for OEU rationalization. Assuming (for simplicity of exposition) that $u$ is differentiable, the first-order condition of the maximization problem~\eqref{eq:maxEU} for choice problem $k$ is 
\[ \lambda^k p^k_s= \mu^*_s u'(x^k_s), \text{ or equivalently, } \rho^k_s= \frac{u'(x^k_s)}{\lambda^k} , \] 
where $\la^k>0$ is a Lagrange multiplier. 

By substituting this equation on the left hand side of~\eqref{eq:conclusion}, we have 
\begin{equation*}
\prod_{i=1}^n \frac{\rho^{k_i}_{s_i}}{\rho^{k'_i}_{s'_i}} 
= \prod_{i=1}^n \frac{\la^{k'_i}}{\la^{k_i}} 
\cdot \prod_{i=1}^n \frac{u'(x^{k_i}_{s_i})}{u'(x^{k'_i}_{s'_i})} 
\leq 1 .
\end{equation*}
To see that this term is smaller than 1, note that the first term of the product of the $\lambda$-ratios is equal to one because of the condition~\ref{1it:sarseuthree} of the test sequence: all $\lambda^k$ must cancel out. The second term of the product of $u'$-ratio is less than one because of the concavity of $u$, and the condition~\ref{1it:sarseuone} of the test sequence (i.e., $u'(x^{k_i}_{s_i})/u'(x^{k'_i}_{s'_i})\le 1$). Thus, SAROEU is implied. It is more complicated to show that SAROEU is sufficient \citep[see][]{echenique2015savage}. 

Now, $e$-perturbed OEU rationality allows the decision maker to use different beliefs $\mu^k \in \Delta_{++}(S)$ for each choice problem $k$. Consequently, SAROEU is not necessary for $e$-perturbed OEU rationality. To see that SAROEU can be violated, note that the first-order condition of the maximization~\eqref{eq:maxEUrobust} for choice $k$ is as follows: there exists a positive number (Lagrange multiplier) $\lambda^k$ such that for each $s \in S$, 
\[ \lambda^k p^k_s = \mu^k_s u'(x^k_s), \text{ or equivalently, } \rho^k_s = \frac{\mu^k_s}{\mu^*_s} \frac{u'(x^k_s)}{\lambda^k} . \]

Suppose that $x^k_s > x^k_t$. Then $(x^k_s, x^k_t)$ is a test sequence (of length one) according to Definition~\ref{def:testsequence}. We have 
\begin{equation*}
\frac{\rho^{k}_{s}}{\rho^{k}_{t}} 
= \left( \frac{\mu^k_s}{\mu^*_s} \frac{u'(x^k_s)}{\lambda^k}\right ) \bigg/ \left( \frac{\mu^k_t}{\mu^*_t} \frac{u'(x^k_t)}{\lambda^k} \right) 
= \frac{u'(x^k_s)}{u'(x^k_t)} \frac{\mu^k_s/\mu^k_t}{\mu^*_s/\mu^*_t}.
\end{equation*}
Even though $x^k_s> x^k_t$ implies the first term of the ratio of $u'$ is less than one, the second term can be strictly larger than one. When $x^k_s$ is close enough to $x^k_t$, the first term is almost one while the second term can be strictly larger than one. Consequently, SAROEU can be violated.  

However, by~\eqref{eq:oeubound}, we know that the second term is bounded by $1+e$. So we must have 
\[ \frac{\rho^{k}_{s}}{\rho^{k}_{t}}\le 1+e. \]
In general, for a sequence $(x^{k_i}_{s_i}, x^{k'_i}_{s'_i})_{i=1}^n$ of pairs, one may suspect that the bound is calculated as $(1+e)^n$. This is not true because if $x^k_s$ appears both as $x^{k_i}_{s_i}$ for some $i$ (on the left of the pair) and as $x^{k'_j}_{s'_j}$ for some $j$ (on the right of the pair), then all $\mu^k_s$ can be canceled out. What matters is the number of times $x^k_s$ appears without being canceled out. This number can be defined as follows. 

\begin{definition}
\label{def:sk_remaining}
Consider any sequence $(x^{k_i}_{s_i}, x^{k'_i}_{s'_i})_{i=1}^n$ of pairs. Let  $(x^{k_i}_{s_i}, x^{k'_i}_{s'_i})_{i=1}^n \equiv\sigma$. For any  $k\in K$ and $s \in S$, 
\[ d(\sigma,k,s)= \# \{i \mid x^k_s= x^{k_i}_{s_i} \} - \# \{ i \mid x^k_s= x^{k'_i}_{s'_i} \} , \] 
and 
\[ m(\sigma)=\sum_{s\in S} \sum_{k\in K: d(\sigma,k,s)>0}d(\sigma,k,s). \]
\end{definition}

Note that, if $d(\sigma,k,s)$ is positive, then $d(\sigma,k,s)$ is the number of times $\mu^k_s$ appears as a numerator without being canceled out. If it is negative, then $d(\sigma,k,s)$ is the number of times $\mu^k_s$ appears as a denominator without being canceled out. So $m(\sigma)$ is the ``net'' number of terms such as $\mu^k_s/\mu^k_t$ that are present in the numerator. Thus the relevant bound is $(1+e)^{m(\sigma)}$.

Given the discussion above, it is easy to see that the following axiom is necessary for $e$-perturbed OEU rationality. 

\begin{axiom}[$e$-Perturbed Strong Axiom for Revealed Objective Expected Utility ($e$-PSAROEU)] For any test sequence of pairs 
$(x^{k_i}_{s_i}, x^{k'_i}_{s'_i})_{i=1}^n\equiv \sigma$, we have 
\[
\prod_{i=1}^n \frac{\rho^{k_i}_{s_i}}{\rho^{k'_i}_{s'_i}} \leq (1+e)^{m(\sigma)}.
\]
\end{axiom}

The main result of this section is to show that the axiom is also sufficient. 

\begin{theorem}\label{theorem:robustoeu}
Given $e\in \Re_+$, and let $D$ be a dataset. The following are equivalent:
\begin{itemize}
\item $D$ is $e$-belief-perturbed OEU rational.
\item $D$ satisfies $e$-PSAROEU.
\end{itemize}
\end{theorem}

The proof appears in Appendix~\ref{proof}. 

Axioms like $e$-PSAROEU can be interpreted as a statement about downward-sloping demand \citep[see][]{echenique2020edu}. For example, $(x^k_s,x^k_{s'})$ with $x^k_s > x^k_{s'}$ is a test sequence. If risk neutral prices satisfy $\rho^k_s > \rho^k_{s'}$, then the dataset violates downward-sloping demand. Now $e$-PSAROEU measures the extent of the violation by controlling the size of $\rho^k_s / \rho^k_{s'}$.

In its connection to downward-sloping demand, Theorem~\ref{theorem:robustoeu} formalizes the idea of testing OEU through the correlation of risk-neutral prices and quantities: see \cite{friedman2018varieties} and our discussion in Section~\ref{section:oeu_application_results}. Theorem~\ref{theorem:robustoeu} and the axiom $e$-PSAROEU give the precise form that the downward-sloping demand property takes in order to characterize OEU, and provide a non-parametric justification to the practice of analyzing the correlation of prices and quantities.

As mentioned, $0$-PSAROEU is equivalent to SAROEU. When $e=\infty$, the $e$-PSAROEU always holds because $(1+e)^{m(\sigma)}=\infty$. 

Given a dataset, we shall calculate the {\em smallest} $e$ for which the dataset satisfies $e$-PSAROEU. It is easy to see that such a minimal level of $e$ exists.\footnote{In Online Appendix~\ref{sec:compute_e}, we show that $e_*$ can be obtained as a solution of minimization of a continuous function on a compact space. Hence, the minimum exists.} 
We explain in Online Appendices~\ref{sec:compute_e} and~\ref{appendix:implementation} how it is calculated in practice.

\begin{definition}
{\em Minimal~$e$}, denoted $e_*$, is the smallest $e'\geq 0$ for which the data satisfies $e'$-PSAROEU.
\end{definition}

The number $e_*$ is a crucial component of our empirical analysis. Importantly, it is the basis of a statistical procedure for testing the null hypothesis of OEU rationality. 

As mentioned above, $e_*$ is a bound that has to hold across all observations, and therefore may be sensitive to extreme outliers. It is, however, easy to check the sensitivity of the calculated $e_*$ to an extreme observation. One can, for example, re-calculate $e_*$ after dropping one or two observations, and look for large changes. 

Finally, $e_*$ depends on the prices and the objective probability which a decision maker faces. In particular, it is clear from $e$-PSAROEU that $1 + e$ is bounded by the maximum ratio of risk-neutral prices (i.e., $\max_{k, k' \in K, s, s' \in S} \rho^{k}_{s} / \rho^{k'}_{s'}$). 

We should mention that Theorem~\ref{theorem:robustoeu} is similar in spirit to some of the results in \cite{allen2020satisficing}, who consider approximate rationalizability of quasilinear utility. They present a revealed preference characterization with a measure of error ``built in'' to the axiom, similar to ours, which they then use as an input to a statistical test. The two papers were developed independently, and since the models in question are very different, the results are unrelated.

\section{Testing Objective Expected Utility}
\label{section:oeu_application}

We apply our methodology to data from three large-scale online experiments. The experiments were implemented through representative surveys, and the task involved objective risk, not uncertainty. The data are taken from \citet[hereafter CKMS]{choi2014more}, \citet[hereafter CMW]{carvalho2016poverty}, and \citet[hereafter CS]{carvalho2017complexity}. All three experiments share a common experimental structure, the portfolio allocation task introduced by \cite{loomes1991evidence} and \cite{choi2007}. 

It is worth mentioning again that the three studies focus on CCEI as a measure of violation of basic rationality. We shall instead look at OEU, and use $e_*$ as our measure of violations of OEU. The procedure for calculating $e_*$ is explained in Online Appendices~\ref{sec:compute_e} and~\ref{appendix:implementation}.

\subsection{Datasets}
\label{section:oeu_application_data}

In the experiments, subjects were presented with a sequence of decision problems under risk in a graphical illustration of a two-dimensional budget line. They were asked to select a point $(x_1, x_2)$, an ``allocation,'' by clicking on the budget line (subjects were therefore forced to exhaust the income). The coordinates of the selected point represent an allocation of points between ``accounts''~$1$ and~$2$. They received the points allocated to one of the accounts, determined at random with an equal chance ($\mu_1^* = \mu_2^* = 0.5$). Subjects faced 25 budgets, as illustrated in Figure~\ref{fig:sample_budgets}. 

\begin{figure}[t]
\centering
\includegraphics[width=0.35\textwidth]{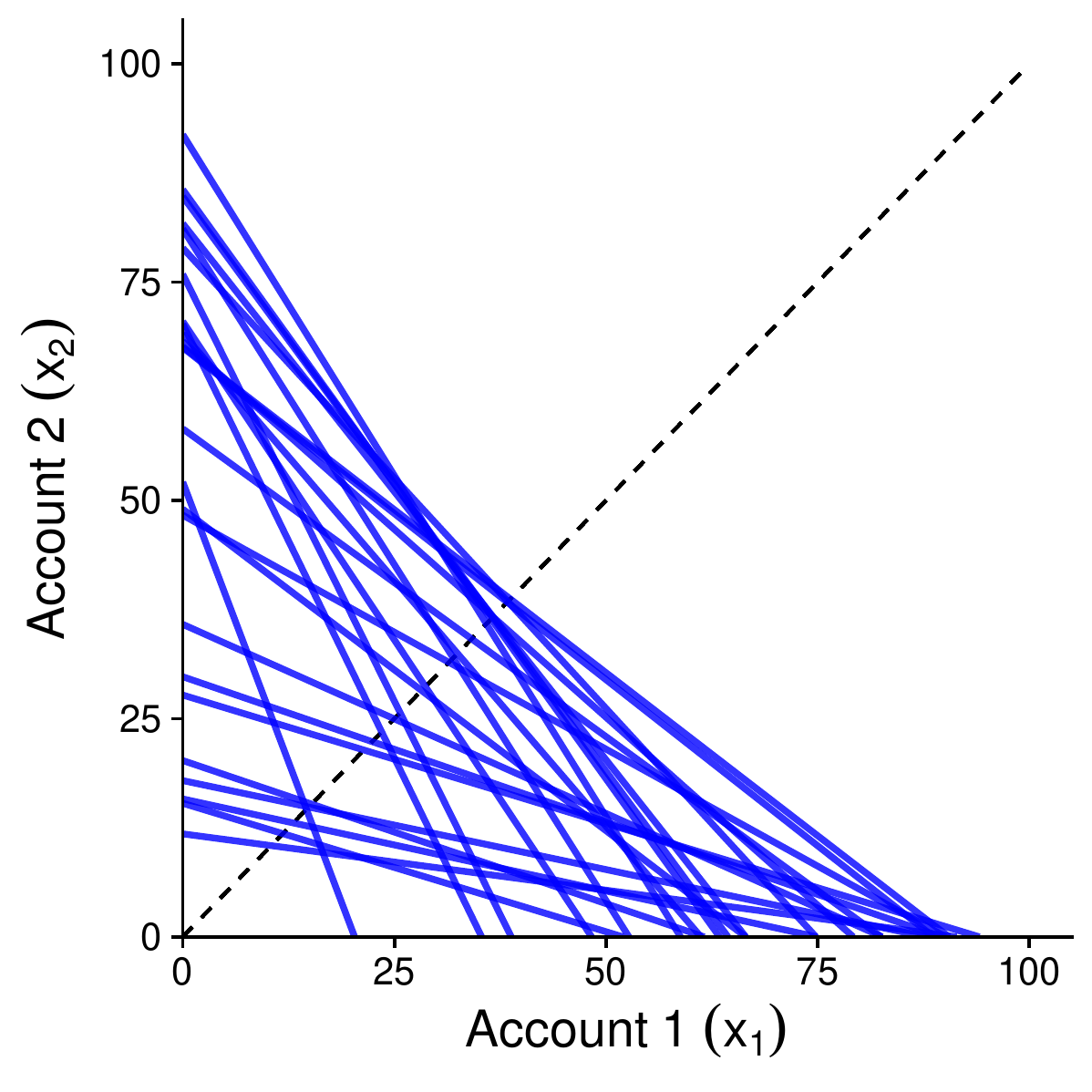}
\caption{Sample budget lines. A set of 25 budgets from one real subject in \cite{choi2014more}.} 
\label{fig:sample_budgets}
\end{figure}

We note some interpretations of the design that matter for our posterior discussion. First, points on the 45-degree line correspond to equal allocations between the two accounts and therefore involve no risk. The 45-degree line is the ``full insurance'' line. Second, we can interpret the slope of a budget line as a price in the usual sense: if the $x_2$-intercept is larger than the $x_1$-intercept, points in the account~$2$ are ``cheaper'' than those in the account~$1$. 

\cite{choi2014more} implemented the task using the instrument of the CentERpanel, randomly recruiting subjects from the entire panel sample in the Netherlands. \cite{carvalho2016poverty} administered the task using the GfK KnowledgePanel, a representative panel of the adult U.S. population. \cite{carvalho2017complexity} used the Understanding America Study panel. The number of subjects who completed the task in each study is 1,182 in CKMS, 1,119 in CMW, and 1,423 in CS. 

The survey instruments in these studies allowed them to collect a wide variety of individual demographic and economic information from the respondents. The main demographic information they obtained include gender, age, education level, household income, occupation, and household composition. 

The selection of~25 budget lines was independent across subjects in CKMS (i.e., the subjects were given different sets of budget lines), fixed in CMW (i.e., all subjects saw the same set of budgets), and semi-randomized across subjects in CS (i.e., each subject drew one of the prepared sets of~25 budgets).

\subsection{Results}
\label{section:oeu_application_results}

\paragraph{Summary statistics.} 
We exclude five subjects who are ``exactly'' OEU rational, leaving us a total of 3,719 subjects in the three experiments. 
About~76\% of subjects never chose corners of the budget lines, and there is only two percent of the entire sample who chose corners in more than half of the~25 questions. Finally, no subjects chose corners in all~25 questions. Given these observations, our focus on risk aversion does not seem to be too restrictive in these datasets. 

We calculate $e_*$ for each individual subject. The distributions of $e_*$ are displayed in Figure~\ref{fig:minimal_e_cdf}A.\footnote{Earlier drafts of the paper (posted before summer 2019) reported $\log{(1+e_*)}$, not $e_*$ itself.}$^{,}$\footnote{The empirical CDF for the CMW data has several ``steps'' since all subjects faced with the same set of~25 budget lines. For example, there are 172 subjects with $e_* = 3.5925$. The maximum adjustment required to make their data $e$-perturbed OEU rational is on the budget line, with prices $(p_1, p_2) = (1, 0.2177)$.} 
The CKMS sample has a mean $e_*$ of~3.034, and a median of~2.729. The CMW subjects have a mean of~2.487 and a median of~2.533. The CS sample has a mean of~2.494 and a median of~2.088.\footnote{Since $e_*$ depends on the design of set(s) of budgets, comparing $e_*$ across studies requires caution.}
Recall that the smaller a subject's $e_*$ is, the closer are her choices to OEU rationality. It is, however, hard to exactly interpret the magnitude of $e_*$. We turn to this issue in Section~\ref{sec:implementation_perturbation}. 

\begin{figure}[!t]
\centering 
\includegraphics[scale=0.6]{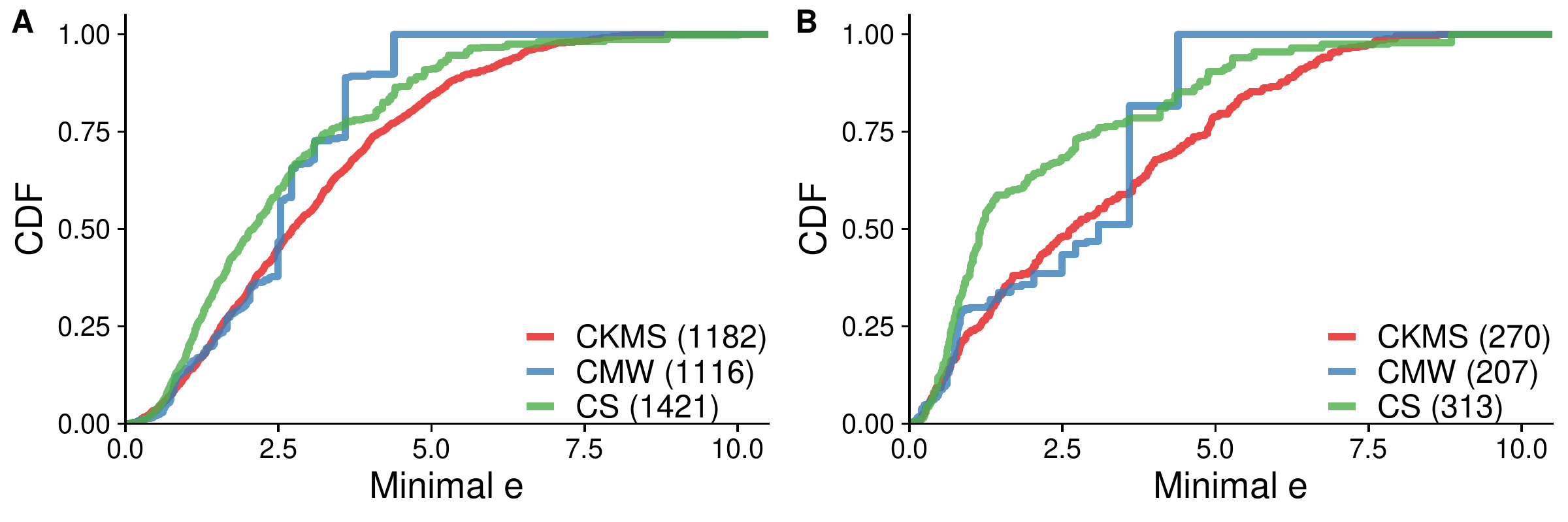}
\caption{Empirical CDFs of $e_*$. (A) All subjects. (B) The subsample of subjects with $\text{CCEI} = 1$. {\em Notes}: The number of observations in each dataset is presented in parentheses.} 
\label{fig:minimal_e_cdf}
\end{figure}

\paragraph{Downward-sloping demand and $e_*$.} 
Perturbations in beliefs, prices, or utility, seek to accommodate a dataset so that it is OEU rationalizable. The accommodation can be seen as correcting a mismatch of relative prices and marginal rates of substitution: recall our discussion in the introduction. Another way to see the accommodation is through the relation between prices and quantities. Our revealed preference axiom, $e$-PSAROEU, bounds certain deviations from downward-sloping demand. The minimal~$e$ is therefore a measure of the kinds of deviations from downward-sloping demand that are crucial to OEU rationality. 

Figure~\ref{fig:45degree_and_measures} illustrates this idea. We calculate the Spearman's correlation coefficient between $\log (x_2 / x_1)$ and $\log (p_2 / p_1)$ for each subject in the datasets.\footnote{Note that $\log (x_2 / x_1)$ is not defined at the corners. We thus adjust corner choices (less than~5\% of all choices) by a small constant,~0.1\% of the budget in each choice, in calculation of the correlation coefficient.} 
Roughly speaking, downward-sloping demand corresponds to the correlation between changes in quantities $\log (x_2 / x_1)$, and changes in prices $\log (p_2 / p_1)$, being negative. The idea is that if a subject properly responds to price changes, then as $\log (x_2 / x_1)$ becomes larger, $\log (x_2 / x_1)$ should become lower. The correlation is close to zero if subjects do not respond to price changes. 

The top panels of Figure~\ref{fig:45degree_and_measures} confirms that $e_*$ and the correlation between prices and quantities are closely related. This means that subjects with smaller $e_*$ tend to exhibit downward-sloping demand, while those with larger $e_*$ are insensitive to price changes. Across all three datasets, $e_*$ and downward-sloping demand are strongly and positively related.

\begin{figure}[t]
\centering 
\includegraphics[width=\textwidth]{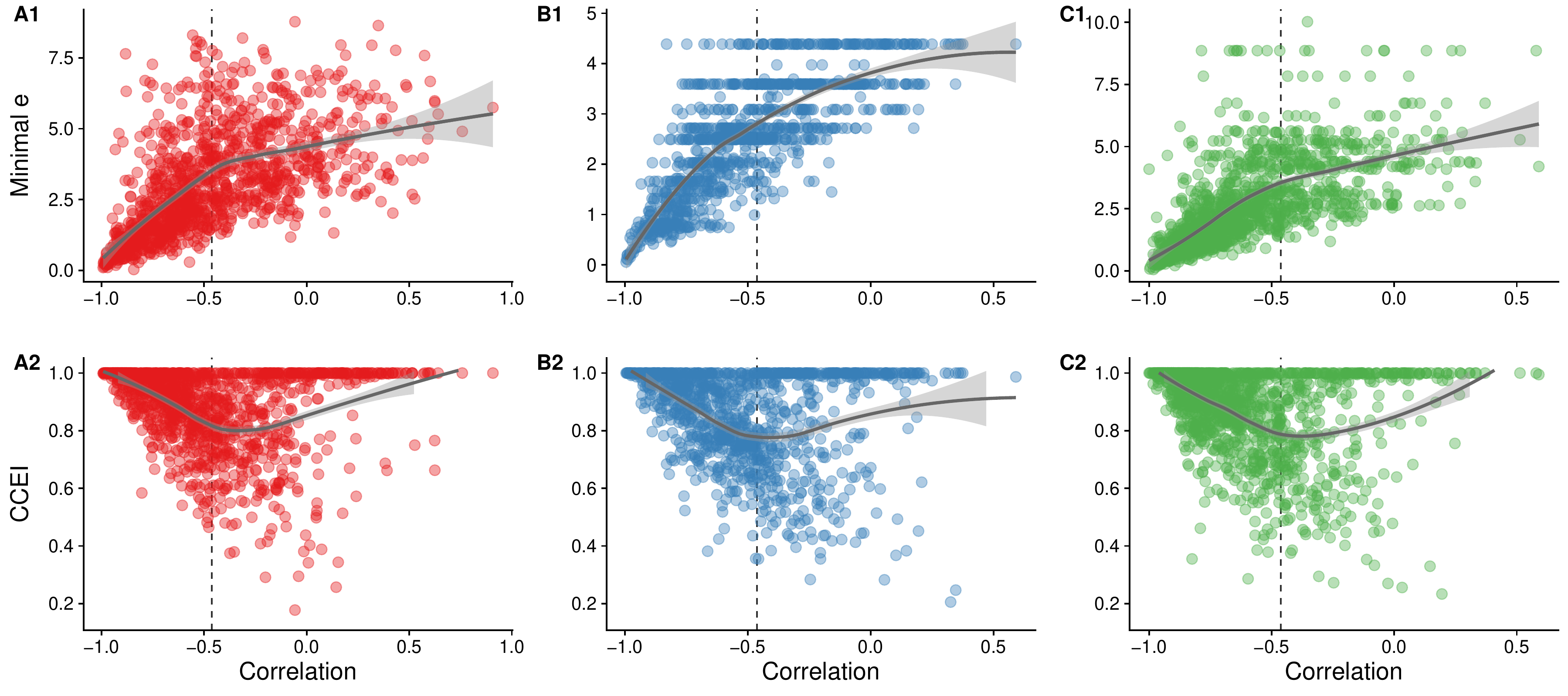}
\caption{Correlation between $\log (x_2 / x_1)$ and $\log (p_2 / p_1)$ and measures of rationality. Panels: (A) CKMS, (B) CMW, (C) CS. {\em Notes}: The vertical dashed line indicates the threshold below which Spearman's correlation is significantly negative (one-sided, at the~1\% level). Black curves represent LOESS smoothing with~95\% confidence bands.}
\label{fig:45degree_and_measures}
\end{figure}

The CCEI, on the other hand, is not clearly related to downward-sloping demand. As illustrated in the bottom panels of Figure~\ref{fig:45degree_and_measures}, the relation between CCEI and the correlation between prices and quantities is not monotonic. Agents who are closer to complying with utility maximization do not necessarily display a stronger negtive correlation between prices and quantities. The finding is consistent with our comment about CCEI, $e_*$, and OEU rationality: CCEI measures the distance from utility maximization, which is related to parallel shifts in budget lines, while $e_*$ and OEU are about the slope of the budget lines, and about a negative relation between quantities and prices. 

We should mention the practice by some authors, notably, \cite{friedman2018varieties}, to evaluate compliance with OEU by looking at the correlation between risk-neutral prices and quantities. Our $e_*$ is related to that idea, and the empirical results presented in this section can be read as a validation of the correlational approach. \cite{friedman2018varieties} use their approach to estimate a parametric functional form, using experimental data in which they vary objective probabilities, not just prices. 
Our approach is non-parametric, and focused on testing OEU itself, not estimating any particular utility specification.

\paragraph{First-order stochastic dominance and $e_*$.}
In the experiments we consider, choosing $(x_1, x_2)$ at prices $(p_1, p_2)$ violates {\em monotonicity with respect to first-order stochastic dominance} (hereafter {\em FOSD-monotonicity}) when either (i) $p_1 > p_2$ and $x_1 > x_2$ or (ii) $p_2 > p_1$ and $x_2 > x_1$. Since the two states have the same objective probability in our datasets, choosing a greater payoff in the more expensive state violates FOSD-monotonicity. Violations of FOSD-monotonicity are related to downward-sloping demand, as they involve consuming more in the more expensive state. 
Choices that violate FOSD-monotonicity are not uncommon in the data (see Online Appendix~\ref{appendix:fosd}).

Since OEU-rational choices must satisfy FOSD-monotonicity, $e_*=0$ implies no violations of FOSD-monotonicity. Moreover, the value of $e_*$ is a good indicator of FOSD-monotonicity violations. 
See the positive relationship between the fraction of FOSD-monotonicity violations and $e_*$ in the top row of Figure~\ref{fig:e-ccei_fosd}: subjects who frequently made choices violating FOSD-monotonicity tend to have larger $e_*$ compared to those with fewer such violations. 

\begin{figure}[t]
\centering
\includegraphics[width=\textwidth]{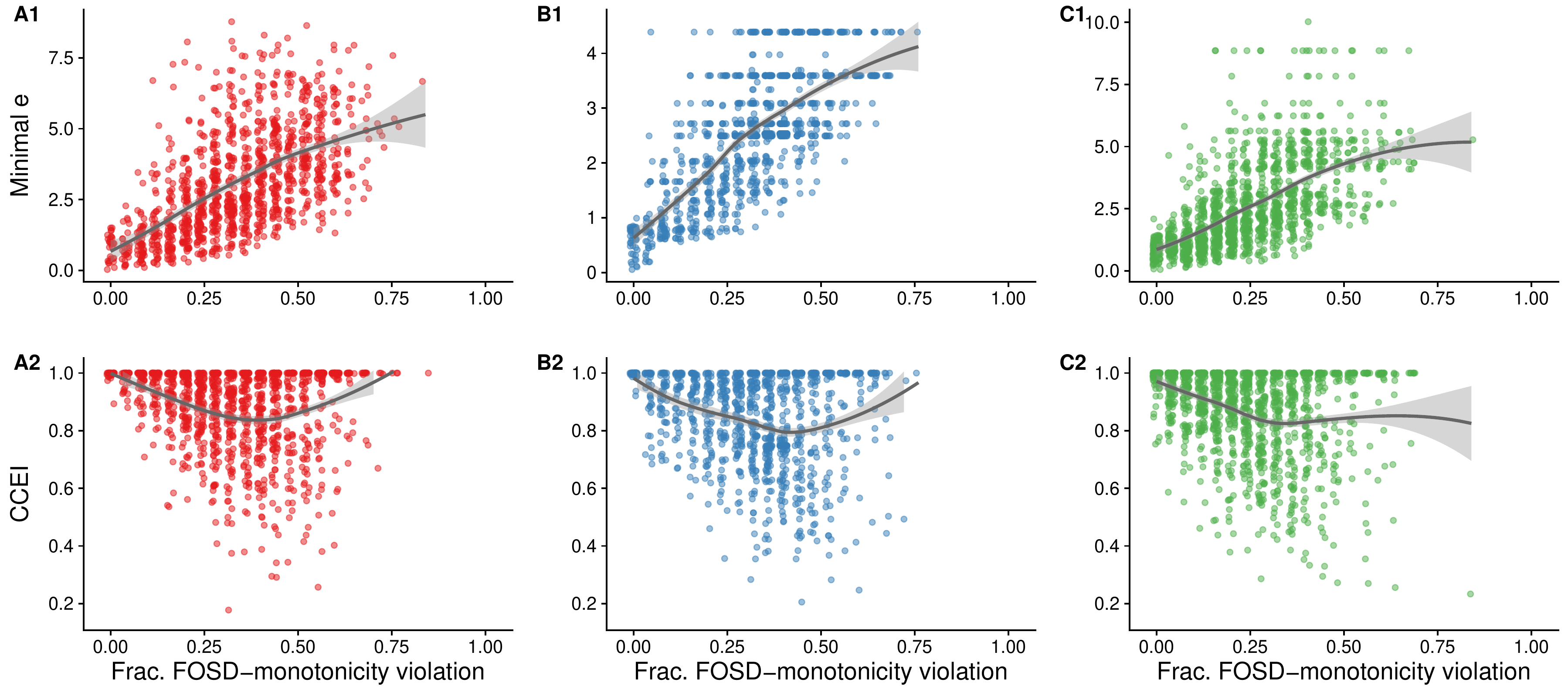}
\caption{Violation of FOSD-monotonicity and measures of rationality. Black curves represent LOESS smoothing with~95\% confidence bands. Panels: (A) CKMS, (B) CMW, (C) CS.}
\label{fig:e-ccei_fosd}
\end{figure}

The relation between $e_*$ and violations of FOSD-monotonicity stands in sharp contrast with CCEI. First, choices that violate FOSD-monotonicity {\em can} be consistent with GARP. Our data exhibits subjects that pass GARP while making choices that violate FOSD-monotonicity \citep[an empirical fact that was first pointed out by][]{choi2014more}. The bottom panels of Figure~\ref{fig:e-ccei_fosd} show that a substantial number of subjects with perfect compliance with GARP ($\text{CCEI} = 1$) make at least one violation of FOSD-monotonicity. 
The existence of these subjects generates a nonmonotonic relationship between CCEI and the frequency of violation of FOSD-monotonicity.

\paragraph{Typical patterns of choices.}
We can gain some insights into the data by considering ``typical'' patterns of choice. Figure~\ref{fig:choice_pattern_ccei_1} presents choice patterns from selected subjects with $\text{CCEI} = 1$ and varying degrees of~$e_*$.\footnote{The patterns in Figure~\ref{fig:choice_pattern_ccei_1} are not an exhaustive list by any means. See Online Appendix~\ref{appendix:more_choice_pattern} for more examples.} 
Panels~A-F plot observed choices and panels~a-f plot the relationship between $\log (x_2 / x_1)$ and $\log (p_2 / p_1)$ associated with each choice pattern. As discussed above, panels~a-f should exhibit a negative relationship (downward-sloping demand) for the subject to be OEU rational.

\begin{figure}[!tp]
\centering 
\includegraphics[width=0.85\textwidth]{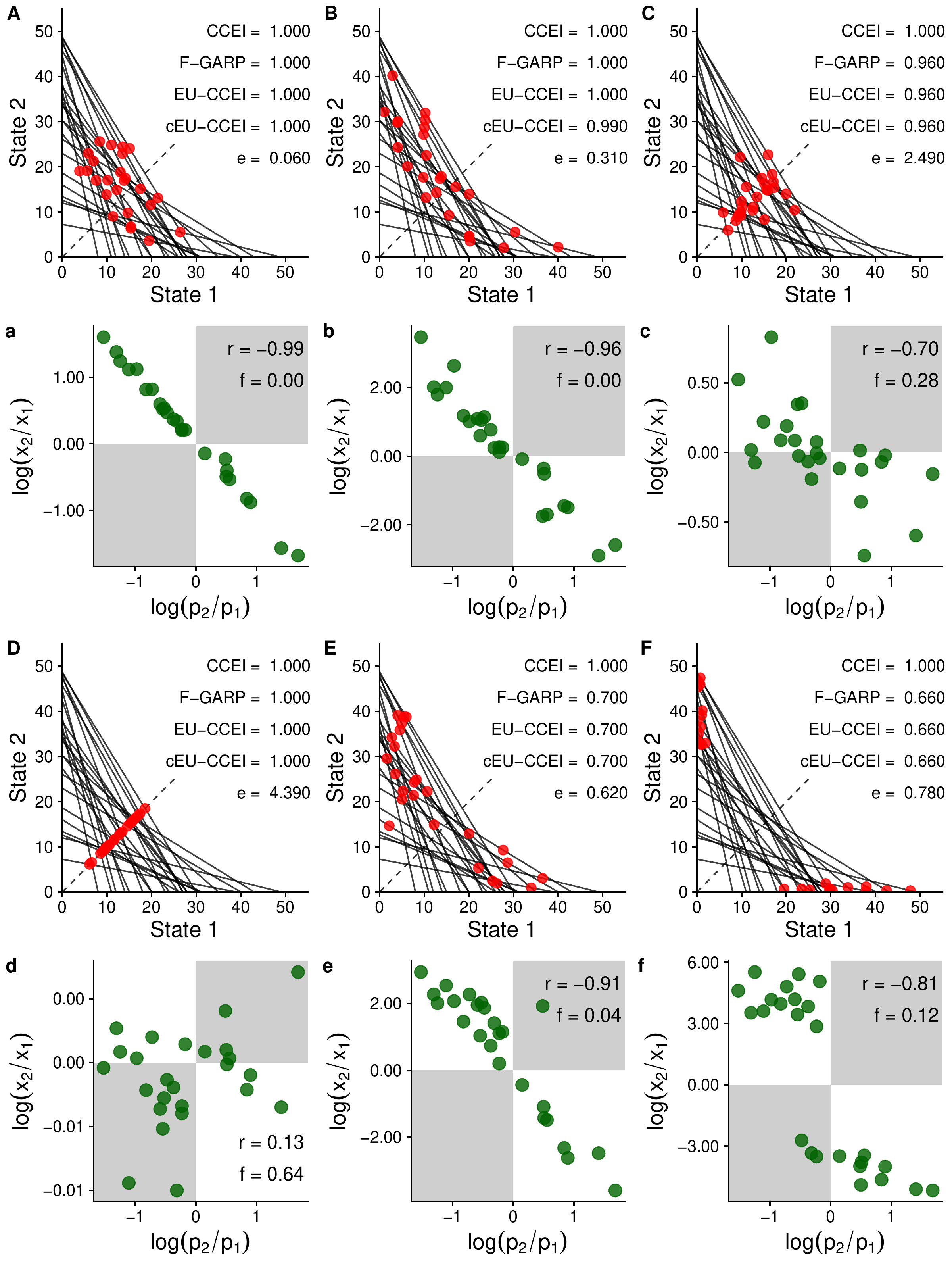}
\caption{Choice patterns from six subjects in the CMW data with $\text{CCEI} = 1$ and varying $e_*$. (A-F) Observed choices. (a-f) The relation between $\log (x_2 / x_1)$ and $\log (p_2 / p_1)$. {\em Notes}: Choices in shaded areas violate FOSD-monotonicity. $r$ indicates the Spearman's correlation coefficient and $f$ indicates the fraction of choices violating FOSD-monotonicity. In this data, median CCEI is 0.889, median EU-CCEI is 0.730, and median $e_*$ is 2.533. 
F-GARP, EU-CCEI, cEU-CCEI are calculated with the GRID method of \cite{polissonquahrenou17}.} 
\label{fig:choice_pattern_ccei_1}
\end{figure}

Panel~A presents a choice pattern that is ``almost'' consistent with OEU. The relation between $\log (x_2 / x_1)$ and $\log (p_2 / p_1)$ fits close to a line with negative slope, but there is a small deviation around $\log (p_2 / p_1) = -1$ which makes the subject's $e_*$ nonzero. 
Panel~B also shows a pattern that does not involve any FOSD-monotonicity violations but is not OEU rational due to small deviations from the downward-sloping demand (see panel~b). 
The pattern in panel~C exhibits larger deviations from the downward-sloping demand (panel~c), which push its $e_*$ higher than the previous two subjects. 

The subject's choices in panel~D are close to the 45-degree line. At first glance, such choices might seem to be rationalizable by a very risk-averse expected utility function. However, as panel~d shows, the subject's choices deviate from the downward-sloping demand property, and hence cannot be rationalized by any risk-averse expected utility function. Note that the ``size'' of the deviation from the downward-sloping demand is small (see the scale of the $y$-axis in panel~d). One might be able to rationalize the choices made in panel~D with some models of errors in choices, but not with the types of errors captured by our model.\footnote{This is, in our opinion, a strength of our approach. We do not ex-post seek to invent a model of errors that might rescue EU. Instead we have written down what we think are natural sources of errors and perturbation (random utility, beliefs, and measurement errors). Our results deal with what can be rationalized when these sources of errors, and only those, are used to explain the data. A general enough model of errors will, of course, render the theory untestable.} 
We will discuss other two subjects (panels~E and~F) below. 

Figure~\ref{fig:choice_pattern_ccei_1} also illustrates how $e_*$ operates in practice when there are two states. Under the price-perturbation interpretation, it measures how big of an adjustment of prices would be needed to satisfy downward-sloping demand. Such adjustments will be represented as ``horizontal shifts'' of points in the bottom panels of the figure (since we fix the chosen bundle and rotate the budget line), and the largest adjustment corresponds to $e_*$. 
A scatterplot of $\log(x_2/x_1)$ versus $\log (p_2/p_1)$, as in panels~a-f of Figure~\ref{fig:choice_pattern_ccei_1}, works as a graphical tool to get a sense of whether a subject's $e_*$ is big or small. Online Appendix~\ref{appendix:example_e-price-perturbation} discusses this idea, and illustrates $e$-price-perturbed OEU rationalization using the choice data presented in Figure~\ref{fig:choice_pattern_ccei_1}.

\paragraph{Relationship between $e_*$, CCEI, and EU-CCEI.} 
CCEI serves a different purpose than $e_*$; it is meant to capture deviations from general utility maximization, and not OEU. Nevertheless, it is informative to understand the relationship between these measures in the data. We also comment on the recent proposal by \cite{polissonquahrenou17} of an adaptation of CCEI to test for OEU.

We observe, in Figure~\ref{fig:minimal_e_cdf}, that the distribution of $e_*$ among subjects whose CCEI is equal to one (panel~B) varies as much as in thw whole population (panel~A). Many subjects have CCEI equal to one, but their $e_*$'s can be far from zero. This means that consistency with general utility maximization is not necessarily a good indication of consistency with OEU. 

\begin{figure}[!t]
\centering 
\includegraphics[width=\textwidth]{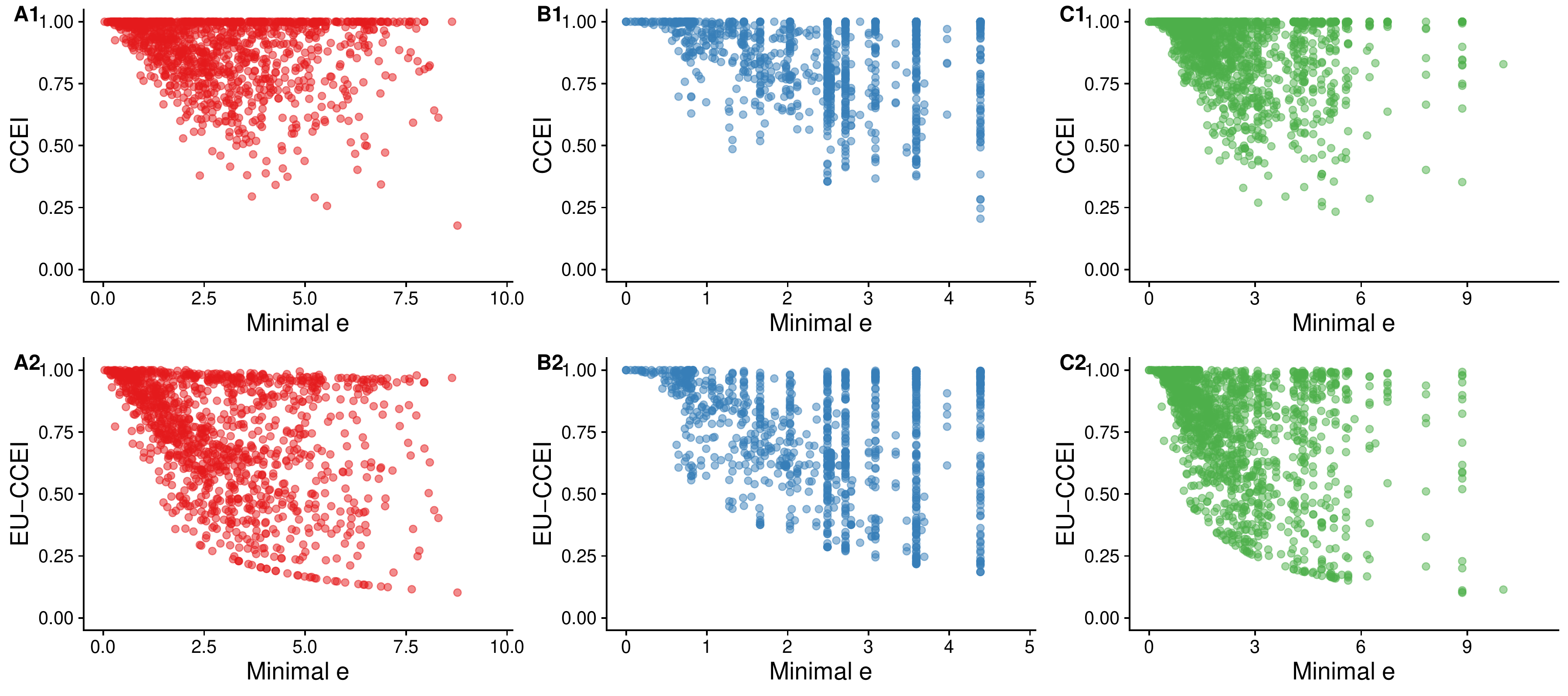}
\caption{Correlation between $e_*$ and CCEI (top panels) and EU-CCEI of \cite{polissonquahrenou17} (bottom panels). Panels: (A) CKMS, (B) CMW, (C) CS.} 
\label{fig:minimal_e_vs_ccei}
\end{figure}

That said, the measures are clearly correlated. 
Figure~\ref{fig:minimal_e_vs_ccei}, top panels, plot the relation between CCEI and $e_*$. As we expect from their definitions ({\em larger} CCEI and {\em smaller} $e_*$ correspond to higher consistency), there is a negative and significant relation between them (Spearman's correlation coefficient: $r = -0.18$ for CKMS, $r = -0.11$ for CMW, $r = -0.35$ for CS, all $p < 0.001$). Of course, subjects that are consistent with OEU as measured by $e_*$ (they have $e_*=0$) must exhibit $\text{CCEI}=1$.

Notice that the variability of the CCEI widens as $e_*$ becomes larger. Obviously, subjects with a small $e_*$ are close to being consistent with general utility maximization, and therefore have a CCEI that is close to one. However, subjects with large $e_*$ seem to have dispersed values of CCEI. 

\cite{polissonquahrenou17} propose a version of CCEI meant to measure departures from EU using their GRID method. We term this measure EU-CCEI. In contrast with our measure $e_*$, which assumes risk aversion and is based on rotating budget lines, EU-CCEI does not impose risk aversion and uses the same idea of shrinking budget lines as in standard CCEI. 
The bottom panels of Figure~\ref{fig:minimal_e_vs_ccei} exhibit the relationship between $e_*$ and EU-CCEI. It is clear that the relation between  $e_*$ and EU-CCEI is similar to that between $e_*$ and CCEI. The two measures are strongly correlated, but they also provide different conclusions for many subjects. 

There are many subjects that EU-CCEI deems consistent with OEU, but have high levels of $e_*$. This could be attributed to the more restrictive theory being tested by $e_*$. Subjects with EU-CCEI close to one and large $e_*$ could simply be non-risk-averse OEU maximizers. Perhaps more puzzling is the existence of subjects that $e_*$ sees as close to OEU while EU-CCEI does not: subjects with small values of both $e_*$ and EU-CCEI. 

It is hard to investigate the differences between EU-CCEI and $e_*$ methodologically. EU-CCEI does not specify a source of deviations from OEU, so we cannot say that one measure emphasizes one source of errors and the other a different source. Instead, we look at some of the patterns in the data that gives rise to differences. 
An example of a choice pattern in which $e_*$ and EU-CCEI differ is provided by Figure~\ref{fig:choice_pattern_ccei_1}, panel~D. The subject in question exhibits $\text{CCEI} = \text{EU-CCEI} = 1$, while $e_*$ is large and indicates a violation of OEU. (The pattern involves choices close to the 45-degree line, but with a clear violation of downward sloping demand, see panel~d.) 
Panels~E and~F exhibit subjects that $e_*$ says are close to (risk-averse) OEU, but EU-CCEI deems far from OEU. We see in panels~e and~f that the conclusion using $e_*$ can be understood by the subjects' compliance with downward sloping demand. The subjects in panels~E and~F make a few FOSD-monotonicity violations, which might explain the behavior of EU-CCEI, but that cannot be the end of the story because the subject in panel~D makes substantial FOSD-monotonicity violations and exhibits the opposite behavior of $e_*$ and EU-CCEI. Finally, we should say that there are many other patterns for which the conclusions of $e_*$ and EU-CCEI differ: see Online Appendices~\ref{appendix:diagonal} and~\ref{appendix:more_choice_pattern} for additional examples.

In Online Appendix~\ref{appendix:comparing_measures}, we examine the relationship between $e_*$ and modified CCEI indices for two additional models considered in \cite{polissonquahrenou17}: stochastically monotone utility maximization and risk-averse EU. We call these indices F-GARP and cEU-CCEI, respectively. Their values are reported for the patterns in Figure~\ref{fig:choice_pattern_ccei_1}; see Figures~\ref{fig:compare_measures_ckms}-\ref{fig:compare_measures_cs} in the Online Appendix for pairwise scatter plots of five indices (CCEI, F-GARP, EU-CCEI, cEU-CCEI, and $e_*$). The modified CCEI measures provide a more refined index for consistency for EU than CCEI, but differences with $e_*$ persist. In fact, the basic conclusions outlined in the comparison between $e_*$ and EU-CCEI hold true for these indices.

\paragraph{Correlation with demographic characteristics.} 
We investigate the correlation between our measure of consistency with OEU, $e_*$, and various demographic variables available in the data. The exercise is analogous to findings in \cite{choi2014more} that use CCEI. 

We find that younger subjects, those who have high cognitive abilities, and those who are working, are closer to being consistent with OEU than older, low ability, or non-working, subjects. For some of the three experiments we also find that highly educated, high-income subjects, and males, are closer to OEU. Figure~\ref{fig:minimal_e_demographics} summarizes the mean $e_*$ (along with the standard error of mean) across several socioeconomic categories. 
We use the same categorization as in \cite{choi2014more} to compare our results with their Figure~3. 

\begin{figure}[t]
\centering 
\includegraphics[width=\textwidth]{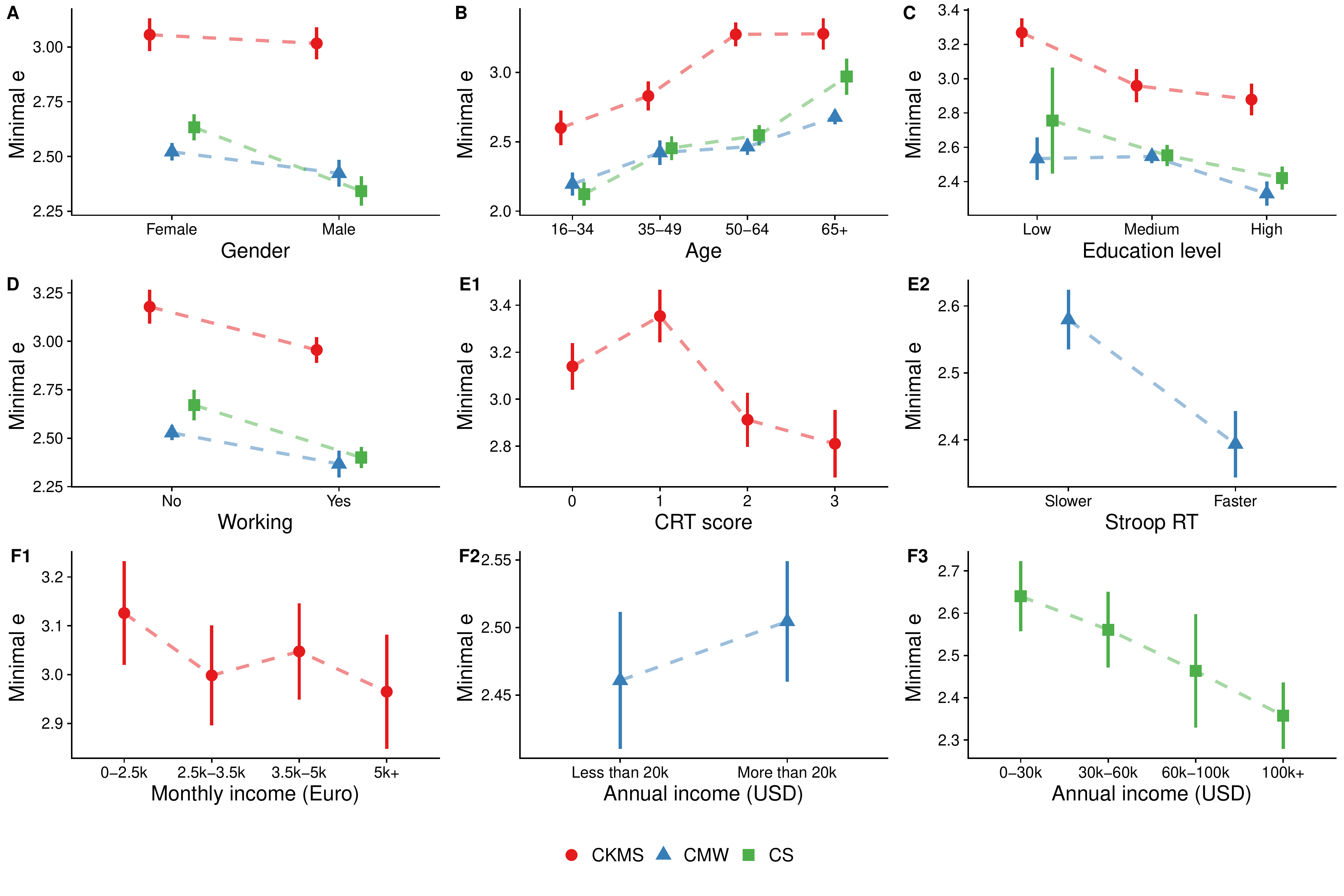}
\caption{Correlation between $e_*$ and demographic variables. {\em Notes}: Bars represent standard errors of means.}
\label{fig:minimal_e_demographics}
\end{figure}

We observe statistically significant (at the~5\% level) gender differences in CS (Welch's $t = -3.21$, $\mathit{df} = 1381.7$, $p = 0.001$) but not in CKMS (Welch's $t = -0.37$, $\mathit{df} = 1162.8$, $p = 0.708$) and CMW (Welch's $t = -1.35$, $\mathit{df} = 715.5$, $p = 0.178$). 
Male subjects were on average closer to OEU rationality than female subjects in the CS sample (panel~A). 

We find significant effects of age in all three datasets. Panel~B shows that younger subjects are on average closer to OEU rationality than older subjects (the comparison between age groups~16-34 and~65+ reveals a statistically significant difference in all three datasets; all Welch's $t$-tests give $p < 0.001$). 

We observe weak effects of education on $e_*$ (panel~C).\footnote{The low, medium, and high education levels correspond to primary or prevocational secondary education, pre-university secondary education or senior vocational training, and vocational college or university education, respectively.} 
Subjects with higher education are on average closer to OEU than those with lower education in CKMS (Welch's $t = 3.11$, $\mathit{df} = 826.9$, $p = 0.002$), but the difference is not significant in the CMW and CS (Welch's $t = 1.43$, $\mathit{df} = 121.6$, $p = 0.155$ in CMW; Welch's $t = 1.06$, $\mathit{df} = 47.2$, $p = 0.295$ in CS). 

Panel~D shows that subjects who were working at the time of the survey are on average closer to OEU than those who were not (Wlech's $t = 2.03$, $\mathit{df} = 865.1$, $p = 0.043$ in CKMS; Welch's $t = 2.04$, $\mathit{df} = 469.8$, $p = 0.042$ in CMW; Welch's $t = 2.82$, $\mathit{df} = 972.0$, $p = 0.005$ in CS). 

In panels~E1 and~E2, we classify subjects according to their Cognitive Reflection Test score \citep[CRT;][]{frederick2005} or average reaction times in the numerical Stroop task.\footnote{CRT consists of three questions, all of which have an intuitive and spontaneous, but incorrect, answers, and a deliberative and correct answer. 
In the numerical Stroop task, subjects are presented with a number, such as 888, and are asked to identify the number of times the digit is repeated (in this example the answer is~``3'', while an ``intuitive'' response is~``8''). It has been shown that response times in this task capture the subject's cognitive control ability.} 
The average $e_*$ for those who correctly answered two questions or more of the CRT is lower than the average for those who answered at most one question (Welch's $t = -3.16$, $\mathit{df} = 929.4$, $p = 0.002$). Subjects with lower response times in the numerical Stroop task have significantly lower $e_*$ (Welch's $t = -2.78$, $\mathit{df} = 1101.8$, $p = 0.005$).

One of the key findings in \cite{choi2014more} is that consistency with utility maximization as measured by CCEI correlates with household wealth. 
When we look at the relation between $e_*$ and household income, there is a negative trend but the differences across income brackets are not statistically significant (bracket ``0-2.5k'' vs. ``5k+'', Welch's $t = 1.02$, $\mathit{df} = 527.5$, $p = 0.309$; panel~F1). Panel~F2 presents a similar result between subjects who earned more than 20~thousand USD annually or not in the CMW sample (Welch's $t = 0.64697$, $\mathit{df} = 1011.3$, $p = 0.518$). When we compare poor households (annual income less than 20~thousand USD) and wealthy households (annual income more than 100~thousand USD) from the CS sample, average $e_*$ is significantly smaller for the latter sample (Welch's $t = 2.468$, $\mathit{df} = 852.7$, $p = 0.014$; panel~F3). 

\paragraph{Robustness of the results.} 
The measure $e_*$ is a bound that has to hold across all observations and states (see conditions~\eqref{eq:oeubound},~\eqref{eq:oeubound:price}, and~\eqref{eq:oeubound:utility} in the definitions of $e$-perturbed OEU). One may wonder how sensitive $e_*$ is to a small number of ``bad'' choices. Online Appendix~\ref{appendix:sensitivity} presents two robustness checks. 
In the first robustness check, we recalculate $e_*$ using subsets of observed choices after dropping one or two ``critical mistakes''. More precisely, for each subject, we calculate $e_*$ for all combinations of $25 - m$ ($m = 1, 2$) choices and pick the smallest $e_*$ among them. 
In the second robustness check, we calculate the ``average'' perturbation necessary to rationalize the data to mitigate the effect of extreme mistakes. 
These alternative ways of calculating $e_*$ do not change the general pattern of correlation between $e_*$ and CCEI or $e_*$ and demographic variables. The main empirical results are robust to the presence of a small number of bad choices.

\subsection{Minimum Perturbation Test} 
\label{sec:implementation_perturbation} 

Our discussion so far has sidestepped one issue: How are we to interpret the absolute magnitude of $e_*$? {\em When can we say that $e_*$ is large enough to ``reject'' consistency with OEU rationality?}  To answer this question, we present a statistical test of the hypothesis that an agent is OEU rational. The test needs some assumptions, but it gives us a threshold level (a critical value) for $e_*$. Any value of $e_*$ that exceeds the threshold indicates inconsistency with OEU at some given statistical significance level. 

Our approach follows the methodology laid out in \cite{echenique2011money} and \cite{echenique2016edu}. First, we adopt the price perturbation interpretation of $e$ in Section~\ref{sec:priceperturb}, that is we consider an agent who may misperceive prices. The advantage of doing so is that we can use the observed variability in price to get a handle on the assumptions we need to make on perturbed prices. 
To this end, let $D_{\text{true}} = (p^k, x^k)_{k=1}^K$ denote a dataset and $D_{\text{pert}} = (\tilde{p}^k, x^k)_{k=1}^K$ denote an ``perturbed'' dataset, where $\tilde{p}_s^k = p_s^k \varepsilon_s^k$ and $\varepsilon_s^k > 0$ for all $s \in S$ and $k \in K$. Prices $\tilde{p}^k$ are prices $p^k$ measured with error, or misperceived.

If the {\em variance} of $\epsilon$ is large, it will be easy to accommodate a dataset as OEU rational. The larger is the variance of $\epsilon$, the larger the magnitudes of $e$ that can rationalize a dataset as consistent with OEU.  In other words, we can attribute the agent's large $e$ as his misperception of prices rather than his violation of EU rationality. Our procedure is thus sensitive to the assumptions we make about the variance of $\epsilon$.

To get a handle on the variance of $\epsilon$, our approach is to assume that an agent mistakes true prices $p$ with perturbed prices $\tilde{p}$. The distributions of $p$ and $\tilde{p}$ should be similar enough that the agent might plausibly confuse the two. To make this operational, we imagine an agent who conducts a statistical test for the variance of prices. If the true variance of $p$ is $\sigma^2_0$ and the implied variance of $\tilde p$ is $\sigma^2_1 > \sigma^2_0$, then the agent would conduct a test for the null of $\sigma^2 = \sigma^2_0$ against the alternative of $\sigma^2 = \sigma^2_1$. We want the variances to be close enough that the agent might reasonably get inconclusive results from such a test (i.e., the agent may reasonably mistake true prices $p$ with perturbed prices $\tilde p$, as we assumed). {\em Specifically, we assume the sum of probabilities of type~I and type~II errors in this test is relatively large}.\footnote{The problem of variance is pervasive in statistical implementations of revealed preference tests, see \cite{varian1990goodness}, \cite{echenique2011money}, and \cite{echenique2016edu} for example. The use of the sum of type~I and type~II errors to calibrate a variance, is new to the present paper.} 
The details of how we design the test are presented in Online Appendix~\ref{appendix:minimum_perturbation_test}. 

\begin{figure}[!t]
\centering 
\includegraphics[width=0.9\textwidth]{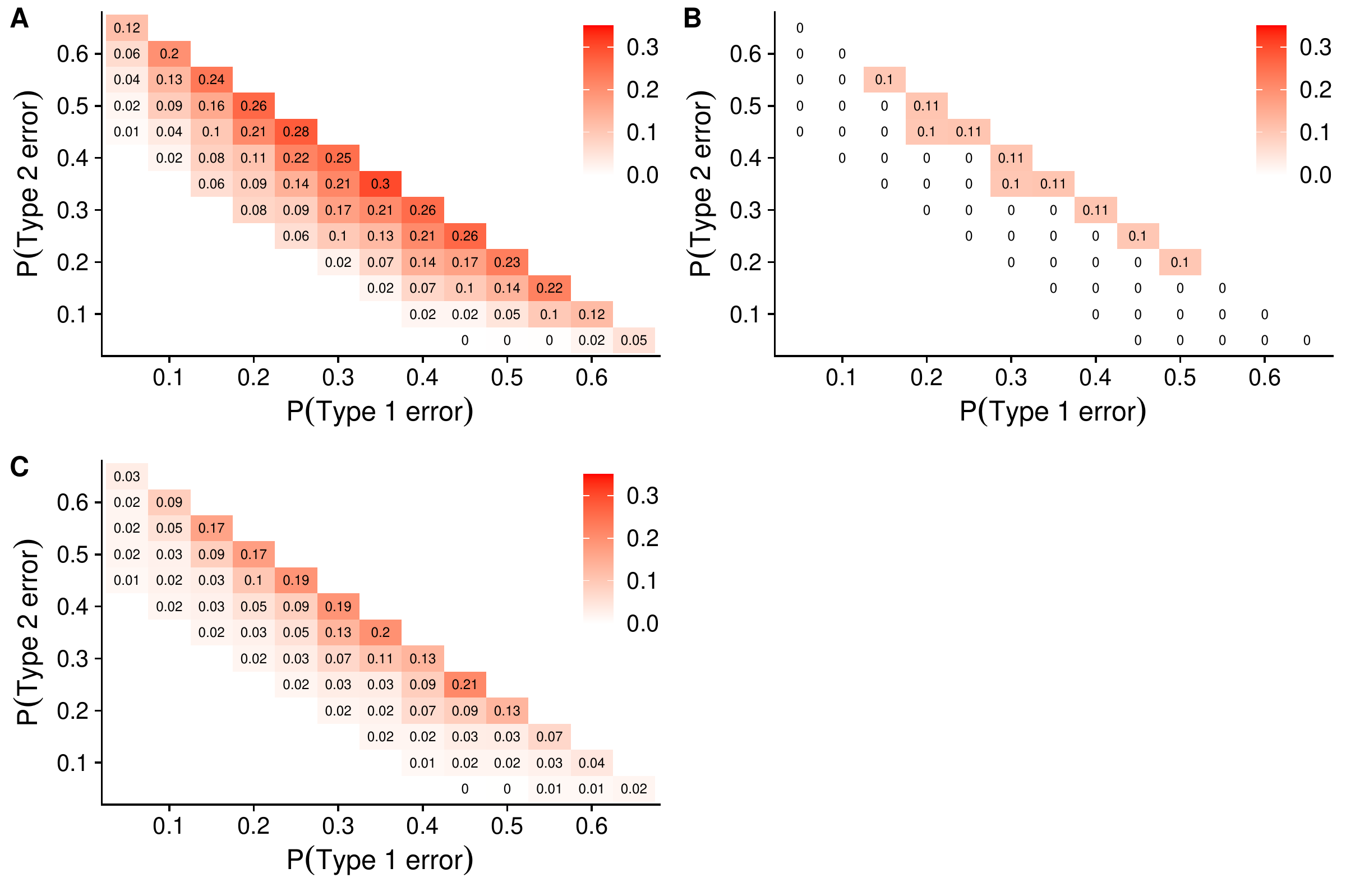}
\caption{Rejection rates under each combination of type~I and type~II error probabilities $(\eta^{I}, \eta^{\mathit{II}})$. Panels: (A) CKMS, (B) CMW, (C) CS.} 
\label{fig:rejection_merged}
\end{figure}

The main results are summarized in Figure~\ref{fig:rejection_merged}. The probability of a type~I error is $\eta^I$ and the probability of a type~II error is $\eta^{\mathit{II}}$. Recall that we focus on situations where $\eta^I+\eta^{\mathit{II}}$ is relatively large, as we want our consumer to plausibly mistake the distributions of $p$ and $\tilde p$. 
Consider, for example, our results for CKMS. The outermost numbers assume that $\eta^I + \eta^{\mathit{II}} = 0.7$. For such numbers, the rejection rates range from $5\%$ to $30\%$. This means that if prices $p$ and $\tilde{p}$ are close enough so that the agent may misperceive the prices and make type~I and type~II errors with probability $70\%$, then we can reject the hypothesis that the agent is an OEU maximizer at most $30\%$ of the cases. 

Overall, it is fair to say that rejection rates of the hypothesis that the decision maker is an OEU miximizer are modest. Notice also that smaller values of $\eta^I + \eta^{\mathit{II}}$ corresponds to smaller rejection rates. This is because when values of $\eta^I + \eta^{\mathit{II}}$ are smaller (i.e., the decision maker does not misperceive prices much), the difference between $p$ and $\tilde{p}$ should be large, which corresponds to larger variances of $\epsilon$. Larger variance, in turn, leads to smaller rejection rates. The figure also illustrates that the conclusions of the test are very sensitive to what one assumes about variances, through the assumptions about $\eta^I$ and $\eta^{\mathit{II}}$. But if we look at the largest rejection rates, for the largest values of $\eta^I + \eta^{\mathit{II}}$, we get 30\% for CKMS, 11\% for CMW, and 21\% for CS. Hence, while many subjects in the experiments are inconsistent with OEU, for most of these subjects, our statistical tests would attribute such inconsistency to misperception of prices and do not reject that the subjects are OEU maximizers.

\section{Perturbed Subjective Expected Utility}
\label{section:robustseu}

We now turn to the model of subjective expected utility (SEU), in which beliefs are not known. Instead, beliefs are subjective and unobservable. The analysis will be analogous to what we did for OEU, and therefore proceed at a faster pace. In particular, all the definitions and results parallel those of the section on OEU. The proof of the main result (the axiomatic characterization) is substantially more challenging here because both beliefs and utilities are unknown: there is a classical problem in disentangling beliefs from utility. The technique for solving this problem was introduced in \cite{echenique2015savage}. 
The proofs of the theorems are in Online Appendix~\ref{appendix:proof}. 

\begin{definition} 
Let $e \in \Re_+$. A dataset $(x^k,p^k)_{k=1}^K$ is {\em $e$-belief-perturbed SEU rational} if there exist $\mu^k \in \Delta_{++}(S)$ for each $k\in K$ and a concave and strictly increasing function $u : \Re_+ \rightarrow \Re$ such that, for all $k$, 
\begin{equation*}
y\in B(p^k,p^k\cdot x^k) \implies 
\sum_{s\in S}\mu^k_su(y_s)  \leq \sum_{s\in S}\mu^k_s u(x^k_s)
\end{equation*}
and for each $k, l \in K$ and $s,t \in S$
\begin{equation}\label{eq:seubound}
\frac{\mu^k_s/\mu^k_t}{\mu^l_s/\mu^l_t}\le 1+e.
\end{equation}
\end{definition}

Note that the definition of $e$-belief-perturbed SEU rationality differs from the definition of $e$-belief-perturbed OEU rationality, only in condition~\eqref{eq:seubound}, establishing bounds on perturbations. Here there is no objective probability from which we can evaluate the deviation of the set $\{\mu^k\}$ of beliefs. Thus we evaluate perturbations {\it among} beliefs, as in~\eqref{eq:seubound}. 

\begin{remark}
The constraint on the perturbation applies for each $k,l \in K$ and $s,t \in S$, so it implies for each $k,l \in K$ and $s,t \in S$ 
\[ \frac{1}{1+e} \le \frac{\mu^k_s/\mu^k_t}{\mu^l_s/\mu^l_t}\le 1+e. \]
Hence, when $e=0$, it must be that $\mu^k_s/\mu^k_t=\mu^l_s/\mu^l_t$. This implies that $\mu^k=\mu^l$ for a dataset that is $0$-belief perturbed SEU rational. 
\end{remark}

Next, we propose perturbed SEU rationality with respect to prices. 

\begin{definition} 
Let $e \in \Re_+$. A dataset $(x^k,p^k)_{k=1}^K$ is {\em $e$-price-perturbed SEU rational} if there exist $\mu \in \Delta_{++}(S)$ and a concave and strictly increasing function $u : \Re_+ \rightarrow \Re$ and $\ep^k \in \Re^S_{+}$ for each $k \in K$ such that, for all $k$, 
\begin{equation*}
y\in B(\tilde{p}^k ,\tilde{p}^k \cdot x^k) \implies \sum_{s\in S}\mu_su(y_s)  \leq \sum_{s\in S}\mu_s u(x^k_s),
\end{equation*}
where for each $k\in K$ and $s \in S$
\begin{equation*}
\tilde{p}^k_s= p^k_s \ep^k_s,
\end{equation*}
and for each $k, l \in K$ and $s,t \in S$
\begin{equation}\label{eq:seubound:price}
\frac{\ep^k_s/\ep^k_t}{\ep^l_s/\ep^l_t}\le 1+e.
\end{equation}
\end{definition}

Again, the definition differs from the corresponding definition of price-perturbed OEU rationality only in condition~\eqref{eq:seubound:price}, establishing bounds on perturbations. In condition~\eqref{eq:seubound:price}, we measure the size of the perturbations by 
\[ \frac{\ep^k_s/\ep^k_t}{\ep^l_s/\ep^l_t} , \]
not $\ep^k_s/\ep^k_t$ as in~\eqref{eq:oeubound:price}. This change is necessary to accommodate the existence of subjective beliefs. By choosing subjective beliefs appropriately, one can neutralize the perturbation in prices if  $\ep^k_s/\ep^k_t=\ep^l_s/\ep^l_t$ for all $k, l \in K$. That is, as long as $\ep^k_s/\ep^k_t=\ep^l_s/\ep^l_t$ for all $k, l \in K$, if we can rationalize the dataset by introducing the noise with some subjective belief $\mu$, then without using the noise, we can rationalize the dataset with another subjective belief $\mu'$ such that $\ep^k_s \mu'_s / \ep^k_t \mu'_t= \mu_s/\mu_t$. 

Finally, we define utility-perturbed SEU rationality.

\begin{definition} 
Let $e \in\Re_+$. A dataset $(x^k,p^k)_{k=1}^K$ is {\em $e$-utility-perturbed SEU rational} if there exist $\mu \in \Delta_{++}(S)$, a concave and strictly increasing function $u : \Re_+ \rightarrow \Re$, and $\ep^k \in \Re^S_{+}$ for each $k \in K$ such that, for all $k$, 
\begin{equation*}
y\in B(p^k ,p^k \cdot x^k) \implies \sum_{s\in S}\mu_s \ep^k_s u(y_s)  \leq \sum_{s\in S}\mu_s  \ep^k_s u(x^k_s),
\end{equation*}
and for each $k, l \in K$ and $s,t \in S$
\begin{equation*}
\label{eq:seubound:utility}
\frac{\ep^k_s/\ep^k_t}{\ep^l_s/\ep^l_t} \le 1+e.
\end{equation*}
\end{definition}

As in the previous section, given $e$, we can show that these three concepts of rationality are equivalent. 

\begin{theorem}
\label{theorem:robustseu0} 
Let $e\in\Re_+$ and $D$ be a dataset. The following are equivalent:
\begin{itemize}
\item $D$ is $e$-belief-perturbed SEU rational; 
\item $D$ is $e$-price-perturbed SEU rational;
\item $D$ is $e$-utility-perturbed SEU rational.
\end{itemize}
\end{theorem}

In light of Theorem~\ref{theorem:robustseu0}, we shall speak simply of $e$-perturbed SEU rationality to refer to any of the above notions of perturbed SEU rationality. 

\cite{echenique2015savage} prove that a dataset is SEU rational if and only if it satisfies a revealed-preference axiom termed the Strong Axiom for Revealed Subjective Expected Utility (SARSEU). SARSEU states that, for any test sequence $(x^{k_i}_{s_i}, x^{k'_i}_{s'_i})_{i=1}^n$, if each $s$  appears as $s_i$ (on the left of the pair) the same number of times it appears as $s'_i$ (on the right), then 
\[ \prod_{i=1}^n \frac{p^{k_i}_{s_i}}{p^{k'_i}_{s'_i}} \leq 1. \]

SARSEU is no longer necessary for perturbed SEU-rationality. This is easy to see, as we allow the decision maker to have a different belief $\mu^k$ for each choice $k$, and reason as in our discussion of SAROEU. Analogous to our analysis of OEU, we introduce a perturbed version of SARSEU to capture perturbed SEU rationality. Let $e\in\Re_+$. 

\begin{axiom}[$e$-Perturbed SARSEU ($e$-PSARSEU)]
For any test sequence $(x^{k_i}_{s_i}, x^{k'_i}_{s'_i})_{i=1}^n \equiv \sigma$, if each $s$ appears as $s_i$ (on the left of the pair) the same number of times it appears as $s'_i$ (on the right), then 
\[ \prod_{i=1}^n \frac{p^{k_i}_{s_i}}{p^{k'_i}_{s'_i}} \leq (1+e)^{m(\sigma)}. \]
\end{axiom}

We can easily see the necessity of $e$-PSARSEU by reasoning from the first-order conditions, as in our discussion of $e$-PSAROEU. The main result of this section shows that $e$-PSARSEU is not only necessary for $e$-perturbed SEU rationality, but also sufficient. 

\begin{theorem}
\label{theorem:robustseu}
Let $e\in\Re_+$ and $D$ be a dataset. The following are equivalent:
\begin{itemize}
\item $D$ is $e$-perturbed SEU rational; 
\item $D$ satisfies $e$-PSARSEU.
\end{itemize}
\end{theorem}

It is easy to see that $0$-PSARSEU is equivalent to SARSEU, and that by choosing $e$ to be arbitrarily large it is possible to rationalize any dataset. As a consequence, we shall be interested in finding a minimal value of $e$ that rationalizes a dataset.  \cite{echenique2019decision} apply the idea to datasets of choice under uncertainty collected in the laboratory as well as on the large-scale online survey of the general U.S.\ population.

\section{Conclusion}

We present a measure of deviations from expected utility theory, called minimal~$e$ (or $e_*$), that is based on a revealed-preference characterization of the ``perturbed'' version of the model. 

We start from an observation that the empirical content of EU is captured by the relation between prices and marginal rates of substitution. We measure the deviations from EU by the smallest amount of perturbations one needs to add in order to get the ``right'' relation between prices and marginal rates of substitution. There are three components of the EU model, beliefs, prices, and utilities, which we can perturb, but we can interpret the measure in any of the ways (Theorem~\ref{theorem:robustoeu0}). 

We apply our method to data from three large-scale experiments and find that the measure delivers additional insights on datasets that had been analyzed with CCEI, a measure of consistency with general utility maximization. Our measure can be used as an additional toolkit for data analysis in empirical studies employing choices from linear budgets.

\begin{appendices}
\section{Proofs of Theorems~\ref{theorem:robustoeu0} and~\ref{theorem:robustoeu}}
\label{proof}

\subsection{Proof of Theorem~\ref{theorem:robustoeu0}}

First we prove a lemma that implies Theorem~\ref{theorem:robustoeu0}, and is useful for the sufficiency part of Theorem~\ref{theorem:robustoeu}. The lemma provides ``Afriat inequalities'' for the problem at hand.

\begin{lemma}\label{1lem:smoothisok}
Given $e \in \Re_+$, and let $(x^k,p^k)_{k=1}^K$ be a dataset. The following statements are equivalent. 
\begin{enumerate}[label=(\alph*),ref=(\alph*)]
\item\label{1lem:smoothisok:1} $(x^k,p^k)_{k=1}^K$ is $e$-belief-perturbed OEU rational.
\item\label{1lem:smoothisok:2} There are  strictly positive numbers $v^k_s$, $\la^k$, $\mu^k_s$, for $s\in S$ and $k\in K$, such that
\begin{equation} \label{1thesystem}
\mu^k_s v^k_s = \la^k p^k_s, \;\; \text{ and } \;\;
x^k_s > x^{k'}_{s'} \implies v^k_s \leq v^{k'}_{s'},
\end{equation}
and for all $k \in K$ and $s,t \in S$, 
\begin{equation} \label{1constraint}
\frac{1}{1+e} \le \frac{\mu^k_s/\mu^k_t}{\mu^*_s/\mu^*_t}\le 1+e.
\end{equation}
\item\label{1lem:smoothisok:3} $(x^k,p^k)_{k=1}^K$ is $e$-price-perturbed OEU rational.
\item\label{1lem:smoothisok:4} There are  strictly positive numbers $\hat{v}^k_s$, $\hat{\la}^k$, and $\ep^k_s$ for $s\in S$ and $k\in K$, such that
\begin{equation*} 
\mu^*_s \hat{v}^k_s = \hat{\la}^k \ep^k_s p^k_s, \;\;  \text{ and } \;\;
x^k_s > x^{k'}_{s'} \implies \hat{v}^k_s \leq \hat{v}^{k'}_{s'},
\end{equation*}
and for all $k \in K$ and $s,t \in S$, 
\begin{equation*} 
\frac{1}{1+e} \le \frac{\ep^k_s}{\ep^k_t} \le 1+e . 
\end{equation*} 
\item\label{1lem:smoothisok:5} $(x^k,p^k)_{k=1}^K$ is $e$-utility-perturbed OEU rational.
\item\label{1lem:smoothisok:6} There are  strictly positive numbers $\hat{v}^k_s$, $\hat{\la}^k$, and $\hat{\ep}^k_s$ for $s\in S$ and $k\in K$, such that
\begin{equation*} 
\mu^*_s \hat{\ep}^k_s  \hat{v}^k_s = \hat{\la}^k p^k_s, \;\;  \text{ and } \;\;
x^k_s > x^{k'}_{s'} \implies \hat{v}^k_s \leq \hat{v}^{k'}_{s'},
\end{equation*}
and for all $k \in K$ and $s,t \in S$, 
\begin{equation*} 
\frac{1}{1+e}\le \frac{\hat{\ep}^k_s}{\hat{\ep}^k_t} \le 1+e.
\end{equation*}
\end{enumerate}
\end{lemma}

\begin{proof}
The equivalence between~\ref{1lem:smoothisok:1} and~\ref{1lem:smoothisok:2}, the equivalence between~\ref{1lem:smoothisok:3} and~\ref{1lem:smoothisok:4}, and the equivalence between~\ref{1lem:smoothisok:5} and~\ref{1lem:smoothisok:6} follow from arguments in  \cite{echenique2015savage}. The equivalence between~\ref{1lem:smoothisok:4} and~\ref{1lem:smoothisok:6} with $\ep^k_s=1/ \hat{\ep}^k_s$ for each $k \in K$ and $s \in S$ is straightforward. Thus, to show the result, it suffices to show that~\ref{1lem:smoothisok:2} and~\ref{1lem:smoothisok:4} are equivalent.

To show that~\ref{1lem:smoothisok:4} implies~\ref{1lem:smoothisok:2}, define $v = \hat{v}$ and $\mu^k_s = \frac{\mu^*_s}{\ep^k_s} / \left( \sum_{s \in S} \frac{\mu^*_s}{\ep^k_s} \right)$ for each $k \in K$ and $s \in S$ and $\la^k = \hat{\la}^k / \left( \sum_{s \in S} \frac{\mu^*_s}{\ep^k_s} \right)$ for each $k \in K$.  Then, $\mu^k \in \Delta_{++}(S)$. Since $\mu^*_s \hat{v}^k_s = \hat{\la}^k \ep^k_s p^k_s$, we have $\mu^k_s v^k_s = \la^k p^k_s$. 
Moreover, for each $k \in K$ and $s,t \in S$, $\frac{\ep^k_s}{\ep^k_t} = \frac{\mu^k_s/\mu^k_t}{\mu^*_s/\mu^*_t}$. Hence, $\frac{1}{1+e} \le \frac{\ep^k_s}{\ep^k_t} \le 1+e$.

To show that~\ref{1lem:smoothisok:2} implies~\ref{1lem:smoothisok:4}, for all $s \in S$ define $\hat{v}=v$ and for all  $k \in K$, $\hat{\la}^k=\la^k$. For all  $k \in K$ and $s \in S$, define $\ep^k_s = \frac{\mu^*_s}{\mu^k_s}$. For each $k \in K$ and $s \in S$, since $\mu^k_s u^k_s = \la^k p^k_s$, we have $\mu^*_s v^k_s = \hat{\la}^k \ep^k_s p^k_s$. Finally, for each $k \in K$ and $s,t \in S$, $\frac{\ep^k_s}{\ep^k_t} = \frac{\mu^*_s/ \mu^k_s}{\mu^*_t/\mu^{k}_t} = \frac{\mu^{k}_t/ \mu^k_s}{\mu^*_t/\mu^*_s}$. Therefore, we obtain $\frac{1}{1+e} \le \frac{\ep^k_s}{\ep^k_t} \le 1+e$.
\end{proof}

\subsection{Proof of the Necessity Direction of Theorem~\ref{theorem:robustoeu}}

\begin{lemma}
\label{1lem:necessity}
Given $e \in \Re_+$, if a dataset is $e$-belief-perturbed OEU rational, then the dataset satisfies $e$-PSAROEU.
\end{lemma}

\begin{proof}
Fix any sequence $(x^{k_i}_{s_i}, x^{k'_i}_{s'_i})_{i=1}^n \equiv \sigma$ of pairs that satisfies conditions~\ref{1it:sarseuone} and~\ref{1it:sarseuthree} in Definition~\ref{def:testsequence}. By Lemma~\ref{1lem:smoothisok}, there exist $v^{k_i}_{s_i}, v^{k'_i}_{s'_i}, \la^{k_i}, \la^{k'_i}, \mu^{k_i}_{s_i}, \mu^{k'_i}_{s'_i}$ such that $v^{k'_i}_{s'_i} \ge v^{k_i}_{s_i}$ and $v^{k_i}_{s_i}=\frac{\mu^*_{s_i}}{\mu^{k_i}_{s_i}}\la^{k_i} \rho^{k_i}_{s_i}$, and $v^{k'_i}_{s'_i}=\frac{\mu^*_{s'_i}}{\mu^{k'_i}_{s'_i}}\la^{k'_i} \rho^{k'_i}_{s'_i}$. 
Thus, we have 
\begin{equation*}
1 \ge \prod_{i=1}^n \frac{\la^{k_i}(\mu^{k'_i}_{s'_i}/\mu^*_{s'_i})\rho^{k_i}_{s_i}}{\la^{k'_i} (\mu^{k_i}_{s_i}/\mu^*_{s_i})\rho^{k'_i}_{s'_i}}
= \prod_{i=1}^n \frac{\mu^{k'_i}_{s'_i}/\mu^*_{s'_i}}{\mu^{k_i}_{s_i}/\mu^*_{s_i}}\prod_{i=1}^n \frac{\rho^{k_i}_{s_i}}{\rho^{k'_i}_{s'_i}} ,
\end{equation*}
where the second equality holds by condition~\ref{1it:sarseuthree}. 
Hence, 
\begin{equation*}
\prod_{i=1}^n \frac{\rho^{k_i}_{s_i}}{\rho^{k'_i}_{s'_i}} \le \prod_{i=1}^n \frac{\mu^{k_i}_{s_i}/\mu^*_{s_i}}{\mu^{k'_i}_{s'_i}/\mu^*_{s'_i}} . 
\end{equation*}

In the following, we evaluate the right hand side. For each $(k,s)$, we first cancel out all the terms $\mu^k_s$ that can be canceled out. Then, the number of $\mu^k_s$'s that remain in the numerator is $d(\sigma,k,s)$, as in Definition~\ref{def:sk_remaining}. Since the number of terms in the numerator and the denominator must be the same, the number of remaining fractions is $m(\sigma)\equiv \sum_{s\in S} \sum_{k\in K:d (\sigma, k,s)>0} d(\sigma, k,s)$. So by relabeling the index $i$ to $j$ if necessary, we obtain 
\begin{equation*}
\prod_{i=1}^n \frac{\mu^{k_i}_{s_i}/\mu^*_{s_i}}{\mu^{k'_i}_{s'_i}/\mu^*_{s'_i}}=\prod_{j=1}^{m(\sigma)} \frac{\mu^{k_j}_{s_j}/\mu^*_{s_j}}{\mu^{k'_j}_{s'_j}/\mu^*_{s'_j}} . 
\end{equation*}

Consider the corresponding sequence $(x^{k_j}_{s_j}, x^{k'_j}_{s'_j})_{j=1}^{m(\sigma)}$. Since the sequence is obtained by canceling out $x^k_s$ from the first element and the second element of the pairs, and since the original sequence $(x^{k_i}_{s_i}, x^{k'_i}_{s'_i})_{i=1}^n$ satisfies condition~\ref{1it:sarseuthree}, it follows that  $(x^{k_j}_{s_j}, x^{k'_j}_{s'_j})_{j=1}^{m(\sigma)}$ satisfies condition~\ref{1it:sarseuthree}. 

By condition~\ref{1it:sarseuthree}, we can assume without loss of generality that $k_j=k'_j$ for each $j$. Therefore, by the condition on the perturbation, 
\begin{equation*}
\prod_{j=1}^{m(\sigma)} \frac{\mu^{k_j}_{s_j}/\mu^*_{s_j}}{\mu^{k'_j}_{s'_j}/\mu^*_{s'_j}}\le (1+e)^{m(\sigma)} . 
\end{equation*}
In conclusion, we obtain that $\prod_{i=1}^{n} (\rho^{k_i}_{s_i} / \rho^{k'_i}_{s'_i}) \le (1+e)^{m(\sigma)}$.
\end{proof}

\subsection{Proof of the Sufficiency Direction of Theorem~\ref{theorem:robustoeu}}

We need three lemmas to prove the sufficiency direction. The idea behind the argument is the same as in \cite{echenique2015savage}. We know from Lemma~\ref{1lem:smoothisok} that it suffices to find a solution to the relevant system of Afriat inequalities. We take logarithms to linearize the Afriat inequalities in Lemma~\ref{1lem:smoothisok}.  Then we set up the problem to find a solution to the system of linear inequalities. 

The first lemma, Lemma~\ref{1lem:rationalprice}, shows that $e$-PSAROEU is sufficient for $e$-belief-perturbed OEU rationality under the assumption that the logarithms of the prices are rational numbers. The assumption of rational logarithms comes from our use of a version of the theorem of the alternative (see Lemma~\ref{lem:motzkin2} in Appendix~\ref{appendix:toa}): when there is no solution to the linearized Afriat inequalities, a rational solution to the dual system of inequalities exists. Then we construct a violation of $e$-PSAROEU from the given solution to the dual. 

The second lemma, Lemma~\ref{1lem:approximate}, establishes that we can approximate any dataset satisfying $e$-PSAROEU with a dataset for which the logarithms of prices are rational, and for which $e$-PSAROEU is satisfied.

The last lemma, Lemma~\ref{1lem:anyprice}, establishes the result by using another version of the theorem of the alternative, stated as Lemma~\ref{lem:motzkin1}. 

The rest of the section is devoted to the statement of these lemmas.

\begin{lemma}
\label{1lem:rationalprice}
Given $e \in \Re_+$, let a dataset $(x^k, p^k)_{k=1}^k$ satisfy $e$-PSAROEU. Suppose that $\log(p^k_s)\in \Q$ for all $k \in K$ and $s \in S$, $\log(\mu^*_s)\in \Q$  for all $s \in S$, and $\log (1+e)\in \Q$. Then there are numbers $v^k_s$, $\la^k$, $\mu^k_s$ for $s\in S$ and $k\in K$ satisfying~\eqref{1thesystem} and~\eqref{1constraint} in Lemma~\ref{1lem:smoothisok}. 
\end{lemma}

\begin{lemma}
\label{1lem:approximate}
Given $e \in \Re_+$, let a dataset $(x^k, p^k)_{k=1}^k$ satisfy $e$-PSAROEU with respect to  $\mu^*$. Then for all positive numbers $\ol{\ep}$, there exist a positive real numbers $e' \in [e,e+\ol{\ep}]$, $\mu'_s \in [\mu^*_s-\ol{\ep}, \mu^*_s+\ol{\ep}]$, and $q^{k}_{s} \in [p^k_s-\ol{\ep}, p^k_s]$ for all $s\in S$ and $k \in K$ such that $\log q^k_s \in \Q$ for all $s\in S$ and $k \in K$, $\log(\mu'_s)\in \Q$ for all $s \in S$, and $\log (1+e')\in \Q$, $\mu' \in \Delta_{++}(S)$, and the dataset $(x^k, q^k)_{k=1}^k$ satisfy $e'$-PSAROEU with respect to  $\mu'$. 
\end{lemma} 

\begin{lemma}
\label{1lem:anyprice} 
Given $e \in \Re_+$, let a dataset $(x^k, p^k)_{k=1}^k$ satisfy $e$-PSAROEU with respect to $\mu$. Then there are numbers $v^k_s$, $\la^k$, $\mu^k_s$ for $s\in S$ and $k\in K$ satisfying~\eqref{1thesystem} and~\eqref{1constraint} in Lemma~\ref{1lem:smoothisok}. 
\end{lemma}

\subsubsection{Proof of Lemma~\ref{1lem:rationalprice}}

The proof is similar to the proof of the main result in \cite{echenique2015savage}, which corresponds to the case $e=0$. By log-linearizing the equation in system~\eqref{1thesystem} and the inequality~\eqref{1constraint} in Lemma~\ref{1lem:smoothisok}, we have for all $s\in S$ and $k\in K$, such that
\begin{equation}\label{1eq:logs}
\log \mu^k_s+ \log v^k_s = \log \la^k +\log p^k_s,
\end{equation}
\begin{equation}\label{1eq:concave}
x^k_s > x^{k'}_{s'} \implies \log v^k_s \leq \log v^{k'}_{s'},
\end{equation}
and for all $k \in K$ and $s,t \in S$, 
\begin{equation}\label{1eq:bounds}
-\log (1+e)+\log \mu^*_s-\log \mu^*_t \le \log \mu^k_s-\log \mu^k_t  \le \log (1+e)+\log \mu^*_s-\log \mu^*_t.
\end{equation}

We are going to write the system of inequalities~\eqref{1eq:logs}-\eqref{1eq:bounds} in matrix form, following \cite{echenique2015savage} with some modifications. 

Let $A$ be a matrix with $K \times S$ rows and $2 (K \times S) + K + 1$ columns, defined  as follows: We have one row for every pair $(k, s)$, two columns for every pair $(k, s)$, one columns for each $k$, and one last column. In the row corresponding to $(k, s)$, the matrix has zeroes everywhere with the following exceptions: it has $1$'s in columns for $(k, s)$; it has a~$-1$ in the column for $k$; it has $- \log p^k_s$ in the very last column. 
Matrix $A$ looks as follows: 
\begin{align*}
\hspace{-10em}
\resizebox{1.2\textwidth}{!}{%
\kbordermatrix{
&\cdots&v^k_s&v^k_t&v^l_s&v^l_t&\cdots&&\cdots&\mu^k_s&\mu^k_t&\mu^l_s&\mu^l_t&\cdots&&\cdots&\la^k&\la^l&\cdots&&p \\
&&\vdots&\vdots&\vdots&\vdots&&\vrule&&\vdots&\vdots&\vdots&\vdots&&\vrule&&\vdots&\vdots&&\vrule&\vdots \\
(k,s)&\cdots&1&0&0&0&\cdots&\vrule&\cdots&1&0&0&0&\cdots&\vrule&\cdots&-1&0&\cdots&\vrule&-\log  p^k_s \\
(k,t)&\cdots&0&1&0&0&\cdots&\vrule&\cdots&0&1&0&0&\cdots&\vrule&\cdots&-1&0&\cdots&\vrule&-\log  p^k_t \\
(l,s)&\cdots&0&0&1&0&\cdots&\vrule&\cdots&0&0&1&0&\cdots&\vrule&\cdots&0&-1&\cdots&\vrule&-\log  p^l_s \\
(l,t)&\cdots&0&0&0&1&\cdots&\vrule&\cdots&0&0&0&1&\cdots&\vrule&\cdots&0&-1&\cdots&\vrule&-\log  p^l_t \\
&&\vdots&\vdots&\vdots&\vdots&&\vrule&&\vdots&\vdots&\vdots&\vdots&&\vrule&&\vdots&\vdots&&\vrule&\vdots 
}.}
\end{align*}

Next, we write the system of inequalities~\eqref{1eq:concave} and~\eqref{1eq:bounds} in a matrix form. There is one row in matrix $B$ for each pair $(k, s)$ and $(k', s')$ for which $x_s^k > x_{s'}^{k'}$. In the row corresponding to $x_s^k > x_{s'}^{k'}$, we have zeroes everywhere with the exception of a~$-1$ in the column for $(k, s)$ and a~$1$ in the column for $(k', s')$. 
Matrix $B$ has additional rows, that capture the system of inequalities~\eqref{1eq:bounds}, as follows: 
\begin{align*}
\hspace{-7.5em}
\resizebox{1.15\textwidth}{!}{%
\kbordermatrix{
&\cdots&v^k_s&v^k_t&v^l_s&v^l_t&\cdots&&\cdots&\mu^k_s&\mu^k_t&\mu^l_s&\mu^l_t&\cdots&&\cdots&\la^k&\la^l&\cdots&&p \\
&&\vdots&\vdots&\vdots&\vdots&&\vrule&&\vdots&\vdots&\vdots&\vdots&&\vrule&&\vdots&\vdots&&\vrule&\vdots \\
&\cdots&0&0&0&0&\cdots&\vrule&\cdots&1&-1&0&0&\cdots&\vrule&\cdots&0&0&\cdots&\vrule&\log(1+e)-\log \mu^*_s+\log \mu^*_t\\
&\cdots&0&0&0&0&\cdots&\vrule&\cdots&-1&1&0&0&\cdots&\vrule&\cdots&0&0&\cdots&\vrule&\log(1+e)+\log \mu^*_s-\log \mu^*_t \\
&\cdots&0&0&0&0&\cdots&\vrule&\cdots&0&0&-1&1&\cdots&\vrule&\cdots&0&0&\cdots&\vrule&\log(1+e)+\log \mu^*_s-\log \mu^*_t \\
&\cdots&0&0&0&0&\cdots&\vrule&\cdots&0&0&1&-1&\cdots&\vrule&\cdots&0&0&\cdots&\vrule&\log(1+e)-\log \mu^*_s+\log \mu^*_t \\
&&\vdots&\vdots&\vdots&\vdots&&\vrule&&\vdots&\vdots&\vdots&\vdots&&\vrule&&\vdots&\vdots&&\vrule&\vdots 
}.}
\end{align*}

Finally, we have a matrix $E$ which has a single row and has zeroes everywhere except for~$1$ in the last column. 

To sum up, there is a solution to the system~\eqref{1eq:logs}-\eqref{1eq:bounds} if and only if there is a vector $u \in \Re^{2(K \times S)+K+1}$ that solves the system of equations and linear inequalities
\begin{equation*}
S1:\; 
\begin{cases}
A \cdot u = 0, \\ 
B \cdot u \geq 0, \\
E \cdot u > 0.
\end{cases}
\end{equation*}

The entries of $A$, $B$, and $E$ are either $0$, $1$ or $-1$, with the exception of the last column of $A$ and $B$. Under the hypotheses of the lemma we are proving, the last column consists of rational numbers. By Motzkin's theorem, then, there is such a solution $u$ to $S1$ if and only if there is no rational vector $(\theta,\eta,\pi)$ that solves the system of equations and linear inequalities 
\begin{equation*}
S2:\; 
\begin{cases}
\theta \cdot A + \eta \cdot B +\pi \cdot E = 0, \\ 
\eta \geq 0, \\
\pi > 0.
\end{cases}
\end{equation*}

In the following, we shall prove that the non-existence of a solution
$u$ implies that the dataset must violate $e$-PSAROEU. Suppose then that there is no solution $u$ and let $(\theta,\eta,\pi)$ be a rational vector as above, solving system $S2$. 

The outline of the rest of the proof is similar to the proof of  \cite{echenique2015savage}. Since $(\theta,\eta,\pi)$ are rational vectors, by multiplying a large enough integer, we can make the vectors integers.  Then we transform the matrices $A$ and $B$ using $\theta$ and $\eta$. (i) If $\theta_r>0$, then creat $\theta_r$ copies of the $r$th row; (ii) omitting row $r$ when $\theta_r=0$; and (iii) if $\theta_r< 0$, then $\theta_r$ copies of the $r$th row multiplied by $-1$. 

Similarly, we create a new matrix by including the same columns as $B$ and  $\eta_r$ copies of each row (and thus omitting row $r$ when $\eta_r=0$; recall that $\eta_r\geq 0$ for all $r$). 

By using the transformed matrices and the fact that $\theta \cdot A + \eta \cdot B +\pi \cdot E = 0$ and $\eta \geq 0$, we can prove the following claims:
\medskip

\begin{claim} 
\label{1claim11} 
There exists a sequence $(x^{k_i}_{s_i},x^{k'_i}_{s'_i})_{i=1}^{n^*}\equiv \sigma$ of pairs that satisfies conditions~\ref{1it:sarseuone} and~\ref{1it:sarseuthree} in Definition~\ref{def:testsequence}. 
\end{claim}

\begin{proof} 
We can construct a sequence $(x^{k_i}_{s_i},x^{k'_i}_{s'_i})_{i=1}^{n^*}$ in a similar way to the proof of Lemma~11 of \cite{echenique2015savage}. By construction, the sequence satisfies condition~\ref{1it:sarseuone} that $x^{k_i}_{s_i} >  x^{k'_i}_{s'_i}$ for all $i$.

In the following, we show that the sequence satisfies condition~\ref{1it:sarseuthree} that  each $k$ appears as $k_i$  the same  number of times it appears as $k'_i$. Let $n(x^k_s)\equiv \#\{i \mid x^k_s=x^{k_i}_{s_i}\}$ and $n'(x^k_s)\equiv  \#\{i \mid x^k_s= x^{k'_i}_{s'_i}\}$.
It suffices to show that for each $k\in K$, $\sum_{s \in S} \left[ n(x^k_s)-n'(x^k_s) \right]=0$.

Recall our construction of the matrix $B$. We have a constraint for each triple $(k,s,t)$ with $s<t$. 
Denote the weight on the rows capturing $\frac{\mu^k_s/\mu^k_t}{\mu^*_s/\mu^*_t}\le 1+e$ by $\eta(k,s,t)$ and $1+e \le \frac{\mu^k_s/\mu^k_t}{\mu^*_s/\mu^*_t}$ by $\eta(k,t,s)$. 

For each  $k\in K$ and $s \in S$, in the column corresponding to $\mu^k_s$ in matrix $A$, remember that we have $1$ if we have $x^k_s= x^{k_i}_{s_i}$ for some $i$ and $-1$ if we have $x^k_s= x^{k'_i}_{s'_i}$ for some $i$. This is because a row in $A$ must have $1$ ($-1$) in the column corresponding to $v^k_s$ if and only if it has $1$ ($-1$, respectively) in the column corresponding to $\mu^k_s$. By summing over the column corresponding to $\mu^k_s$, we have $n(x^k_s)-n'(x^k_s)$.  

Now we consider matrix $B$. In the column corresponding to $\mu^k_s$,  we have $1$ in the row multiplied by $\eta(k,t,s)$ and $-1$ in the row multiplied by $\eta(k,s,t)$.  By summing over the column corresponding to $\mu^k_s$, we also have $-\sum_{t\neq s} \eta(k,s,t) + \sum_{t\neq s}\eta(k,t,s)$. 

For each  $k\in K$ and $s \in S$, the column corresponding to $\mu^k_s$ of matrices $A$ and $B$ must sum up to zero; so we have
\begin{equation}\label{eq:123}
n(x^k_s)-n'(x^k_s)+\sum_{t\neq s} \left[ -\eta(k,s,t)+ \eta(k,t,s) \right]=0.
\end{equation}
Hence for each $k\in K$ for each $k\in K$ $\sum_{s \in S} \left[ n(x^k_s)-n'(x^k_s) \right]=0$. 
\end{proof}

\begin{claim} 
\label{1claim13}
$\prod_{i=1}^{n^*} ( \rho^{k_i}_{s_i} / \rho^{k'_i}_{s'_i} ) >(1+e)^{m(\sigma^*)}$.
\end{claim}

\begin{proof}
By~\eqref{eq:123}, 
So for each $s \in S$
\begin{eqnarray*}
\sum_{k \in K} \sum_{s\in S}\sum_{t\neq s} \left[ \eta(k,s,t)-\eta(k,t,s) \right] \log \mu^*_s 
=\sum_{k \in K} \sum_{s\in S} \left[ n(x^k_s)-n'(x^k_s) \right] \log \mu^*_s
=\sum_{i=1}^{n^*} \log  \frac{\mu^*_{s_i}}{\mu^*_{s'_i}},
\end{eqnarray*}
where the last equality holds by the definition of $n$ and $n'$.  Moreover, since $d(\sigma^*, k,s)=n(x^k_s)-n'(x^k_s)= \sum_{t\neq s} \left[ \eta (k,s,t)-\eta(k,t,s) \right] \le  \sum_{t\neq s}\eta(k,s,t)$, we have
\begin{eqnarray*}
m(\sigma^*)\equiv \sum_{s \in S} \sum_{k\in K:d(\sigma^*,k,s)>0 }d(\sigma^*,k,s)=\sum_{s \in S} \sum_{k\in K} \min\{n(x^k_s)-n'(x^k_s),0\}\le \sum_{s \in S} \sum_{k\in K}  \sum_{t\neq s}\eta(k,s,t).
\end{eqnarray*}
By the equality and the inequality above and by the fact that the last column must sum up to zero and $E$ has one at the last column, we have 
\begin{equation*}
\begin{aligned}
0
&> \sum_{i=1}^{n^*} \log  \frac{p^{k'_i}_{s'_i}}{p^{k_i}_{s_i}} +\log (1+e) \sum_{k \in K} \sum_{s\in S}\sum_{t\neq s} \eta(k,s,t) +\sum_{k \in K} \sum_{s\in S}\sum_{t\neq s} (\eta(k,s,t)-\eta(k,t,s)) \log \mu^*_s \\
&= \sum_{i=1}^{n^*} \log  \frac{p^{k'_i}_{s'_i}}{p^{k_i}_{s_i}}-\sum_{i=1}^{n^*} \log  \frac{\mu^*_{s_i}}{\mu^*_{s'_i}}  +\log (1+e) \sum_{k \in K} \sum_{s\in S}\sum_{t\neq s} \eta(k,s,t)\\
&= \sum_{i=1}^{n^*} \log  \frac{\rho^{k'_i}_{s'_i}}{\rho^{k_i}_{s_i}} +\log (1+e) \sum_{k \in K} \sum_{s\in S}\sum_{t\neq s} \eta(k,s,t)\ge \sum_{i=1}^{n^*} \log  \frac{\rho^{k'_i}_{s'_i}}{\rho^{k_i}_{s_i}} +\log (1+e) m(\sigma^*). 
\end{aligned}
\end{equation*}
That is, $\sum_{i=1}^{n^*} \log (\rho^{k_i}_{s_i} / \rho^{k'_i}_{s'_i} ) > m(\sigma^*)\log (1+e)$. This is a contradiction.
\end{proof}

\subsubsection{Proof of Lemma~\ref{1lem:approximate}}

Let $\X= \{x^k_s \mid k \in K, s \in S\}$. 
Consider the set of sequences that satisfy conditions~\ref{1it:sarseuone} and~\ref{1it:sarseuthree} in Definition~\ref{def:testsequence}: 
\begin{equation*} 
\Sigma = 
\left\{ (x^{k_i}_{s_i},x^{k'_i}_{s'_i})_{i=1}^n \subset \X^2 \;\middle|\; 
\begin{array}{l}
(x^{k_i}_{s_i},x^{k'_i}_{s'_i})_{i=1}^n 
\text{ satisfies conditions~\ref{1it:sarseuone} and~\ref{1it:sarseuthree}} \\
\text{in Definition~\ref{def:testsequence} for some } n 
\end{array}
\right\} .
\end{equation*} 
For each sequence $\sa \in \Sigma$, we define a vector $t_\sa\in \N^{K^2S^2}$. For each pair $(x^{k_i}_{s_i},x^{k'_i}_{s'_i})$, we shall identify the pair with $((k_i,s_i),(k'_i,s'_i))$. Let $t_{\sa}((k,s),(k',s'))$ be the number of times that  the pair $(x^{k}_{s}, x^{k'}_{s'})$ appears in the sequence $\sa$. One can then describe the satisfaction of $e$-PSAROEU by means of the vectors $t_\sa$. Observe that $t$ depends only on $(x^k)_{k=1}^K$ in the dataset $(x^k,p^k)_{k=1}^K.$  It does not depend on prices. 

For each $((k,s),(k',s'))$ such that $x^{k}_{s}> x^{k'}_{s'}$, define $\delta((k,s),(k',s')) =\log (p^k_s/p^{k'}_{s'})$. And define $\delta((k,s),(k',s')) = 0$ when $x^{k}_{s}\leq x^{k'}_{s'}$. Then, $\delta$ is a $K^2S^2$-dimensional real-valued vector.
If $\sa = (x^{k_i}_{s_i},  x^{k'_i}_{s'_i})_{i=1}^n$, then 
\[\delta \cdot 
t_\sa = \sum_{((k,s),(k',s')) \in (KS)^2} 
\delta((k,s),(k',s')) 
t_{\sa}((k,s),(k',s')) 
= \log \left( \prod^n_{i=1}\frac{\rho^{k_i}_{s_i}}{\rho^{k'_i}_{s'_i}} \right).
\] 
So the dataset satisfies $e$-PSAROEU with respect to $\mu$ if and only if $\delta \cdot t_{\sigma} \leq m(\sigma)\log (1+e)$ for all $\sigma\in \Sigma$.   

Enumerate the elements in $\X$ in increasing order: $y_1< y_2< \dots <y_N$, and fix an arbitrary  $\ul \xi \in (0,1)$.  We shall construct by induction a sequence $\{(\ep^k_s(n))\}_{n=1}^N$, where  $\ep^k_s(n)$ is defined for all $(k,s)$ with $x^k_s=y_n$. 

By the denseness of the rational numbers, and the continuity of the exponential function, for each $(k,s)$ such that $x^k_s =y_1$, there exists a positive number $\ep^k_s(1)$ such that $\log (\rho^{k}_{s}\ep^k_s(1)) \in \Q$ and $\ul \xi < \ep^k_s(1)< 1$. Let $\ep(1)= \min\{\ep^k_s(1) \mid x^k_s = y_1\}$. 

In second place, for each $(k,s)$ such that $x^k_s =y_2$, there exists a positive $\ep^k_s(2)$ such that $\log (\rho^{k}_{s}\ep^k_s (2)) \in \Q$ and $\ul\xi< \ep^k_s(2) < \ep(1)$. Let $\ep(2)= \min\{\ep^k_s(2) \mid x^k_s = y_2\}$. 

In third place, and reasoning by induction, suppose that $\ep(n)$ has been defined and that $\ul\xi<\ep(n)$. For each $(k,s)$ such that  $x^k_s =y_{n+1}$, let $\ep^k_s(n+1)>0$ be such that $\log (\rho^{k}_{s}\ep^k_s(n+1)) \in \Q$, and $\ul\xi< \ep^k_s(n+1) < \ep(n)$. Let $\ep(n+1)=\min\{\ep^k_s(n+1) \mid x^k_s = y_n\}$. 

This defines the sequence $(\ep^k_s(n))$ by induction. Note that $\ep^k_s(n+1) / \ep(n) < 1$ for all $n$. Let $\bar\xi<1$ be such that $\ep^k_s(n+1) / \ep(n) < \bar\xi$. 

For each $k \in K$ and $s \in S$, let $\hat{\rho}^k_s=\rho^k_s\ep^k_s(n)$, where $n$ is such that  $x^k_s = y_n$. Choose $\mu' \in \Delta_{++}(S)$ such that  for all $s \in S$ $\log \mu'_s \in \Q$ and $\mu'_s \in [\bar\xi \mu_s,\mu_s/\bar\xi]$ for all $s \in S$. Such $\mu'$ exists by the denseness of the rational numbers. Now for each $k \in K$ and $s \in S$, define 
\begin{equation}
\label{price:q}
q^k_s = \frac{\hat{\rho}^k_s}{\mu'_s}.
\end{equation}
Then, $\log q^k_s = \log \hat{\rho}^k_s- \log \mu'_s \in \Q$.

We claim that  the dataset $(x^k,q^k)_{k=1}^K$ satisfies $e'$-PSAROEU with respect to $\mu'$. Let $\da^*$ be defined  from $(q^k)_{k=1}^K$ in the same manner as $\delta$ was defined from $(\rho^k)_{k=1}^K$. 

For each pair $((k,s),(k',s'))$ with $x^{k}_{s}> x^{k'}_{s'}$, if $n$ and $m$ are such that $x^{k}_{s}=y_n$ and $x^{k'}_{s'}=y_m$, then $n>m$. By definition of $\epsilon$, 
\[
\frac{\ep^{k}_{s}(n)}{\ep^{k'}_{s'}(m)}
<\frac{\ep^k_s(n)}{\ep(m)} <\bar \xi
< 1 . 
\] 
Hence,  
\[
\da^*((k,s),(k',s')) = 
\log \frac{\rho^k_s \ep^k_s(n)}{\rho^{k'}_{s'}\ep^{k'}_{s'}(m)} < 
\log \frac{\rho^k_s}{\rho^{k'}_{s'}} +\log \bar \xi 
<\log \frac{\rho^k_s}{\rho^{k'}_{s'}} = \delta((k,s), (k',s')).
\]
Now, we choose $e'$ such that $e' \ge e$ and $\log (1+e') \in \Q$. 

Thus, for all $\sigma \in \Sigma$, $\da^* \cdot t_{\sigma} \leq \delta \cdot t_{\sigma} \leq m(\sigma) \log (1+e) \le m(\sigma) \log (1+e')$ as $t_{\cdot}\geq 0$ and the dataset $(x^k,p^k)_{k=1}^K$ satisfies $e$-PSAROEU with respect to $\mu$. 

Thus the dataset $(x^k,q^k)_{k=1}^K$ satisfies $e'$-PSAROEU with respect to $\mu'$. Finally, note that $\ul \xi < \ep^k_s(n)<1$ for all $n$ and each $k \in K, s \in S$. So that by choosing $\ul\xi$ close enough to~$1$, we can take $\hat{\rho}$ to be as close to $\rho$ as desired. By the definition, we also can take $\mu'$ to be as close to $\mu$ as desired. Consequently, by~\eqref{price:q}, we can take $(q^k)_{k=1}^K$ to be as close to $(p^k)_{k=1}^K$ as desired. We also can take $e'$ to be as close to $e$ as desired.

\subsubsection{Proof of Lemma~\ref{1lem:anyprice}}

We use the following notational convention: For a matrix $D$ with $2(K \times S )+ K +1$ columns, write $D_1$ for the submatrix of $D$ corresponding to the first $K \times S$ columns; let $D_2$ be the submatrix corresponding to the following $K \times S$ columns; $D_3$ correspond to the next $K$ columns; and $D_4$ to the last column. Thus, $D=\left[D_1\, \vrule\,  D_2\,  \vrule\,  D_3 \, \vrule\,  D_4\,  \right]$.

Consider the system comprised by~\eqref{1eq:logs},~\eqref{1eq:concave}, and ~\eqref{1eq:bounds} in the proof of Lemma~\ref{1lem:rationalprice}. 
Let $A$, $B$, and $E$ be constructed from the dataset as in the proof of Lemma~\ref{1lem:rationalprice}. The difference with respect to  Lemma~\ref{1lem:rationalprice} is that now the entries of $A_4$ and $B_4$ may not be rational. Note that the entries of $E$, $B$, and $A_i$, $i=1,2,3$ are rational.

Suppose, towards a contradiction, that there is no solution to the system comprised by~\eqref{1eq:logs},~\eqref{1eq:concave}, and~\eqref{1eq:bounds}. Then, by the argument in the proof of Lemma~\ref{1lem:rationalprice} there is no solution to system $S1$. Lemma~\ref{lem:motzkin1} (in Appendix~\ref{appendix:toa}) with $\F=\Re$ implies that there is a real vector $(\theta,\eta,\pi)$ such that $\theta \cdot A + \eta \cdot B + \pi \cdot E =0$ and $\eta \geq 0, \pi>0$. Recall that $E_4=1$, so we obtain that $\theta\cdot A_4+ \eta\cdot B_4+ \pi =0$. 

Consider $(q^k)_{k=1}^K$, $\mu'$,  and $e'$ be such that the dataset $(x^k,q^k)_{k=1}^K$  satisfies $e'$-PSAROEU with respect to $\mu'$, and $\log q^k_s\in \Q$ for all $k$ and $s$, $\log \mu'_s \in \Q$ for all $s \in S$, and $\log (1+e') \in \Q$. (Such $(q^k)_{k=1}^K$, $\mu'$, and $e'$ exist by Lemma~\ref{1lem:approximate}.)
Construct matrices $A'$, $B'$, and $E'$ from this dataset in the same way as $A$, $B$, and $E$ is constructed in the proof of Lemma~\ref{1lem:rationalprice}. Note that only the prices, the objective probabilities, and the bounds are different. So $E'=E$ and $A'_i=A_i$ and $B'_i=B_i$ for $i=1,2,3$. Only $A'_4$ and $B'_4$ may be different from $A_4$ and $B_4$, respectively. 

By Lemma~\ref{1lem:approximate}, we can choose $q^k$, $\mu'$, and $e'$ such that $|(\theta\cdot A'_4+\eta \cdot B'_4) - (\theta\cdot A_4+\eta \cdot B_4)| < \pi/2$. We have shown that $\theta\cdot A_4 +\eta \cdot B_4= - \pi$, so the choice of $q^k$, $\mu'$, and $e'$ guarantees that $\theta\cdot A'_4 +\eta \cdot B'_4< 0$. Let $\pi'= - \theta\cdot A'_4-\eta \cdot B'_4 > 0$. 

Note that  $\theta\cdot A'_i + \eta \cdot B'_i + \pi' E_i = 0 $ for $i=1,2,3$, as $(\theta,\eta,\pi)$ solves system $S2$ for matrices $A$, $B$ and $E$, and $A'_i=A_i$, $B'_i=B_i$ and $E_i=0$ for $i=1,2,3$. 
Finally, $\theta\cdot A'_4 + \eta \cdot B'_4 + \pi' E_4 = \theta\cdot A'_4  +\eta \cdot B'_4 + \pi'=0$. We also have that $\eta\geq 0$ and $\pi'> 0$. 
Therefore $\theta$, $\eta$, and $\pi'$ constitute a solution to $S2$ for matrices $A'$, $B'$, and $E'$. 

Lemma~\ref{lem:motzkin1}  then implies that there is no solution to system $S1$ for matrices $A'$, $B'$, and $E'$. So there is no solution to the system comprised by~\eqref{1eq:logs},~\eqref{1eq:concave}, and~\eqref{1eq:bounds} in the proof of Lemma~\ref{1lem:rationalprice}. 
However, this contradicts Lemma~\ref{1lem:rationalprice} because the dataset $(x^k,q^k)$ satisfies $e'$-PSAROEU with $\mu'$, $\log (1+e') \in \Q$, $\log \mu'_s \in \Q$ for all $s\in S$, and $\log q^k_s \in \Q$ for all $k\in K$ and $s\in S$. 
\end{appendices}

\renewcommand{\baselinestretch}{0}
\bibliographystyle{ecta}
\bibliography{approximate_eu}

\clearpage 
\renewcommand{\baselinestretch}{1.2}
\setcounter{page}{1}
\appendix 
\setcounter{section}{1}

\numberwithin{equation}{section}
\counterwithin{figure}{section}
\counterwithin{table}{section}
\setcounter{lemma}{\getrefnumber{1lem:anyprice}}  
\setcounter{equation}{\getrefnumber{price:q}}  

\begin{center}
{\LARGE \bf Online Appendix}
\end{center}

\section{Omitted Proofs}
\label{appendix:proof}

\subsection{Proof of Theorem~\ref{theorem:robustseu0}}

First, we prove a lemma which establishes Theorem~\ref{theorem:robustseu0} and proves useful for the sufficiency part of Theorem~\ref{theorem:robustseu}. This lemma provides ``Afriat inequalities'' for the problem at hand. 

\begin{lemma}\label{lem:smoothisok:seu}
Given $e \in \Re_+$, and let $(x^k,p^k)_{k=1}^K$ be a dataset. The following statements are equivalent. 
\begin{enumerate}[label=(\alph*),ref=(\alph*)]
\item\label{1lem:smoothisok:seu:1} $(x^k,p^k)_{k=1}^K$ is $e$-belief-perturbed SEU rational.
\item\label{1lem:smoothisok:seu:2} There are strictly positive numbers $v^k_s$, $\la^k$, $\mu^k_s$, for $s\in S$ and $k\in K$, such that
\begin{equation} \label{thesystem:seu}
\mu^k_s v^k_s = \la^k p^k_s,\qu x^k_s > x^{k'}_{s'} \implies v^k_s \leq v^{k'}_{s'},
\end{equation}
and for each $k, l \in K$ and $s,t \in S$, 
\begin{equation}\label{constraint:seu}
\frac{\mu^k_s/\mu^k_t}{\mu^l_s/\mu^l_t}\le 1+e.
\end{equation}
\item\label{1lem:smoothisok:seu:3} $(x^k,p^k)_{k=1}^K$ is $e$-price-perturbed SEU rational.
\item\label{1lem:smoothisok:seu:4} There are strictly positive numbers $\hat{v}^k_s$, $\hat{\la}^k$, $\mu_s$, and $\ep^k_s$ for $s\in S$ and $k\in K$, such that
\begin{equation*} \label{thesystem':seu}
\mu_s \hat{v}^k_s = \hat{\la}^k \ep^k_s p^k_s,\qu x^k_s > x^{k'}_{s'} \implies \hat{v}^k_s \leq \hat{v}^{k'}_{s'},
\end{equation*}
and for all $k,l \in K$ and $s,t \in S$, 
\begin{equation*} \label{constraint':seu}
\frac{\ep^k_s/\ep^k_t}{\ep^l_s/\ep^l_t}\le 1+e.
\end{equation*}
\item\label{1lem:smoothisok:seu:5} $(x^k,p^k)_{k=1}^K$ is $e$-utility-perturbed SEU rational.
\item\label{1lem:smoothisok:seu:6} There are strictly positive numbers $\hat{v}^k_s$, $\hat{\la}^k$, $\mu_s$, and $\hat{\ep}^k_s$ for $s\in S$ and $k\in K$, such that
\begin{equation*} \label{1thesystem'':seu}
\mu_s \hat{\ep}^k_s  \hat{v}^k_s = \hat{\la}^k p^k_s,\qu x^k_s > x^{k'}_{s'} \implies \hat{v}^k_s \leq \hat{v}^{k'}_{s'},
\end{equation*}
and for all $k,l \in K$ and $s,t \in S$, 
\begin{equation*} \label{constraint'':seu}
\frac{\hat{\ep}^k_s/\hat{\ep}^k_t}{\hat{\ep}^l_s/\hat{\ep}^l_t}\le 1+e.
\end{equation*}
\end{enumerate}
\end{lemma}

\begin{proof}
The equivalence between~\ref{1lem:smoothisok:seu:1} and~\ref{1lem:smoothisok:seu:2}, the equivalence between~\ref{1lem:smoothisok:seu:3} and~\ref{1lem:smoothisok:seu:4}, and the equivalence between~\ref{1lem:smoothisok:seu:5} and~\ref{1lem:smoothisok:seu:6} follow from standard arguments: see \citeSOMpapers{echenique2015som} for details. Moreover, it is easy to see the equivalence between~\ref{1lem:smoothisok:seu:4} and~\ref{1lem:smoothisok:seu:6} with $\ep^k_s=1/ \hat{\ep}^k_s$ for each $k \in K$ and $s \in S$. Hence, to prove the result, it suffices to show that~\ref{1lem:smoothisok:seu:2} and~\ref{1lem:smoothisok:seu:4} are equivalent.

To show that~\ref{1lem:smoothisok:seu:4} implies~\ref{1lem:smoothisok:seu:2}, define $v = \hat{v}$ and 
\[
\mu^k_s = \frac{\mu_s}{\ep^k_s} \Bigg/ \left( \sum_{s \in S}\frac{\mu_s}{\ep^k_s} \right)
\]
for each $k \in K$ and $s \in S$ and 
\[
\la^k = \hat{\la}^k \Bigg/ \left( \sum_{s \in S}\frac{\mu_s}{\ep^k_s} \right)
\]
for each $k \in K$.  Then, $\mu^k \in \Delta_{++}(S)$. Since $\mu_s \hat{v}^k_s = \hat{\la}^k \ep^k_s p^k_s$, we have $\mu^k_s v^k_s = \la^k p^k_s$. 
Moreover, for each $k, l \in K$ and $s,t \in S$, 
\[
\frac{\mu^k_s/\mu^k_t}{\mu^l_s/\mu^l_t}= \frac{\ep^k_t/\ep^k_s}{\ep^l_t/\ep^l_s}\le 1+e.
\]

To show~\ref{1lem:smoothisok:seu:2} implies~\ref{1lem:smoothisok:seu:4}, for all $s \in S$ define $\hat{v}=v$ and
\[
\mu_s = \sum_{k \in K} \frac{\mu^k_s}{|K|}.
\]
Then, $\mu \in \Delta_{++}(S)$. For all  $k \in K$, $\hat{\la}^k=\la^k$. For all  $k \in K$ and $s \in S$, define 
\[
\ep^k_s= \frac{\mu_s}{\mu^k_s}.
\]
For each $k \in K$ and $s\in S$, since $\mu^k_s v^k_s= \la^k p^k_s$, we have $\mu_s v^k_s = \hat{\la}^k \ep^k_s p^k_s$. 
Finally, for each $k,l \in K$ and $s,t \in S$,
\[
\frac{\ep^k_s/\ep^k_t}{\ep^l_s/\ep^l_t}=\frac{\mu^k_t/ \mu^k_s}{\mu^{l}_t/\mu^{l}_s}\le 1+e.
\]
\end{proof}

\subsection{Proof of the Necessity Direction of Theorem~\ref{theorem:robustseu}}

\begin{lemma}\label{lem:necessity}
Given $e \in \Re_+$, if a dataset is $e$-belief-perturbed SEU rational then the dataset satisfies $e$-PSARSEU. 
\end{lemma}

\begin{proof}
Fix any sequence $(x^{k_i}_{s_i}, x^{k'_i}_{s'_i})_{i=1}^n\equiv\sigma$ of pairs that satisfies conditions~\ref{1it:sarseuone} and~\ref{1it:sarseuthree} in Definition~\ref{def:testsequence} and another condition that each $s$ appears as $s_i$ (on the left of the pair) the same number of times it appears as $s'_i$ (on the right), which we refer to as condition~(iii) throughout this section. 
By the standard argument using the concavity of $u$, for each $i$, there exist $v^{k_i}_{s_i}, v^{k'_i}_{s'_i}, \la^{k_i}, \la^{k'_i}, \mu^{k_i}_{s_i}, \mu^{k'_i}_{s'_i}$ such that $v^{k'_i}_{s'_i} \ge v^{k_i}_{s_i}$ and $v^{k_i}_{s_i}=\frac{\la^{k_i} p^{k_i}_{s_i}}{\mu^{k_i}_{s_i}}$, and $v^{k'_i}_{s'_i}=\frac{\la^{k'_i} p^{k'_i}_{s'_i}}{\mu^{k'_i}_{s'_i}}$. 
Thus, we have 
\[
1\ge \prod_{i=1}^n \frac{\la^{k_i} \mu^{k'_i}_{s'_i} p^{k_i}_{s_i}}{\la^{k'_i} \mu^{k_i}_{s_i} p^{k'_i}_{s'_i}}=  \prod_{i=1}^n \frac{\mu^{k'_i}_{s'_i}}{\mu^{k_i}_{s_i}}\prod_{i=1}^n \frac{p^{k_i}_{s_i}} {p^{k'_i}_{s'_i}},
\]
where the second equality holds by condition (ii). See the proof of Lemma 10 of \citeSOMpapers{echenique2015som} for detail. Thus, 
\[
\prod_{i=1}^{n} \frac{p^{k_i}_{s_i}}{p^{k'_i}_{s'_i}}\le \prod_{i=1}^n \frac{\mu^{k_i}_{s_i}}{\mu^{k'_i}_{s'_i}}.
\]

In the following, we evaluate the right hand side. For each $(k,s)$, we first cancel out the same $\mu^k_s$ as much as possible both from the denominator and the numerator. Then, the number of $\mu^k_s$ remained in the numerator is $d(\sigma,k,s)$ as defined in Definition~\ref{def:sk_remaining}. Since the number of terms in the numerator and the denominator must be the same, the number of remaining fraction is $m(\sigma)\equiv \sum_{s\in S} \sum_{k\in K:d (\sigma, k,s)>0} d(\sigma, k,s)$. So by relabeling the index $i$ to $j$ if necessary, we obtain 
\[
\prod_{i=1}^n \frac{\mu^{k_i}_{s_i}}{\mu^{k'_i}_{s'_i}}=\prod_{j=1}^{m(\sigma)} \frac{\mu^{k_j}_{s_j}}{\mu^{k'_j}_{s'_j}}.
\]

Consider the corresponding sequence $(x^{k_j}_{s_j}, x^{k'_j}_{s'_j})_{j=1}^{m(\sigma)}$. Since the sequence is obtained by canceling out $x^k_s$ from the first element and the second element of the pairs the same number of times; and since the original sequence $(x^{k_i}_{s_i}, x^{k'_i}_{s'_i})_{i=1}^n$ satisfies  conditions~\ref{1it:sarseuthree} and~(iii), it follows that  $(x^{k_j}_{s_j}, x^{k'_j}_{s'_j})_{j=1}^{m(\sigma)}$ satisfies conditions~\ref{1it:sarseuthree} and~(iii). 

By condition~(iii), we can assume without loss of generality that $s_j=s'_j$ for each $j$. Fix $s^* \in S$. Then by condition~\eqref{eq:seubound} of $e$-belief perturbed SEU, for each $j \in \{1, \dots, m(\sigma)\}$, 
\[
\frac{\mu^{k_j}_{s_j}}{\mu^{k'_j}_{s'_j}}=\frac{\mu^{k_j}_{s_j}}{\mu^{k'_j}_{s_j}}\le (1+e) \frac{\mu^{k'_j}_{s^*}}{\mu^{k_j}_{s^*}}.
\]
Moreover by condition~\ref{1it:sarseuthree},
\[
\prod_{j=1}^{m(\sigma)} \frac{\mu^{k'_j}_{s^*}}{\mu^{k_j}_{s^*}}=1.
\]
Therefore,
\[
\prod_{i=1}^n \frac{\mu^{k_i}_{s_i}}{\mu^{k'_i}_{s'_i}}=\prod_{j=1}^{m(\sigma)} \frac{\mu^{k_j}_{s_i}}{\mu^{k'_j}_{s'_j}}\le (1+e)^{m(\sigma)} \prod_{j=1}^n \frac{\mu^{k'_j}_{s^*}}{\mu^{k_j}_{s^*}}=(1+e)^{m(\sigma)} ,
\]
and hence, 
\[
\prod_{i=1}^{n} \frac{p^{k_i}_{s_i}}{p^{k'_i}_{s'_i}}\le (1+e)^{m(\sigma)}.
\]
\end{proof}

\subsection{Proof of the Sufficiency Direction in Theorem~\ref{theorem:robustseu}}

The outline of the argument is the same as the proof of Theorem 2 and \citeSOMpapers{echenique2015som}. As in the proof of Theorem 2,  we  need three lemmas to prove the sufficiency direction. 

We know from Lemma~\ref{lem:smoothisok:seu} that it suffices to find a solution to the Afriat inequalities (actually first-order conditions).  So we set up the problem to find a solution to a system of linear inequalities obtained from using logarithms to linearize the Afriat inequalities in Lemma~\ref{lem:smoothisok:seu}. 

The first lemma, Lemma~\ref{lem:rationalprice:seu}, establishes that $e$-PSARSEU is sufficient for e-belief-perturbed SEU rationality  when the logarithms of the prices are rational numbers. 

The second lemma, Lemma~\ref{lem:approximate:seu}, establishes that we can approximate any dataset  satisfying $e$-PSARSEU with a dataset for which the logarithms of prices are rational, and for which $e$-PSARSEU is satisfied. 

Finally, Lemma~\ref{lem:anyprice:seu} establishes the result by using another version of the theorem of the alternative, stated as Lemma~\ref{lem:motzkin1} above. 

The statement of the lemmas follow. The rest of the section is devoted to the proof of these lemmas.

\begin{lemma}\label{lem:rationalprice:seu}
Given $e \in \Re_+$, let a dataset $(x^k, p^k)_{k=1}^k$ satisfy $e$-PSARSEU. Suppose that $\log(p^k_s)\in \Q$ for all $k$ and $s$ and $\log (1+e)\in \Q$. Then there are numbers $v^k_s$, $\la^k$, $\mu^k_s$ for $s\in S$ and $k\in K$ satisfying~\eqref{thesystem:seu} and~\eqref{constraint:seu} in Lemma~\ref{lem:smoothisok:seu}. 
\end{lemma}

\begin{lemma}\label{lem:approximate:seu}
Given $e \in \Re_+$, let a dataset $(x^k, p^k)_{k=1}^k$ satisfy $e$-PSARSEU. Then for all positive numbers $\ol{\ep}$, there exist a positive real number $e' \in [e,e+\ol{\ep}]$ and $q^{k}_{s} \in [p^k_s-\ol{\ep}, p^k_s]$ for all $s\in S$ and $k \in K$ such that $\log q^k_s \in \Q$ and the dataset $(x^k, q^k)_{k=1}^k$ satisfy $e'$-PSARSEU.
\end{lemma} 

\begin{lemma}\label{lem:anyprice:seu} 
Given $e \in \Re_+$, let a dataset $(x^k, p^k)_{k=1}^k$ satisfy $e$-PSARSEU. Then there are numbers $v^k_s$, $\la^k$, $\mu^k_s$ for $s\in S$ and $k\in K$ satisfying~\eqref{thesystem:seu} and~\eqref{constraint:seu}  in Lemma~\ref{lem:smoothisok:seu}. 
\end{lemma}

\subsubsection{Proof of Lemma~\ref{lem:rationalprice:seu}}

The proof is similar to the proof of \citeSOMpapers{echenique2015som}, which corresponds to the case with $e=0$.  By log-linearizing  system~\eqref{thesystem:seu}, and inequality~\eqref{constraint:seu} in Lemma~\ref{lem:smoothisok:seu}, we have for all $s\in S$ and $k\in K$, such that
\begin{equation}\label{eq:logs:seu}
\log \mu^k_s+ \log v^k_s = \log \la^k +\log p^k_s,
\end{equation}
\begin{equation}
\label{eq:concave:seu}
x^k_s > x^{k'}_{s'} \implies \log v^k_s \leq \log v^{k'}_{s'},
\end{equation}
and for all $k,l \in K$ and $s,t \in S$, 
\begin{equation}\label{eq:bounds:seu}
\log \mu^k_s-\log \mu^k_t - \log \mu^l_s+\log \mu^l_t \le \log (1+e).
\end{equation}

We are going to write the system of inequalities~\eqref{eq:logs:seu}-\eqref{eq:bounds:seu} in  matrix form. The formulation follows \citeSOMpapers{echenique2015som}, with some modifications. 

Let $A$ be a matrix with $K \times S$ rows and $2 (K \times S) + K + 1$ columns, defined  as follows: We have one row for every pair $(k, s)$, two columns for every pair $(k, s)$, one column for each $k$, and one last column. In the row corresponding to $(k, s)$, the matrix has zeroes everywhere with the following exceptions: it has $1$'s in columns for $(k, s)$; it has a~$-1$ in the column for $k$; it has $- \log p^k_s$ in the very last column. 
Matrix $A$ looks as follows:
\begin{align*}
\hspace{-10em}
\resizebox{1.2\textwidth}{!}{%
\kbordermatrix{
&\cdots&v^k_s&v^k_t&v^l_s&v^l_t&\cdots&&\cdots&\mu^k_s&\mu^k_t&\mu^l_s&\mu^l_t&\cdots&&\cdots&\la^k&\la^l&\cdots&&p \\
&&\vdots&\vdots&\vdots&\vdots&&\vrule&&\vdots&\vdots&\vdots&\vdots&&\vrule&&\vdots&\vdots&&\vrule&\vdots \\
(k,s)&\cdots&1&0&0&0&\cdots&\vrule&\cdots&1&0&0&0&\cdots&\vrule&\cdots&-1&0&\cdots&\vrule&-\log  p^k_s \\
(k,t)&\cdots&0&1&0&0&\cdots&\vrule&\cdots&0&1&0&0&\cdots&\vrule&\cdots&-1&0&\cdots&\vrule&-\log  p^k_s \\
(l,s)&\cdots&0&0&1&0&\cdots&\vrule&\cdots&0&0&1&0&\cdots&\vrule&\cdots&0&-1&\cdots&\vrule&-\log  p^l_s \\
(l,t)&\cdots&0&0&0&1&\cdots&\vrule&\cdots&0&0&0&1&\cdots&\vrule&\cdots&0&-1&\cdots&\vrule&-\log  p^l_s \\
&&\vdots&\vdots&\vdots&\vdots&&\vrule&&\vdots&\vdots&\vdots&\vdots&&\vrule&&\vdots&\vdots&&\vrule&\vdots 
}.}
\end{align*}

Next, we write the system of inequalities~\eqref{eq:concave:seu} and~\eqref{eq:bounds:seu} in matrix form. There is one row in matrix $B$ for each pair $(k, s)$ and $(k', s')$ for which $x_s^k > x_{s'}^{k'}$. In the row corresponding to $x_s^k > x_{s'}^{k'}$, we have zeroes everywhere with the exception of a~$-1$ in the column for $(k, s)$ and a~$1$ in the column for $(k', s')$. 
Matrix $B$ has additional rows, that capture the system of inequalities~\eqref{eq:bounds:seu}: 
We do not need a constraint for each quadruple $(k,l,s,t)$, as some of
them would be redundant. Specifically, we need the constraints 
$\frac{\mu^k_s/\mu^k_t}{\mu^l_s/\mu^l_t}\le 1+e$, and
$\frac{\mu^l_s/\mu^l_t}{\mu^k_s/\mu^k_t}\le 1+e$, which is equivalent
to $\frac{\mu^k_s/\mu^k_t}{\mu^l_s/\mu^l_t} \ge 1/(1+e)$. But note
that  $\frac{\mu^l_t/\mu^l_s}{\mu^k_t/\mu^k_s}\le 1+e$ is redundant,
as $\frac{\mu^l_t/\mu^l_s}{\mu^k_t/\mu^k_s}=\frac{\mu^k_s/\mu^k_t}{\mu^l_s/\mu^l_t}$.
So for each $(s,t)$ with $s<t$, and each $k\neq l$ we are going to have the constraint $(k,l,s,t)$.\footnote{The inequality $s<t$ is simply a devise to ensure that we choose only one of the two ordered pairs of $s$ and $t$.} 
For each such $(k,l,s,t)$ we have two rows. One of these rows has a $1$ in the column for $\mu^k_s$ and $\mu^l_t$, a $-1$ in the column for $\mu^k_t$ and $\mu^l_s$, and $\log(1+e)$ in the very last column; one of these rows has a $1$ in the column for $\mu^k_t$ and $\mu^l_s$, a $-1$ in the column for $\mu^k_s$ and $\mu^l_t$, and $\log(1+e)$ in the very last column. So this part of matrix $B$ is as follows:
\begin{align*}
\hspace{-10em}
\resizebox{1.2\textwidth}{!}{%
\kbordermatrix{
&\cdots&v^k_s&v^k_t&v^l_s&v^l_t&\cdots&&\cdots&\mu^k_s&\mu^k_t&\mu^l_s&\mu^l_t&\cdots&&\cdots&\la^k&\la^l&\cdots&&p \\
&&\vdots&\vdots&\vdots&\vdots&&\vrule&&\vdots&\vdots&\vdots&\vdots&&\vrule&&\vdots&\vdots&&\vrule&\vdots \\
&\cdots&0&0&0&0&\cdots&\vrule&\cdots&-1&1&1&-1&\cdots&\vrule&\cdots&0&0&\cdots&\vrule&\log(1+e) \\
&\cdots&0&0&0&0&\cdots&\vrule&\cdots&1&-1&-1&1&\cdots&\vrule&\cdots&0&0&\cdots&\vrule&\log(1+e) \\
&&\vdots&\vdots&\vdots&\vdots&&\vrule&&\vdots&\vdots&\vdots&\vdots&&\vrule&&\vdots&\vdots&&\vrule&\vdots 
}.}
\end{align*}

Finally, we have a matrix $E$ which has a single row and has zeroes everywhere except for~$1$ in the last column. 

To sum up, there is a solution to the system~\eqref{eq:logs:seu}-\eqref{eq:bounds:seu} if and only if there is a vector $u \in \Re^{2(K \times S)+K+1}$ that solves the system of equations and linear inequalities
\begin{equation*}
S1:\; 
\begin{cases}
A \cdot u = 0, \\ 
B \cdot u \geq 0, \\
E \cdot u > 0.
\end{cases}
\end{equation*}

The entries of $A$, $B$, and $E$ are either $0$, $1$ or $-1$, with the exception of the last column of $A$ and $B$. Under the hypotheses of the lemma we are proving, the last column consists of rational numbers. By Motzkin's theorem, then, there is such a solution $u$ to $S1$ if and only if there is no rational vector $(\theta,\eta,\pi)$ that solves the system of equations and linear inequalities 
\[
S2:\; \begin{cases}
\theta \cdot A + \eta \cdot B +\pi \cdot E =0, \\ 
\eta \geq 0, \\
\pi > 0.
\end{cases}
\]

In the following, we shall prove that the non-existence of a solution $u$ implies that the dataset must violate $e$-PSARSEU. Suppose then that there is no solution $u$ and let $(\theta,\eta,\pi)$ be a rational vector as above, solving system $S2$.  

The outline of the rest of the proof is similar to the proof of Theorem~\ref{theorem:robustoeu}. Since $(\theta,\eta,\pi)$ are rational vectors, by multiplying all of their entries by a large enough integer, we can without loss of generality assume that  $(\theta,\eta,\pi)$ are integer vectors.  

Then we transform the matrices $A$ and $B$ using $\theta$ and $\eta$. (i) If $\theta_r>0$, then create $\theta_r$ copies of the $r$th row; (ii) omitting row $r$ when $\theta_r=0$;  and (iii) if $\theta_r< 0$, then $\theta_r$ copies of the $r$th row multiplied by $-1$. 

Similarly, we create a new matrix by including the same columns as $B$ and  $\eta_r$ copies of each row (and thus omitting row $r$ when $\eta_r=0$; recall that $\eta_r\geq 0$ for all $r$). 

By using the transformed matrices and the fact that $\theta \cdot A + \eta \cdot B +\pi \cdot E = 0$ and $\eta \geq 0$, we can prove the following claims:
\medskip

\begin{claim} \label{claim11:seu} 
There exists a sequence $(x^{k_i}_{s_i},x^{k'_i}_{s'_i})_{i=1}^{n^*}$ of pairs that satisfies conditions~\ref{1it:sarseuone} and~\ref{1it:sarseuthree} in Definition~\ref{def:testsequence}. 
\end{claim}

\begin{proof}
The proof is the same as in the proof of Lemma~11 in \citeSOMpapers{echenique2015som}.
\end{proof}

\begin{claim} \label{claim12:seu} 
In the sequence $(x^{k_i}_{s_i},x^{k'_i}_{s'_i})_{i=1}^{n^*}\equiv \sigma^*$,  each $s$ appears as $s_i$ (on the left of the pair) the same number of times it appears as $s'_i$ (on the right). 
\end{claim} 

\begin{proof} Recall our construction of the matrix $B$. We have a constraint for each  quadruple $(k,l,s,t)$ with $s<t$. 
Denote the weight on the rows capturing
$\frac{\mu^k_s/\mu^k_t}{\mu^l_s/\mu^l_t}\le 1+e$  by
$\eta(k,l,s,t)$. Let $n(x^k_s)\equiv \#\{i \mid x^k_s=
x^{k_i}_{s_i}\}$ and $n'(x^k_s)\equiv  \#\{i \mid x^k_s=
x^{k'_i}_{s'_i}\}$. For notational convenience, define
$\eta(k,l,s,t)=0$ for all quadruples $(k,l,s,t)$ with $t<s$.

For each  $k\in K$ and $s \in S$, in the column corresponding to $\mu^k_s$ in matrix $A$, remember that we have $1$ if we have $x^k_s= x^{k_i}_{s_i}$ for some $i$ and $-1$ if we have $x^k_s= x^{k'_i}_{s'_i}$ for some $i$. This is because a row in $A$ must have $1$ ($-1$) in the column corresponding to $v^k_s$ if and only if it has $1$ ($-1$, respectively) in the column corresponding to $\mu^k_s$. By summing over the column corresponding to $\mu^k_s$, we have $n(x^k_s)-n'(x^k_s)$.  

Now we consider matrix $B$. In the column corresponding to $\mu^k_s$
and $s<t$,  we have $-1$ in the row multiplied by $\eta(k,l,s,t)$ and $1$ in the row multiplied by $\eta(l,k, s,t)$.  By summing over the column corresponding to $\mu^k_s$, we also have $-\sum_{l \neq k}\sum_{t\neq s} \eta(k,l,s,t) +\sum_{l \neq k}\sum_{t\neq s}\eta(l,k,s,t)$. 

For each  $k\in K$ and $s \in S$, the column corresponding to $\mu^k_s$ of matrices $A$ and $B$ must sum up to zero; so we have
\[
n(x^k_s)-n'(x^k_s)-\sum_{l \neq k}\sum_{t\neq s} \eta(k,l,s,t) +\sum_{l \neq k}\sum_{t\neq s} \eta(l,k,s,t)=0.
\]

Therefore, for each $s$,
\begin{align*}
  \sum_{k\in K}\Big(n(x^k_s)-n'(x^k_s)\Big)
& =\sum_{k\in K}\left[\sum_{l \neq k}\sum_{t\neq s} \eta(k,l,s,t)
  -\sum_{l \neq k}\sum_{t\neq s} \eta(l,k,s,t)\right] \\
& =\sum_{t\neq s}  \left[\sum_{k\in K}\sum_{l\neq k} \eta(k,l,s,t) - \sum_{k\in
  K}\sum_{l\neq k} \eta(l,k,s,t)\right]  \\
& =0.
\end{align*}
This means that each $s$ appears as $s_i$ (on the left of the pair) the same number of times it appears as $s'_i$ (on the right). 
\end{proof}

\begin{claim} \label{claim13:seu}
$\prod_{i=1}^{n^*} ( p^{k_i}_{s_i} / p^{k'_i}_{s'_i} ) >(1+e)^{m(\sigma^*)}$.
\end{claim}

\begin{proof}
By the fact that the last column must sum up to zero and $E$ has one at the last column, we have 
\[
\sum_{i=1}^{n^*} \log  \frac{p^{k'_i}_{s'_i}}{p^{k_i}_{s_i}} +\left( \sum_{k \in K}\sum_{l \neq k}\sum_{s\in S}\sum_{t \neq s} \eta(k,l,s,t) \right) \log (1+e) = -\pi<0.
\]
Hence, by multiplying $-1$, we have
\[
\sum_{i=1}^{n^*} \log  \frac{p^{k_i}_{s_i}}{p^{k'_i}_{s'_i}}- \left( \sum_{k \in K}\sum_{l \neq k}\sum_{s\in S}\sum_{t \neq s} \eta(k,l,s,t) \right) \log (1+e)>0.
\]
Remember that for all $k \in K$ and $s \in S$, 
\begin{eqnarray*}
\begin{array}{lllll}
n(x^k_s)-n'(x^k_s)
=\displaystyle +\sum_{l \neq k}\sum_{t\neq s} \eta(k,l,s,t) -\sum_{l \neq k}\sum_{t\neq s} \eta(l,k,s,t) \le\displaystyle \sum_{l \neq k}\sum_{t\neq s} \eta(k,l,s,t).
\end{array}
\end{eqnarray*}
Since $d(\sigma^*, k,s)=n(x^k_s)-n'(x^k_s)$, we have
\begin{equation*}
\begin{aligned}
m(\sigma^*)
\equiv \sum_{s \in S} \sum_{k\in K:d(\sigma^*,k,s)>0 }d(\sigma^*, k,s)
&=\sum_{s \in S} \sum_{k\in K} \max\{n(x^k_s)-n'(x^k_s),0\} \\
&\le \displaystyle \sum_{s \in S} \sum_{k\in K}\sum_{l \neq k}\sum_{t\neq s} \eta(k,l,s,t).
\end{aligned}
\end{equation*}
Therefore, 
\begin{equation*}
\sum_{i=1}^{n^*} \log  \frac{p^{k_i}_{s_i}}{p^{k'_i}_{s'_i}}
> \left( \sum_{k \in K}\sum_{l \neq k}\sum_{s\in S}\sum_{t\neq s} \eta(k,l,s,t) \right) \log (1+e) 
\geq m(\sigma^*) \log (1+e).
\end{equation*}
This is a contradiction.
\end{proof}

\subsubsection{Proof of Lemma~\ref{lem:approximate:seu}}
Let $\X= \{x^k_s \mid k \in K, s \in S\}$. 
Consider the set of sequences that satisfy conditions~\ref{1it:sarseuone} and~\ref{1it:sarseuthree} in Definition~\ref{def:testsequence}, and (iii) in $e$-PSARSEU: 
\begin{eqnarray*}
\Sigma = 
\left\{ (x^{k_i}_{s_i},x^{k'_i}_{s'_i})_{i=1}^n \subset \X^2 \,\middle|\, 
\begin{array}{l}
(x^{k_i}_{s_i},x^{k'_i}_{s'_i})_{i=1}^n 
\text{ satisfies conditions~\ref{1it:sarseuone} and~\ref{1it:sarseuthree}} \\ 
\text{in Definition~\ref{def:testsequence} and (iii)  for some } n 
\end{array}
\right\} . 
\end{eqnarray*}
For each sequence $\sa \in \Sigma$, we define a vector $t_\sa\in \N^{K^2S^2}$. For each pair $(x^{k_i}_{s_i},x^{k'_i}_{s'_i})$, we shall identify the pair with $((k_i,s_i),(k'_i,s'_i))$. Let $t_{\sa}((k,s),(k',s'))$ be the number of times that  the pair $(x^{k}_{s}, x^{k'}_{s'})$ appears in the sequence $\sa$. One can then describe the satisfaction of $e$-PSARSEU by means of the vectors $t_\sa$. Observe that $t$ depends only on $(x^k)_{k=1}^K$ in the dataset $(x^k,p^k)_{k=1}^K.$  It does not depend on prices. 

For each $((k,s),(k',s'))$ such that $x^{k}_{s}> x^{k'}_{s'}$, define $\delta((k,s),(k',s')) =\log (p^k_s/p^{k'}_{s'})$. And define $\delta((k,s),(k',s')) = 0$ when $x^{k}_{s}\leq x^{k'}_{s'}$. 
Then, $\delta$ is a $K^2S^2$-dimensional real-valued vector. If $\sa = (x^{k_i}_{s_i},  x^{k'_i}_{s'_i})_{i=1}^n$, then 
\[\delta \cdot 
t_\sa = \sum_{((k,s),(k',s')) \in (K \times S)^2} 
\delta((k,s),(k',s')) 
t_{\sa}((k,s),(k',s')) 
= \log \left( \prod^n_{i=1}\frac{p^{k_i}_{s_i}}{p^{k'_i}_{s'_i}}\right) . \] 

So the dataset satisfies  $e$-PSARSEU if and only if $\delta \cdot t_{\sigma} \leq m(\sigma)\log (1+e)$ for all $\sigma\in \Sigma$.   

Enumerate the elements in $\X$ in increasing order: $y_1< y_2< \dots <y_N$, and fix an arbitrary  $\ul \xi \in (0,1)$. We shall construct by induction a sequence $\{(\ep^k_s(n))\}_{n=1}^N$, where  $\ep^k_s(n)$ is defined for all $(k,s)$ with $x^k_s=y_n$. 

By the denseness of the rational numbers, and the  continuity of the exponential function, for each $(k,s)$ such that $x^k_s =y_1$, there exists a positive number $\ep^k_s(1)$ such that $\log(p^{k}_{s}\ep^k_s(1)) \in \Q$ and $\ul \xi < \ep^k_s(1)< 1$. Let $\ep(1)= \min\{\ep^k_s(1) \mid x^k_s = y_1\}$. 

In second place, for each $(k,s)$ such that $x^k_s =y_2$, there exists a positive $\ep^k_s(2)$ such that $\log (p^{k}_{s}\ep^k_s (2)) \in \Q$ and $\ul\xi< \ep^k_s(2) < \ep(1)$. Let $\ep(2)= \min\{\ep^k_s(2) \mid x^k_s = y_2\}$. 

In third place, and reasoning by induction, suppose that $\ep(n)$ has been defined and that $\ul\xi<\ep(n)$. For each $(k,s)$ such that  $x^k_s =y_{n+1}$, let $\ep^k_s(n+1)>0$ be such that $\log (p^{k}_{s}\ep^k_s(n+1)) \in \Q$, and $\ul\xi< \ep^k_s(n+1) < \ep(n)$. Let $\ep(n+1)=\min\{\ep^k_s(n+1) \mid x^k_s = y_n\}$. 

This defines the sequence $(\ep^k_s(n))$ by induction. Note that $\ep^k_s(n+1) / \ep(n) < 1$ for all $n$.  Let $\bar\xi<1$ be such that $\ep^k_s(n+1) / \ep(n) < \bar\xi$. 

For each $k \in K$ and $s \in S$, let $q^k_s=p^k_s\ep^k_s(n)$, where $n$ is such that  $x^k_s = y_n$. We claim that the dataset $(x^k,q^k)_{k=1}^K$ satisfies $e$-PSARSEU. Let $\da^*$ be defined  from $(q^k)_{k=1}^K$ in the same manner as $\delta$ was defined from $(p^k)_{k=1}^K$.   

For each pair $((k,s),(k',s'))$ with $x^{k}_{s}> x^{k'}_{s'}$, if $n$ and $m$ are such that $x^{k}_{s}=y_n$ and $x^{k'}_{s'}=y_m$, then  $n>m$. By definition of $\ep$, 
\[
\frac{\ep^{k}_{s}(n)}{\ep^{k'}_{s'}(m)}
<\frac{\ep^k_s(n)}{\ep(m)} <\bar \xi
< 1.\] 
Hence,  
\[
\da^*((k,s),(k',s')) = 
\log \frac{p^k_s \ep^k_s(n)}{p^{k'}_{s'}\ep^{k'}_{s'}(m)} < 
\log \frac{p^k_s}{p^{k'}_{s'}} +\log \bar \xi 
<\log \frac{p^k_s}{p^{k'}_{s'}} = \delta((k,s), (k',s')).
\]
Now we choose $e'$ such that $e' \ge e$ and $\log (1+e') \in \Q$. 

Thus, for all $\sigma \in \Sigma$, $\da^* \cdot t_{\sigma} \leq \delta \cdot t_{\sigma} \leq m(\sigma) \log (1+e) \le m(\sigma) \log (1+e')$ as $t_{\cdot}\geq 0$ and  the dataset $(x^k,p^k)_{k=1}^K$ satisfies $e$-PSARSEU.

Therefore, the dataset  $(x^k,q^k)_{k=1}^K$ satisfies $e'$-PSARSEU.
Finally, note that $\ul \xi < \ep^k_s(n)<1$ for all $n$ and each $k \in K, s \in
S$. So that by choosing $\ul\xi$ close enough to~$1$ we can take $(q^k)_{k=1}^K$ to be as close to $(p^k)_{k=1}^K$ as desired.  We also can take $e'$ to be as close to $e$ as desired.

\subsubsection{Proof of Lemma~\ref{lem:anyprice:seu}}

Consider the system comprised by~\eqref{eq:logs:seu},~\eqref{eq:concave:seu}, and ~\eqref{eq:bounds:seu} in the proof of Lemma~\ref{lem:rationalprice:seu}. 
Let $A$, $B$, and $E$ be constructed from the dataset as in the proof of Lemma~\ref{lem:rationalprice:seu}. The difference with respect to  Lemma~\ref{lem:rationalprice:seu} is that now the entries of $A_4$ and $B_4$ may not be rational. Note that the entries of $E$, $B_i$, and $A_i$, for $i=1,2,3$ are rational. 

Suppose, towards a contradiction, that there is no solution to the system comprised by~\eqref{eq:logs:seu},~\eqref{eq:concave:seu}, and ~\eqref{eq:bounds:seu}. Then, by the argument in the proof of Lemma~\ref{lem:rationalprice:seu} there is no solution to system $S1$.  
Lemma~\ref{lem:motzkin1} (in Appendix~\ref{appendix:toa}) with $\F=\Re$ implies that there is a real vector $(\theta,\eta,\pi)$ such that $\theta \cdot A + \eta \cdot B + \pi \cdot E =0$ and $\eta \geq 0, \pi>0$. 
Recall that $E_4=1$, so we obtain that $\theta\cdot A_4 + \eta \cdot B_4+\pi =0$. 

Let $(q^k)_{k=1}^K$ vectors of prices and a positive real number $e'$ be such that the dataset $(x^k,q^k)_{k=1}^K$  satisfies $e'$-PSARSEU and $\log q^k_s\in \Q$ for all $k$ and $s$ and $\log (1+e') \in \Q$. (Such $(q^k)_{k=1}^K$ and $e'$ exist by Lemma~\ref{lem:approximate:seu}.)
Construct matrices $A'$, $B'$, and $E'$ from this dataset in the same way as $A$, $B$, and $E$ is constructed in the proof of Lemma~\ref{lem:rationalprice:seu}. Since only prices $q^k$ and the bound $e'$ are different in this dataset, only $A'_4$ and $B'_4$ may be different from $A_4$ and $B_4$, respectively.  So $E'=E$, $B'_i=B_i$ and $A'_i=A_i$ for $i=1,2,3$. 

By Lemma~\ref{lem:approximate:seu}, we can choose prices $q^k$ and $e'$ such that $|(\theta\cdot A'_4+ \eta \cdot B'_4) - (\theta\cdot A_4+\eta \cdot B_4)| < \pi/2$. 
We have shown that $\theta\cdot A_4+\eta \cdot B_4 = - \pi$, so the choice of prices $q^k$ and $e'$ guarantees that $\theta\cdot A'_4+ \eta \cdot B'_4< 0$. Let $\pi'= - \theta\cdot A'_4-\eta \cdot B'_4 > 0$. 

Note that  $\theta\cdot A'_i + \eta \cdot B'_i + \pi' E_i = 0 $ for $i=1,2,3$, as $(\theta,\eta,\pi)$ solves system $S2$ for matrices $A$, $B$ and $E$, and $A'_i=A_i$, $B'_i=B_i$ and $E_i=0$ for $i=1,2,3$. 
Finally, $\theta\cdot A'_4 + \eta \cdot B'_4 + \pi' E_4 = \theta\cdot A'_4 +\eta\cdot B'_4 + \pi'=0.$
We also have that $\eta\geq 0$ and $\pi'> 0$. 
Therefore $\theta$, $\eta$, and $\pi'$ constitute a solution to $S2$ for matrices $A'$, $B'$, and $E'$. 

Lemma~\ref{lem:motzkin1} then implies that there is no solution to system $S1$ for matrices $A'$, $B'$, and $E'$. So there is no solution to the system comprised by~\eqref{eq:logs:seu},~\eqref{eq:concave:seu}, and~\eqref{eq:bounds:seu} in the proof of Lemma~\ref{lem:rationalprice:seu}. 
However, this contradicts Lemma~\ref{lem:rationalprice:seu} because the dataset $(x^k,q^k)$ satisfies $e'$-PSARSEU, $\log (1+e') \in \Q$,  and $\log q^k_s \in \Q$ for all $k\in K$ and $s\in S$.

\clearpage
\subsection{Theorem of the Alternative}
\label{appendix:toa}

We shall use the following lemma, which is a version of the Theorem of the Alternative. This is Theorem~1.6.1 in \citeSOMpapers{stoer1970convexity}. We shall use it here in the cases where $F$ is either the real or the rational number field. 

\begin{lemma} 
\label{lem:motzkin1} 
Let $A$ be an $m\times n$ matrix, $B$ be an $l\times n$ matrix, and $E$ be an $r\times n$ matrix. Suppose that the entries of the matrices $A$, $B$, and $E$ belong to a commutative ordered field $\F$. 
Exactly one of the following alternatives is true.
\begin{enumerate}
\item There is $u\in\F^n$ such that $A\cdot u = 0$, $B\cdot u \geq 0$, $E\cdot u \gg 0$. 
\item There is $\theta \in \F^r$, $\eta \in \F^l$, and $\pi \in \F^m$ such  that $\theta \cdot A+  \eta \cdot B+ \pi \cdot E = 0$; $\pi>0$ and  $\eta \geq 0$.
\end{enumerate}
\end{lemma}

The next lemma is a direct consequence of Lemma~\ref{lem:motzkin1}. See Lemma~12 in \citeSOMpapers{chambersechenique} for a proof.  

\begin{lemma} 
\label{lem:motzkin2}
Let $A$ be an $m\times n$ matrix, $B$ be an $l\times n$ matrix, and $E$ be an $r\times n$ matrix. Suppose that the entries of the matrices $A$, $B$, and $E$ are rational numbers.
Exactly one of the following alternatives is true.
\begin{enumerate}
\item There is $u\in\Re^n$ such that $A\cdot u = 0$, $B\cdot u \geq 0$, and $E\cdot u \gg 0$.
\item There is $\theta \in \Q^r$, $\eta\in \Q^l$, and $\pi\in \Q^m$ such  that $\theta \cdot A + \eta \cdot B +  \pi \cdot E = 0$; $\pi >0$ and  $\eta \geq 0$. 
\end{enumerate}
\end{lemma}

\clearpage 
\section{Computing $e_*$}
\label{sec:compute_e}

We demonstrate how to calculate $e_*$ given a dataset of choice under risk. To calculate the value, it is easier to use price-perturbed OEU rationality, rather than belief-perturbed OEU rationality. Formally, for a given data set $(x^k,p^k)_{k=1}^K$, we want to compute $e_*$ such that the data set is price perturbed OEU rational given the number $e$. We can transform this problem into an easier problem with the following remark. 

\begin{remark}
Given $e \in \Re_+$, a data set $(x^k,p^k)_{k=1}^K$ is $e$-price-perturbed OEU rational if and only if there are strictly positive numbers $v^k_s$, $\la^k$, $\mu_s$, and $\ep^k_s$ for $s\in S$ and $k\in K$, such that
\begin{equation} \label{1thesystem'''}
\mu^*_s v^k_s = \la^k \ep^k_s p^k_s,\qu x^k_s > x^{k'}_{s'} \implies v^k_s \leq v^{k'}_{s'},
\end{equation}
and for all $k \in K$ and $s,t \in S$
\begin{equation*} \label{1constraint''}
\frac{1}{1+e}\le \frac{\ep^k_s}{\ep^k_t} \le 1+e.
\end{equation*}
\label{remark:e-oeu_linear_problem}
\end{remark}

By the remark, the $e_*$ can be obtained by solving the following problem: 
\begin{equation*}
\begin{aligned}
\min_{(\mu_s, v^k_s, \la^k, \ep^k_s)_{k,s}} & \max_{k \in K, s, t \in S}  \frac{\ep^k_s}{\ep^k_t} \\
\text{ s.t. } & \;\; \mu^*_s v^k_s = \la^k \ep^k_s p^k_s , \\
& \;\; x^k_s > x^{k'}_{s'} \implies v^k_s \leq v^{k'}_{s'} . 
\end{aligned}
\end{equation*}

We then substitute $\ep^k_s$ in the objective function by using the equality constraint in (\ref{1thesystem'''}). By canceling out $\la^k $ and  log-linearizing, we obtain the following:
\begin{equation}
\begin{aligned}
\min_{(v^k_s)_{k,s}} & \max_{k \in K , s, t \in S} (\log \mu^*_s + \log v^k_s - \log p^k_s) - (\log \mu^*_t + \log v^k_t - \log p^k_t) \\
\text{s.t.} & \;\;\; x^k_s > x^{k'}_{s'} \implies \log v^k_s \le  \log v^{k'}_{s'} . 
\end{aligned}
\tag{$\star$}
\label{eq:robustOEU_problem_original}
\end{equation}

By the discussion above, we have the following result:

\begin{remark}
For any data set $(x^k,p^k)_{k=1}^K$, $e_*$ is the solution of the problem~\eqref{eq:robustOEU_problem_original}, which always exists. 
\label{remark:e-oeu_linear_problem_existence}
\end{remark}

By using~\eqref{eq:robustOEU_problem_original} and the peculiarities
of the experiments, we can  simplify the problem: we have $|S| = 2$ and $\mu^*_s = 1/2$ for all $s \in S$. Hence, the problem simplifies to the following: 
\begin{equation*}
\begin{aligned}
\min_{(v^k_s)_{k,s}} & \max_{k \in K, s, t \in S} (\log v^k_s - \log p^k_s) - (\log v^k_t - \log p^k_t) \\
\text{s.t.} & \;\; \phantom{(} x^k_s > x^{k'}_{s'} \implies \log v^k_s \le \log v^{k'}_{s'} . \end{aligned}
\tag{$\diamond$}
\label{eq:robustOEU_problem_reduced}
\end{equation*}

\clearpage 
\section{Implementation Details}
\label{appendix:implementation}

In order to calculate $e_*$ for each subject's data, we solve problem~\eqref{eq:robustOEU_problem_original} using \texttt{Matlab R2017b} (MathWorks). 

For each subject, the decision in every trial is characterized by a tuple $(a_1, a_2, x_1, x_2)$ where $a_i$ represents the intercept of the budget line on each axis (here we call the $x$-axis ``account~1'' and the $y$-axis ``account~2''), and $x_i$ represents the subject's allocation to account $i$. In order to rewrite the choice data in a price-consumption format as in the theory, we set prices $p_1 = 1$ (normalization) and $p_2 = a_1 / a_2$. This gives us a dataset $(x^k, p^k)_{k=1}^K$. 

Remember that the problem we are going to solve is: 
\begin{align}
\begin{aligned}
\min_{(v^k_s)_{k,s}} & \max_{k \in K, s, t \in S} (\log \mu_s^* + \log v^k_s - \log p^k_s) - (\log \mu_t^* + \log v^k_t - \log p^k_t) \\
\text{s.t.} & \;\; \phantom{(} x^k_s > x^{k'}_{s'} \implies \log v^k_s \le \log v^{k'}_{s'} . 
\end{aligned}
\tag{$\star$}
\end{align}
Our main task is to express this problem in a matrix notation. 

Let $\boldsymbol{z}$ be a vector of length $K \times S + K \times S + S$, whose first $K \times S$ entries correspond to each of $(\log v_s^k)_{s, k}$ and the last $K \times S + S$ entries are all~1. This vector corresponds to the control variables of the problem. The reason why we have $K \times S$ additional rows of~1 in the vector will become clear shortly. 

We construct two matrices $A$ and $B$. The first matrix $A$ has $K \times S$ rows and $K \times S + K \times S + S$ columns, and looks as follows: 
\begin{equation*}
\hspace{-8em}
\resizebox{1.2\textwidth}{!}{
\kbordermatrix{
   & \cdots & v_s^k & v_t^k & v_s^l & v_t^l & \cdots & & \cdots & p_s^k & p_t^k & p_s^l & p_t^l & \cdots & & \cdots & \mu_s^* & \mu_t^* & \cdots \\
\vdots &  & \vdots & \vdots & \vdots & \vdots & & \vrule & & \vdots & \vdots & \vdots & \vdots & & \vrule & & \vdots & \vdots & \\
(k,s,t) & \cdots & 1 & -1 & 0 & 0 & \cdots & \vrule & \cdots & - \log p_s^k & \log p_t^k & 0 & 0 & \cdots & \vrule & \cdots & 1 & -1 & \cdots \\
(k,t,s) & \cdots & -1 & 1 & 0 & 0 & \cdots & \vrule & \cdots & \log p_s^k & - \log p_t^k & 0 & 0 & \cdots & \vrule & \cdots & -1 & 1 & \cdots \\
(l,s,t) & \cdots & 0 & 0 & 1 & -1 & \cdots & \vrule & \cdots & 0 & 0 & - \log p_s^l & \log p_t^l & \cdots & \vrule & \cdots & 1 & -1 & \cdots \\
(l,t,s) & \cdots & 0 & 0 & -1 & 1 & \cdots & \vrule & \cdots & 0 & 0 & \log p_s^l & - \log p_t^l & \cdots & \vrule & \cdots & -1 & 1 & \cdots \\
\vdots &  & \vdots & \vdots & \vdots & \vdots & & \vrule & & \vdots & \vdots & \vdots & \vdots & & \vrule & & \vdots & \vdots & \\
} .
} 
\end{equation*}
Similarly, the second matrix $B$ has $K \times S + K \times S + S$ columns. There is one row for every pair $(k, s)$ and $(k', s')$ with $x_s^k > x_{s'}^{k'}$. In the row corresponding to $(k, s)$ and $(k', s')$ we have zeroes everywhere with the exception of a $-1$ in the column for $v_s^k$ and a $1$ in the column for $v_{s'}^{k'}$. 

We use the function \texttt{fmincon} to find a solution $\boldsymbol{z}^*$ and the value of the problem (i.e., $e_*$), with $\max A \cdot \boldsymbol{z}$ being the objective function we are going to minimize and $B \cdot \boldsymbol{z} \geq 0$ being the constraint.

\clearpage 
\section{Minimum Perturbation Test}
\label{appendix:minimum_perturbation_test}

\paragraph{Rationale behind the test.} 
We provide a detailed exposition of how we derive our test. 
Let $H_0$ and $H_1$ denote the null hypothesis that the true dataset $D_{\text{true}} = (p^k, x^k)_{k=1}^K$ is OEU rational and the alternative hypothesis that $D_{\text{true}}$ is not OEU rational. To construct our test, consider a number $\mathcal{E}^*$, which is the result of the following optimization problem given a dataset $D_{\text{true}}$: 
\begin{equation} 
\begin{aligned}
\min_{(v_s^k, \lambda^k, \varepsilon_s^k)_{s, k}} & \;\; \max_{k \in K, s, t \in S} \frac{\varepsilon_s^k}{\varepsilon_t^k} \\
\text{s.t.} & \;\; \log \mu_s^* + \log v_s^k - \log \lambda^k - \log p_s^k - \log \varepsilon_s^k = 0 \\
& \;\; x_s^k > x_{s'}^{k'} \Longrightarrow \log v_s^k \leq \log v_{s'}^{k'} . 
\end{aligned}
\label{eq:minimum_perturbation_test} 
\end{equation}
Under $H_0$, the true dataset $D_{\text{true}} = (p^k, x^k)_{k=1}^K$ is OEU rational. A slight modification of Lemma~7 in \citeSOMpapers{echenique2015som} implies that there exist strictly positive numbers $\widetilde{v}_s^k$, and $\widetilde{\lambda}^k$ for all $s \in S$ and $k \in K$ such that 
\begin{equation*}
\log \mu_s^* + \log \widetilde{v}_s^k - \log \widetilde{\lambda}^k - \log p_s^k = 0 
\;\; \text{ and } \;\; x_s^k > x_{s'}^{k'} \Longrightarrow \log \widetilde{v}_s^k \leq \log \widetilde{v}_{ts}^{k'} . 
\end{equation*}
Substituting the relationship $\tilde p_s^k = p_s^k \varepsilon_s^k$ for all $s \in S$ and $k \in K$ yields 
\begin{equation*}
\log \mu_s^* + \log \widetilde{v}_s^k - \log \widetilde{\lambda}^k -
\log \tilde  p_s^k = \log \varepsilon_s^k  
\;\; \text{ and } \;\; x_s^k > x_{s'}^{k'} \Longrightarrow \log \widetilde{v}_s^k \leq \log \widetilde{v}_{s'}^{k'} , 
\end{equation*}
which implies that the tuple $(\widetilde{v}_s^k, \widetilde{\lambda}^k, \varepsilon_s^k)_{s, k}$ satisfies the constraint in problem~\eqref{eq:minimum_perturbation_test}.

Letting $\mathcal{E}^* \left( (p^k, x^k)_{k=1}^K \right)$ denote the optimal value of the problem \eqref{eq:minimum_perturbation_test}, we have 
\begin{equation*}
\mathcal{E}^* \left( (p^k, x^k)_{k=1}^K \right) \leq \max_{k \in K, s, t \in S} \frac{\varepsilon_s^k}{\varepsilon_s^k} = \widehat{\mathcal{E}} 
\end{equation*}
under the null hypothesis. 

We construct a test as follows: 
\begin{equation*}
\begin{cases}
\text{reject } H_0 & \text{if } \displaystyle \int_{\mathcal{E}^* \left( (p^k, x^k)_{k=1}^K \right)}^\infty f_{\widehat{\mathcal{E}}} (z) dz < \alpha \\
\text{accept } H_0 & \text{otherwise}
\end{cases} , 
\end{equation*}
where $\alpha$ is the size of the test and $f_{\widehat{\mathcal{E}}}$ is the density function of the distribution of $\widehat{\mathcal{E}} = \max_{k, s, t} \varepsilon_s^k / \varepsilon_t^k$. 
Given a nominal size $\alpha$, we can find a critical value $C_\alpha$ satisfying $\Pr [\widehat{\mathcal{E}} > C_\alpha] = \alpha$; we set $C_\alpha = F^{-1}_{\widehat{\mathcal{E}}} (1-\alpha)$, where $F_{\widehat{\mathcal{E}}}$ denotes the cumulative distribution function of $\widehat{\mathcal{E}}$. 
However, because $\mathcal{E}^* \left( ( p^k, x^k)_{k=1}^K \right) \leq \widehat{\mathcal{E}}$, the true size of the test is better than $\alpha$. Concretely, $\text{size} = \Pr [\mathcal{E}^* > C_\alpha] \leq \Pr [\widehat{\mathcal{E}} > C_\alpha] = \alpha$.

\paragraph{Parameter tuning.}
In order to perform the test, we need to obtain the distribution of $\widehat{\mathcal{E}}$ and its critical value $C_\alpha$ given a significance level $\alpha$. We obtain the distribution of $\widehat{\mathcal{E}}$ by assuming that $\varepsilon$ follows a log-normal distribution $\varepsilon \sim \Lambda (\nu, \xi^2)$.~\footnote{Note that parameters $(\nu, \xi^2)$ correspond to the mean and the variance of the random variable in the log-scale. In other words, $\log \varepsilon \sim N (\nu, \xi^2)$. The moments of the log-normal distribution $\varepsilon \sim \Lambda (\nu, \xi^2)$ are then calculated by $\mathbf{E} [\varepsilon] = \exp (\nu + \xi^2 / 2)$ and $\Var (\varepsilon) = \exp (2 \nu + \xi^2) (\exp (\xi^2) - 1)$.} 

The crucial step in our approach is the selection of parameters $(\nu, \xi^2)$. It is natural to choose these parameters so that there is no price perturbation on average (i.e., $\mathbf{E} [\varepsilon] = 1$). However, as we discussed above, there is no objective guide to choosing an appropriate level of $\Var (\varepsilon)$. Therefore, we use variation in (relative) prices observed in the data. 

We have assumed that $\tilde p_s^k = p_s^k \varepsilon_s^k$ for all $s \in S$, $k \in K$, and the noise term $\varepsilon$ is independent of the random selection of budgets $(p_s^k)_{s \in S, k \in K}$. Hence, 
\begin{equation*}
\begin{aligned}
& \Var (\tilde p) = \Var (p) \cdot \Var (\varepsilon) + \Var (p) \cdot \mathbf{E} [\varepsilon]^2 + \mathbf{E} [p]^2 \cdot \Var (\varepsilon) \\
\iff \;\; & \frac{\Var (\tilde p)}{\Var (p)} = \mathbf{E} [\varepsilon]^2 + \left( 1 + \frac{\mathbf{E} [p]^2}{\Var (p)} \right) \Var (\varepsilon) . 
\end{aligned}
\end{equation*}
Given the observed variation in $(p_s^k)_{s \in S, k \in K}$, $\Var (\varepsilon)$ determines how much larger (or smaller, in ratio) the variation of perturbed prices $(\tilde p_s^k)_{s \in S, k \in K}$ is relative to actual prices. 

Let us consider an agent who has trouble telling the two variances apart. More generally, the agent has trouble telling the distributions of prices apart, that is why she is confusing actual and perceived prices, but the distribution depends only on the variance; so we focus on variance. Consider a hypothesis test for the null hypothesis that the variance of a normal random variable with known mean has variance $\sigma_0^2$ against the alternative that $\sigma^2 \geq \sigma_0^2$. Let $\hat \sigma_n^2$ be the sample variance.  

The agent performs an upper-tailed chi-squared test defined as 
\begin{align}
\text{H}_0 :& \;\; \sigma^2 = \sigma_0^2 \notag \\
\text{H}_1 :
& \;\; \sigma^2 > \sigma_0^2 \notag 
\end{align}
The test statistic is: 
\begin{equation*}
T_n = \frac{(n-1) \hat \sigma_n^2}{\sigma_0^2} 
\end{equation*}
where $n$ is the sample size (i.e., the number of budget sets). The sampling distribution of the test statistic $T_n$ under the null hypothesis follows a chi-squared distribution with $n-1$ degrees of freedom. 

We consider the probability $\eta^I$ of rejecting the null hypothesis when it is true, a type~I error; and the probability $\eta^{\mathit{II}}$ of failing to reject the null hypothesis when the alternative $\sa^2= \sa^2_1 > \sigma^2_0$ is true, a type~II
error. The test rejects the null hypothesis that the variance is $\sigma_0^2$ if 
\begin{equation*}
T_n > \chi_{1-\alpha, n-1}^2 \notag 
\end{equation*}
where $\chi_{1-\alpha, n-1}^2$ is the critical value of a chi-squared distribution with $n-1$ degree of freedom at the significance level $\alpha$, defined by $\pr [\chi^2 < \chi_{1-\alpha, n-1}^2] = 1 - \eta^I$.~\footnote{An alternative approach, without assuming that a distribution for $T_n$, and based on a large sample approximation to the distribution of $T_n$, yields very similar results. Calculations and empirical findings are available from the authors upon request.} 

Under the alternative hypothesis that $\sigma^2 = \sigma^2_1>\sigma^2_0$, the statistic $(\sigma_0^2/\sigma_1^2) \cdot T_n$ follows a chi-squared distribution (with $n-1$ degrees of freedom). Then, the probability $\eta^{\mathit{II}}$ of making a type~II error is given by 
\begin{equation*}
\begin{aligned}
\eta^{\mathit{II}} 
= \pr [T_n < \chi_{1-\alpha, n-1}^2 \mid \text{H}_1 : \sigma_1^2 > \sigma_0^2 \text{ is true}] 
&= \pr \left[ \frac{\sigma_0^2}{\sigma_1^2} \cdot T_n < \frac{\sigma_0^2}{\sigma_1^2} \cdot \chi_{1-\alpha, n-1}^2 \right] \\
&= \pr \left[ \chi^2 < \frac{\sigma_0^2}{\sigma_1^2} \cdot \chi_{1-\alpha, n-1}^2 \right] . 
\end{aligned}
\end{equation*}

Let $\chi_{\beta, n-1}^2$ be the value that satisfies $\pr [\chi^2 < \chi_{\beta, n-1}^2] = \eta^{\mathit{II}}$. 
Then, given $\eta^I$ and $\eta^{\mathit{II}}$, we obtain 
\begin{equation*}
\begin{aligned}
\pr \left[ \chi^2 < \frac{\sigma_0^2}{\sigma_1^2} \cdot \chi_{1-\alpha, n-1}^2 \right] = \eta^{\mathit{II}} 
& \iff 
\frac{\sigma_0^2}{\sigma_1^2} \cdot \chi_{1-\alpha, n-1}^2 = \chi_{\beta, n-1}^2 \\
& \iff 
\frac{\sigma_1^2}{\sigma_0^2} = \frac{\chi_{1-\alpha, n-1}^2}{\chi_{\beta, n-1}^2} . 
\end{aligned}
\end{equation*}

As a consequence, given a measured variance $\sigma_0^2$, calculated from observed prices, and assumed values for $\eta^{I}$ and $\eta^{\mathit{II}}$, we can back out the minimum ``detectable'' value of the variance $\sigma^2_1$. From this variance of prices, we obtain $\Var (\ep)$.

\clearpage 
\section{Supplementary Empirical Analysis}
\label{appendix:supplementary_analysis}

\subsection{First-Order Stochastic Dominance} 
\label{appendix:fosd}

In the portfolio allocation environment studied in the three studies we looked at, choosing an allocation $(x_1, x_2)$ from a budget line defined by prices $(p_1, p_2)$ violates {\em monotonicity with respect to first-order stochastic dominance (FOSD-monotonicity)} when either (i) $p_1 > p_2$ and $x_1 > x_2$ or (ii) $p_2 > p_1$ and $x_2 > x_1$ (i.e., the choice involves more allocation toward more-expensive security).

Table~\ref{table:fosd_violation} presents the average fraction (out of~25) of choices violating FOSD-monotonicity and the number of subjects without FOSD-monotonicity violations. On average, subjects made 24-34\% violations of FOSD-monotonicity. The number of subjects who made no FOSD-violating choices is less than 10\% for all datasets. As discussed in \citeSOMpapers{choi2014som}, choices can be consistent with GARP even with violations of FOSD-monotonicity. The average fraction of FOSD-violating choices calculated from the subsample of GARP-compliant ($\text{CCEI} = 1$) subjects is close to the one we obtain from the whole sample. 
The entire distributions are presented in Figure~\ref{fig:fosd-mon_violation_frac}. 

\begin{table}[!h]
\centering 
\caption{FOSD violation.}
\label{table:fosd_violation}
\resizebox{\textwidth}{!}{%
\begin{tabular}{l rrr rrr}
\toprule 
 & \multicolumn{3}{c}{All subjects} & \multicolumn{3}{c}{$\text{CCEI = 1}$} \\
 \cmidrule(lr){2-4} \cmidrule(lr){5-7}
 & CKMS & CMW & CS & CKMS & CMW & CS \\
\midrule 
Number of subjects & 1,182 & 1,116 & 1,421 & 270 & 207 & 313 \\
Average fraction of FOSD-mon. violations & 0.335 & 0.320 & 0.239 & 0.364 & 0.312 & 0.221 \\
Fraction of subjects without FOSD-mon. violations & 0.025 & 0.047 & 0.066 & 0.066 & 0.164 & 0.153 \\
\bottomrule 
\end{tabular}%
}
\end{table}

\begin{figure}[!h]
\centering 
\includegraphics[width=\textwidth]{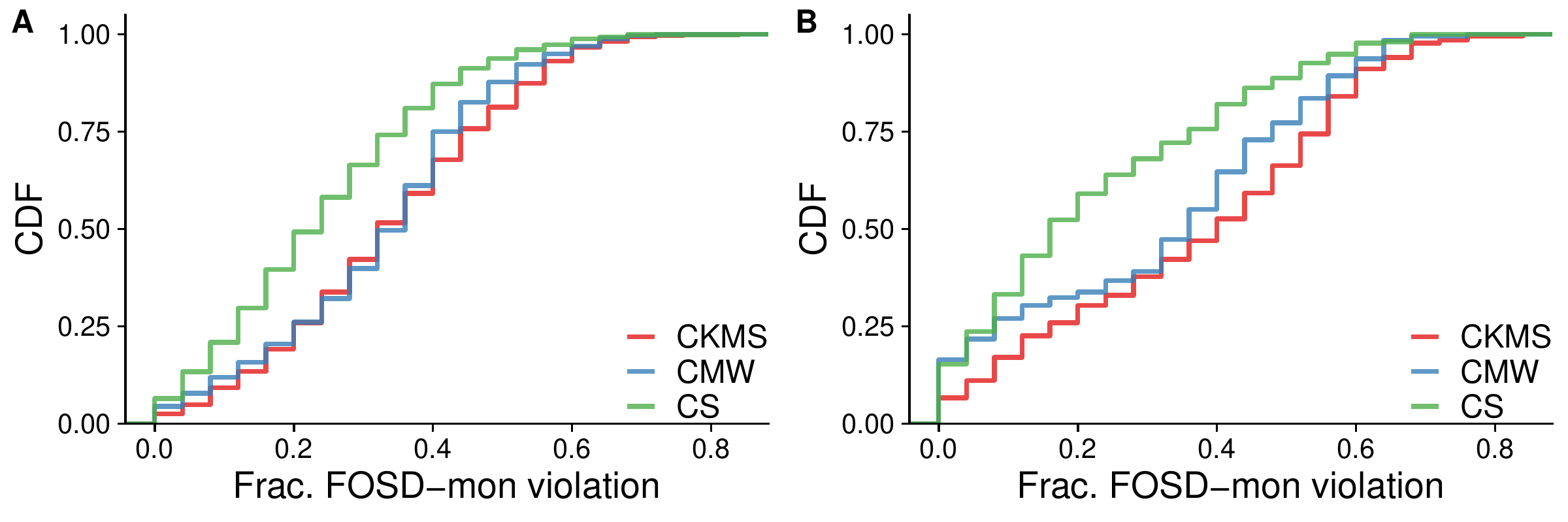} 
\caption{Empirical CDFs of fraction of choices that violate FOSD-monotonicity. (A) All subjects. (B) Subjects with $\text{CCEI} = 1$.}
\label{fig:fosd-mon_violation_frac}
\end{figure}

\clearpage
\subsection{Choices on the 45-Degree Line} 
\label{appendix:diagonal}

In the experiments, subjects made choices of allocations $(x_1, x_2)$ by clicking on the budget line graphically presented on the screen. 
Note that points on the 45-degree line correspond to equal allocations between the two accounts ($x_1 = x_2$) and therefore involve no risk (i.e., the 45-degree line is the ``full insurance'' line). If a subject's all choices are on the 45-degree line (call such pattern {\em diagonal allocations}), we can rationalize the data with EU and hence $e_* = 0$. 

It is, however, extremely difficult (or almost impossible) to choose the point ``exacctly'' on the 45-degree line in practice. Actual choices subjects made may be slightly off from the 45-degree line, and it can generate large $e_*$ (through violations of the downward-sloping demand) while CCEI and EU-CCEI stay close to~1 (see Figure~\ref{fig:choice_pattern_ccei_1}, panel~D). 
In this section, we examine how much of the disagreement between $e_*$ and CCEI or EU-CCEI are driven by small deviations from the diagonal allocations. 

To this end, we first re-define diagonal allocations. Instead of requiring all choices to be exactly on he 45-degree line, we call a data {\em almost diagonal allocations} if all choices are inside small balls (with fixed radius~$r$) drawn around the intersections of budget lines and the 45-degree line. We can control the size of acceptable deviations by changing the radius~$r$ of the ball. 
The idea is shown in Figure~\ref{fig:almost_diagonal_illustration}. In this example, chosen allocations (black dots) are not exactly on the 45-degree line, but they are inside the balls around the diagonal allocations (red circles).\footnote{These choices also violate FOSD-monotonicity. We would expect relatively large $e_*$ from this choice pattern, but its CCEI is~1 because it satisfies GARP.}  

\begin{figure}[!h]
\centering
\begin{tikzpicture}
\draw[-,dashed,thick] (0,0) -- (4.6,4.6);
\draw[-,blue,very thick] (0,4.5) -- (4,0);
\draw[magenta,draw opacity=0.5,fill=magenta,fill opacity=0.5] (36/17,36/17) circle (0.25cm);
\draw (37.5/17,2.0184) node[black,fill,circle,scale=0.3] {};
\draw[-,blue,very thick] (0,2.4) -- (4.8,0);
\draw[magenta,draw opacity=0.5,fill=magenta,fill opacity=0.5] (1.6,1.6) circle (0.25cm);
\draw (1.5,1.65) node[black,fill,circle,scale=0.3] {};
\draw[->,very thick] (0,0) -- (5,0) node[below right] {$x_1$};
\draw[->,very thick] (0,0) -- (0,5) node[above left] {$x_2$};
\end{tikzpicture}
\caption{Almost diagonal allocations.}
\label{fig:almost_diagonal_illustration}
\end{figure}
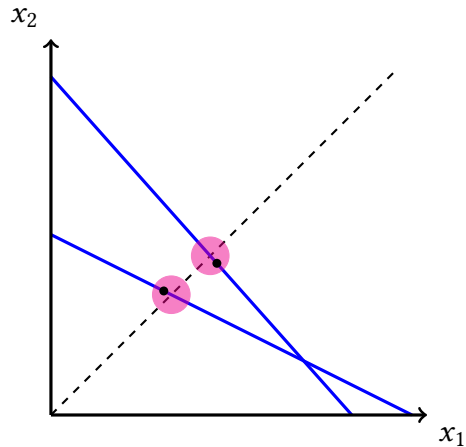

\clearpage 
Table~\ref{table:almost_diagonal} shows the fraction of subjects who made almost diagonal allocations (in all 25 questions) under different sizes of~$r$. Between 6\% and 12\% of subjects made such choice pattern when the radius is set to $r = 1$. 

\begin{table}[!h]
\centering 
\caption{Fraction of subjects who made almost diagonal allocations.}
\label{table:almost_diagonal}
\begin{tabular}{lrrrrr}
\toprule
 & & \multicolumn{4}{c}{Radius of the ball ($r$)} \\
\cmidrule(lr){3-6}
Study & $N$ & $0.05$ & $0.20$ & $0.50$ & $1.00$ \\ 
\midrule 
CKMS & 1182 & 0.000 & 0.000 & 0.035 & 0.083 \\ 
CMW & 1116 & 0.008 & 0.040 & 0.098 & 0.120 \\ 
CS & 1421 & 0.005 & 0.023 & 0.048 & 0.060 \\ 
\bottomrule 
\end{tabular}
\end{table}

Figures~\ref{fig:almost_diagonal_e-ccei} and~\ref{fig:almost_diagonal_e-euccei} below show the relationship between $e_*$ and CCEI as well as EU-CCEI, as in Figure~\ref{fig:minimal_e_vs_ccei} (Section~\ref{section:oeu_application_results}). Bottom panels in each figure focus on subjects who made almost diagonal allocations (the radius of the ball is set to $r = 1$) in all 25 questions, and top panels present the rest of the subjects. 

Bottom panels in each figure confirm that almost diagonal allocations yield values of CCEI and EU-CCEI that are close to~1. The same subjects have dispersed values of $e_*$, including the highest value in each experiment. 

It does not meant that the disagreement between $e_*$ and CCEI-based measures come mainly from slight deviations from the diagonal allocations. Top panels in each figure show that there are choice patterns, other than almost diagonal allocations, that have CCEI/EU-CCEI $\approx 1$ and large $e_*$.

\clearpage 
\begin{figure}[!th]
\centering 
\includegraphics[width=\textwidth]{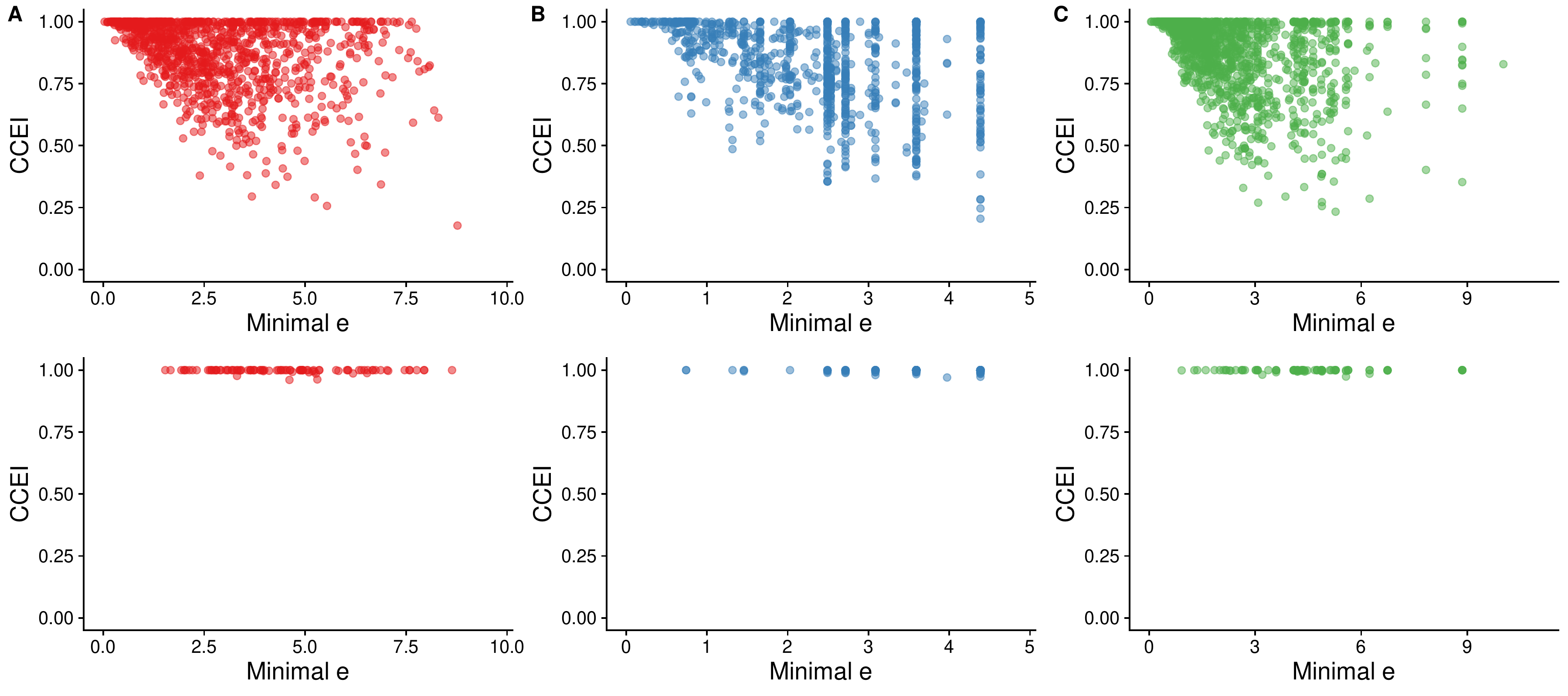} 
\caption{Correlation between $e_*$ and CCEI. Top panels show subjects who did not choose almost diagonal allocations and bottom panels show those who selected almost diagonal allocations (with $r = 1$). Panels: (A) CKMS, (B) CMW, (C) CS.}
\label{fig:almost_diagonal_e-ccei}
\end{figure}

\begin{figure}[!th]
\centering 
\includegraphics[width=\textwidth]{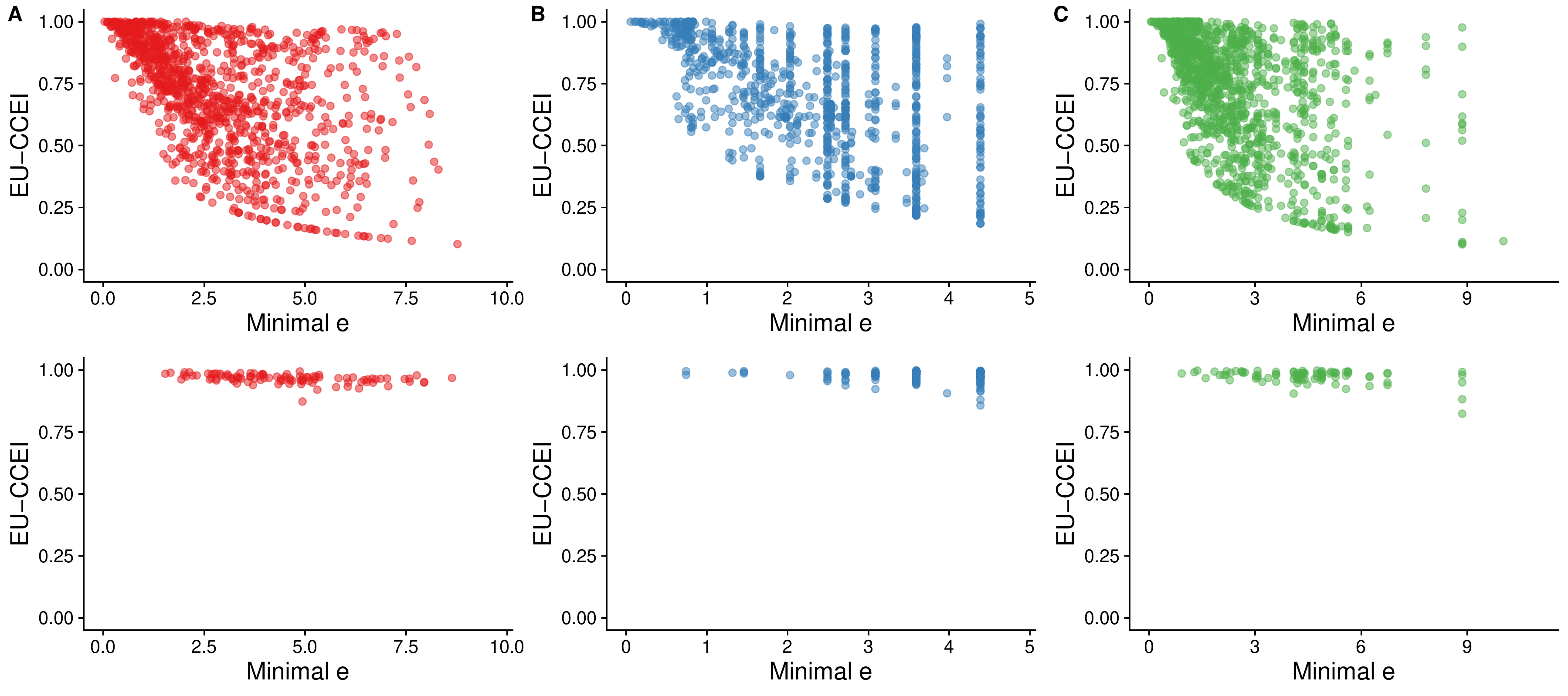} 
\caption{Correlation between $e_*$ and EU-CCEI. Top panels show subjects who did not choose almost diagonal allocations and bottom panels show those who selected almost diagonal allocations (with $r = 1$). Panels: (A) CKMS, (B) CMW, (C) CS.}
\label{fig:almost_diagonal_e-euccei}
\end{figure}

\clearpage
\subsection{Sensitivity} 
\label{appendix:sensitivity} 

As is clear from the definition, our measure $e_*$ is a bound that has to hold across all observations and states (see conditions~\eqref{eq:oeubound},~\eqref{eq:oeubound:price}, and~\eqref{eq:oeubound:utility} in the definitions of $e$-perturbed OEU in Section~\ref{section:robustoeu}). It is possible that a couple of ``bad'' choices significantly influence the measure. This section presents several robustness checks for the main empirical result. 

\paragraph{Dropping critical mistakes.} 
In this robustness check, we recalculate $e_*$ using subsets of observed choices that exclude outliers. More precisely, for each subject, we calculate $e_*$ for all combinations of $25 - m$ choices and pick the smallest $e_*$. We do this for $m = 1, 2$. 

By construction, dropping critical mistakes shifts the distribution of the measure (Figure~\ref{fig:minimal_e_cdf_drop}). However, it does not dramatically change the correlational patterns between $e_*$ and CCEI (Figure~\ref{fig:minimal_e_vs_ccei_drop}) nor between $e_*$ and demographic characteristics (Figures~\ref{fig:minimal_e_demographics_drop1} and~\ref{fig:minimal_e_demographics_drop2}). In this sense, the main empirical results are robust to the presence of small number of bad choices. 

\begin{figure}[h]
\centering 
\includegraphics[width=\textwidth]{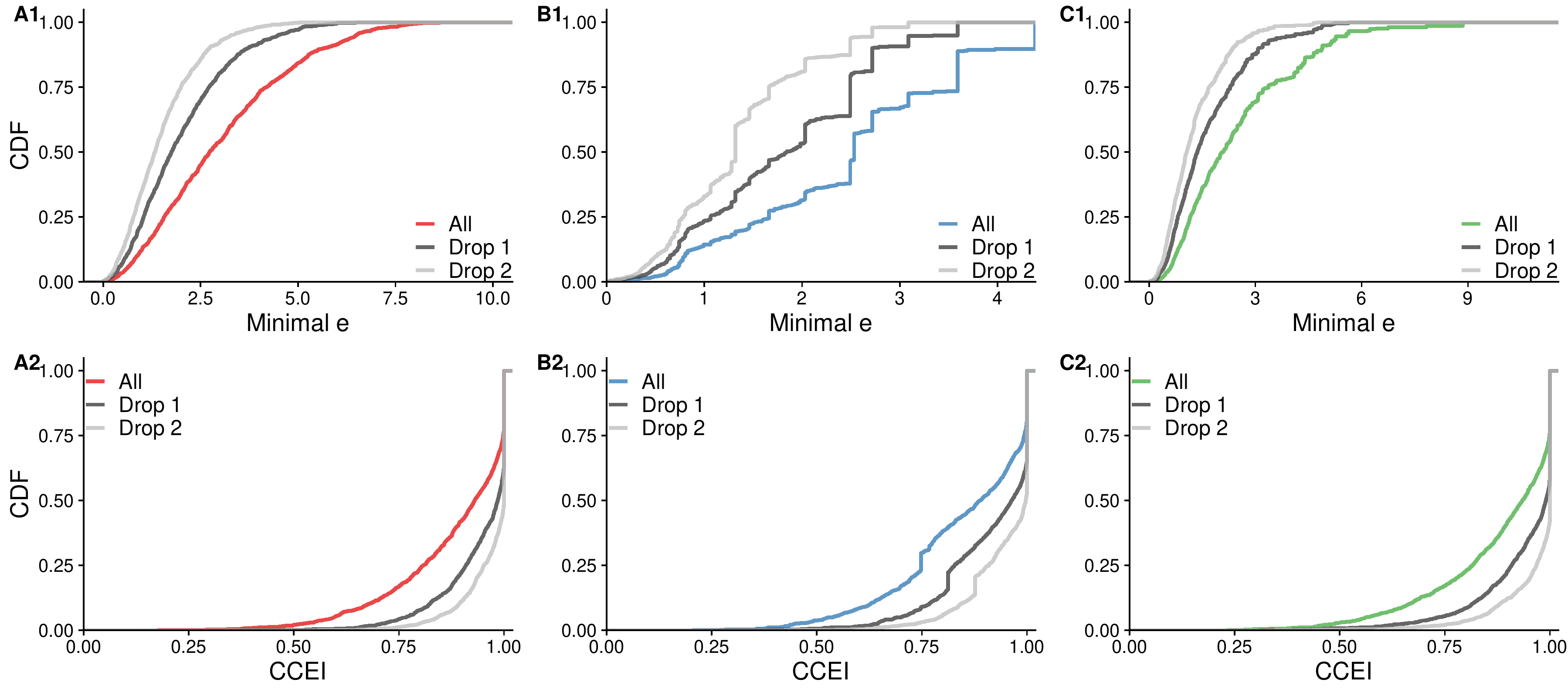} 
\caption{Empirical CDFs of $e_*$ and CCEI, using all observations or subsets of observations dropping one or two critical mistakes. Panels: (A) CKMS, (B) CMW, (C) CS.}
\label{fig:minimal_e_cdf_drop}
\end{figure}

\begin{figure}[p]
\centering 
\includegraphics[width=\textwidth]{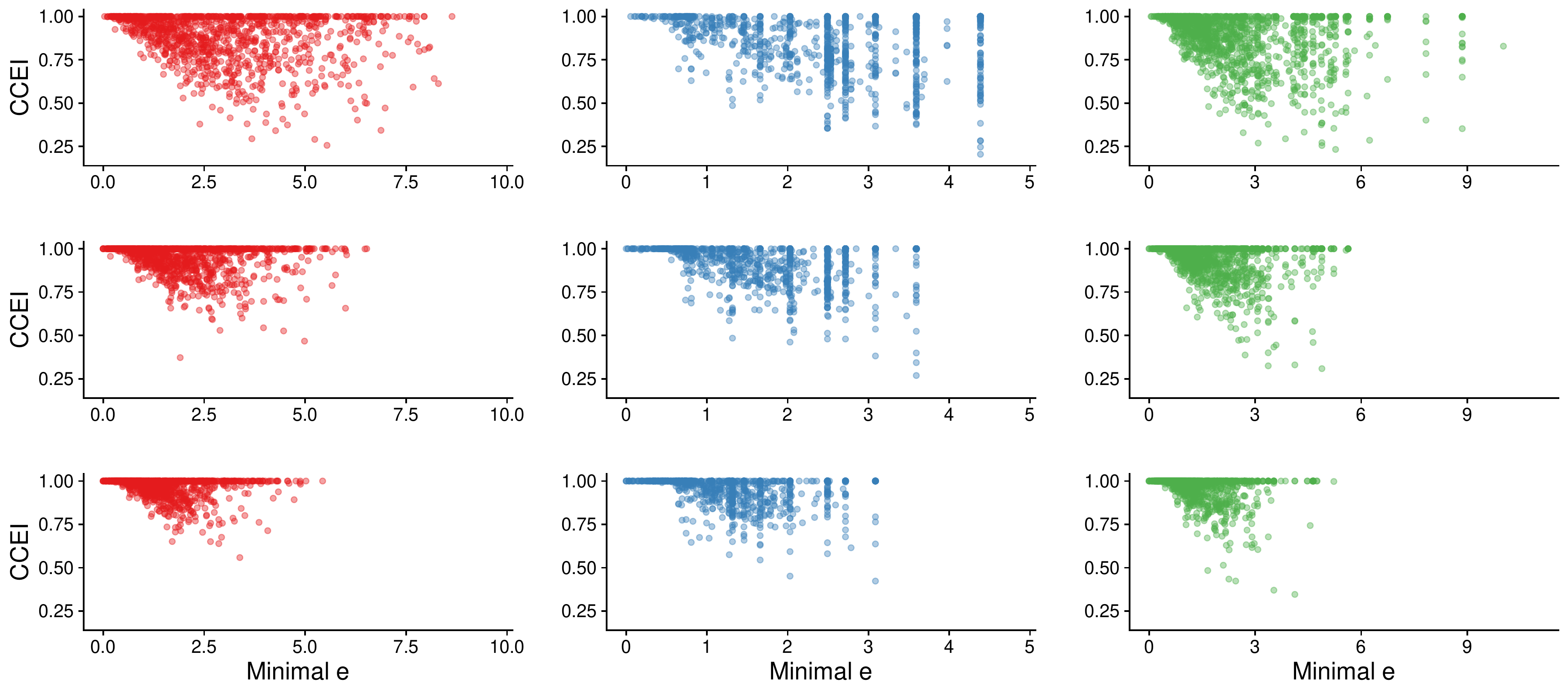} 
\caption{Correlation between $e_?$ and CCEI. (Top panels) All 25 observations. (Middle panels) Drop one critical mistake. (Bottom panels) Drop two critical mistakes.}
\label{fig:minimal_e_vs_ccei_drop}
\end{figure}

\clearpage 
\begin{figure}[p]
\centering 
\includegraphics[width=\textwidth]{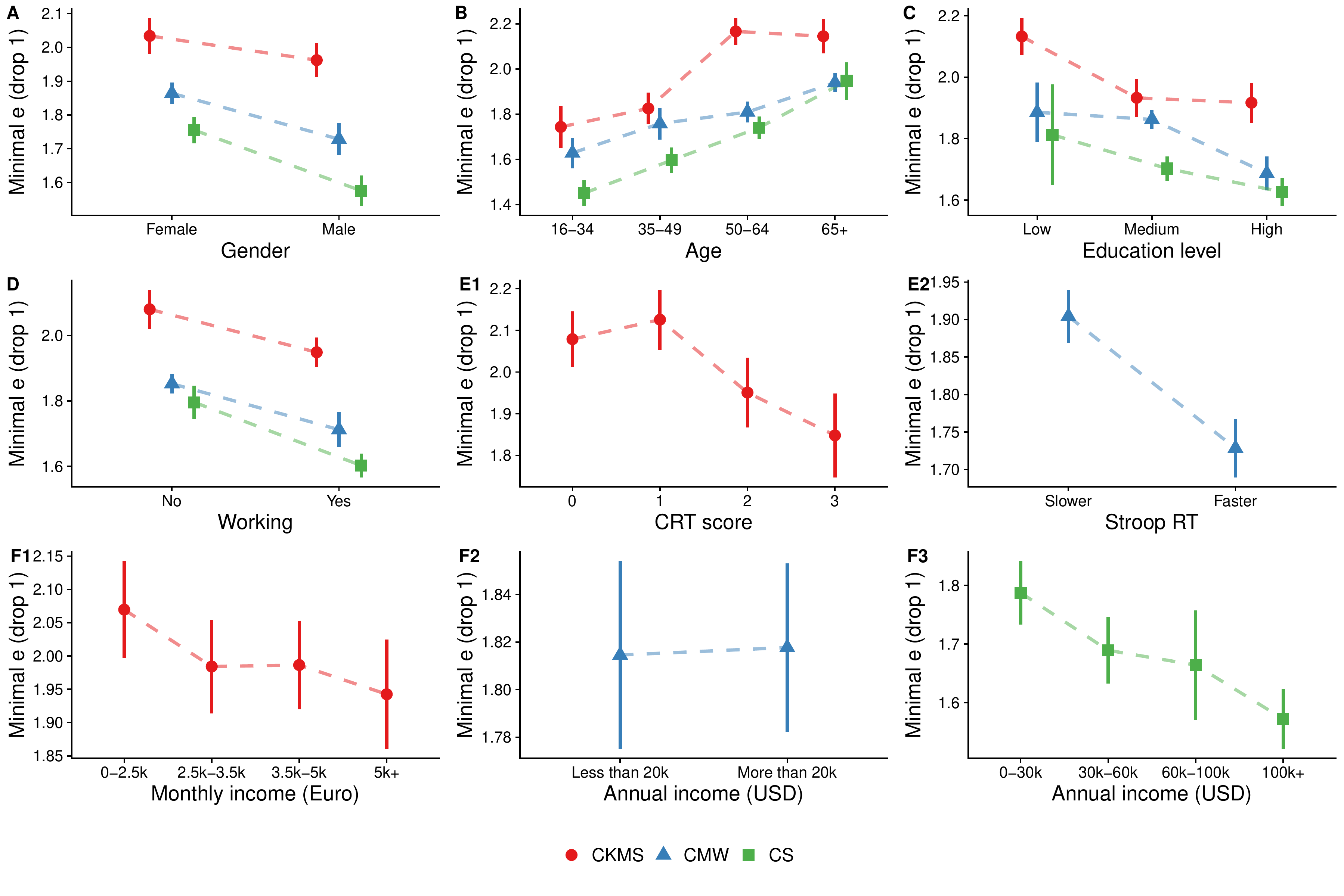} 
\caption{Robustness of demographic correlations in Figure~\ref{fig:minimal_e_demographics}. For each subject, $e_*$ is recalculated after dropping one critical mistake.}
\label{fig:minimal_e_demographics_drop1}
\end{figure}

\begin{figure}[p]
\centering 
\includegraphics[width=\textwidth]{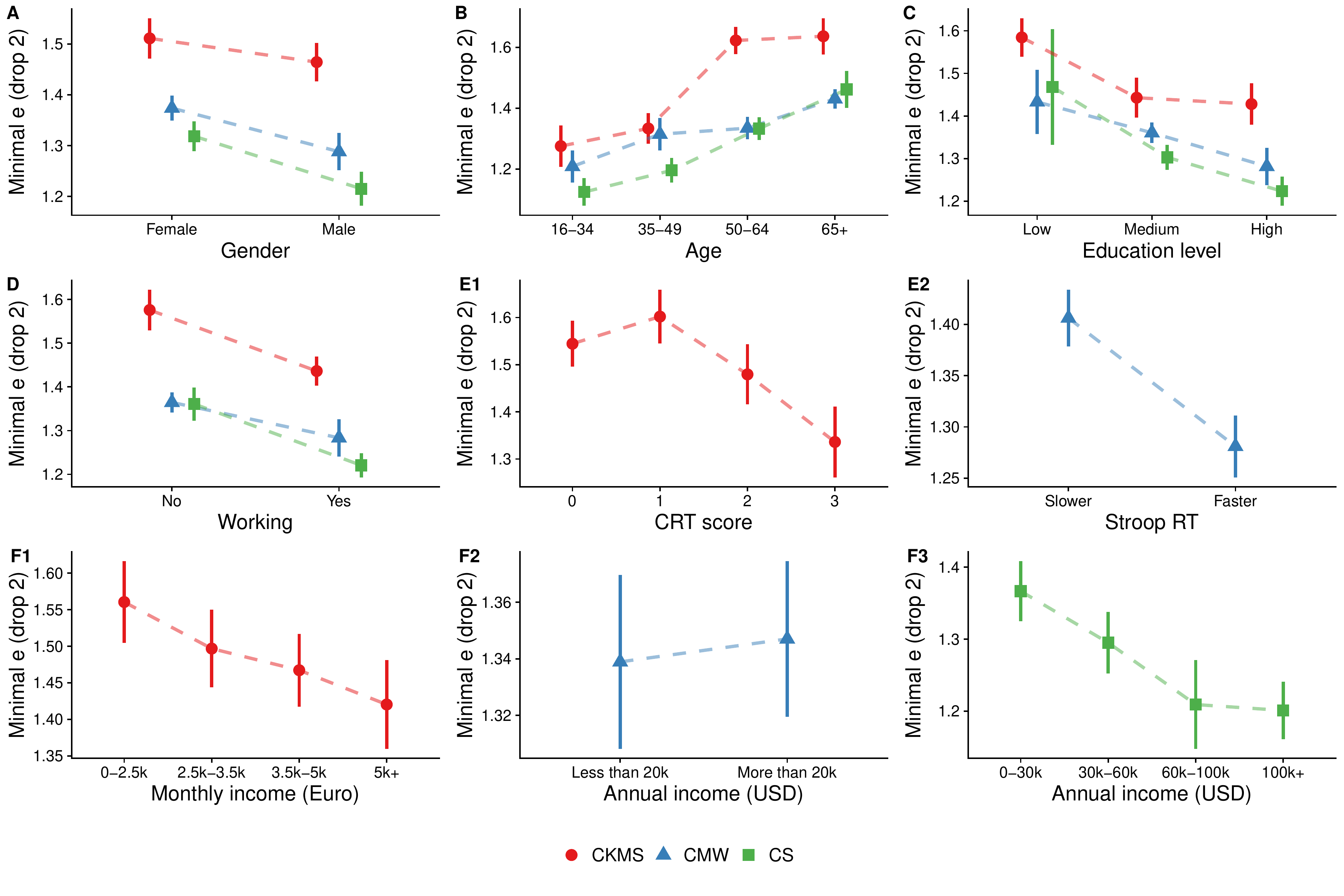} 
\caption{Robustness of demographic correlations in Figure~\ref{fig:minimal_e_demographics}. For each subject, $e_*$ is recalculated after dropping two critical mistakes.}
\label{fig:minimal_e_demographics_drop2}
\end{figure}

\clearpage 
\paragraph{``Average'' perturbation.} 
Let $\bar{e}$ be the solution to the following minimization problem: 

\begin{equation*}
\begin{aligned}
\min_{(\varepsilon_s^k)_{s,k}} &\;\; \sum_{k \in K} \sum_{s \in S} \frac{\left| \log \varepsilon_s^k \right|}{K S} \\
\text{s.t.} &\;\; (x^k, q^k)_{k=1}^K \text{ is OEU rational} \\
&\;\; q_s^k = p_s^k \varepsilon_s^k \text{ for each } s \in S, k \in K 
\end{aligned}
\end{equation*}

The idea behind this alternative measure is simple. As in the case of $e$-price-perturbed utility, we search for sets of multiplicative noises $(\varepsilon_s^k)_{s,k}$ which could rationalize the observed data. Instead of looking at the uniform bound $\max_{s, t, k} (\log{\varepsilon_s^k} - \log{\varepsilon_t^k})$ and minimizing it, we take the average of these perturbations and minimize it. A similar idea was applied to quantify the distance from several models of time preferences in \citeSOMpapers{echenique2016som}. 

Figure~\ref{fig:corr_ptb_e_ccei} presents the relationship between $\bar{e}$, $e_*$, and CCEI. Figure~\ref{fig:ptb_demographics} shows the correlation between $\bar{e}$ and demographic variables. 
These figures do not show correlational patterns that are markedly different from those presented in the main empirical results (Figures~\ref{fig:minimal_e_vs_ccei} and~\ref{fig:minimal_e_demographics} in Section~\ref{section:oeu_application_results}). 

\begin{figure}[!h]
\centering 
\includegraphics[width=\textwidth]{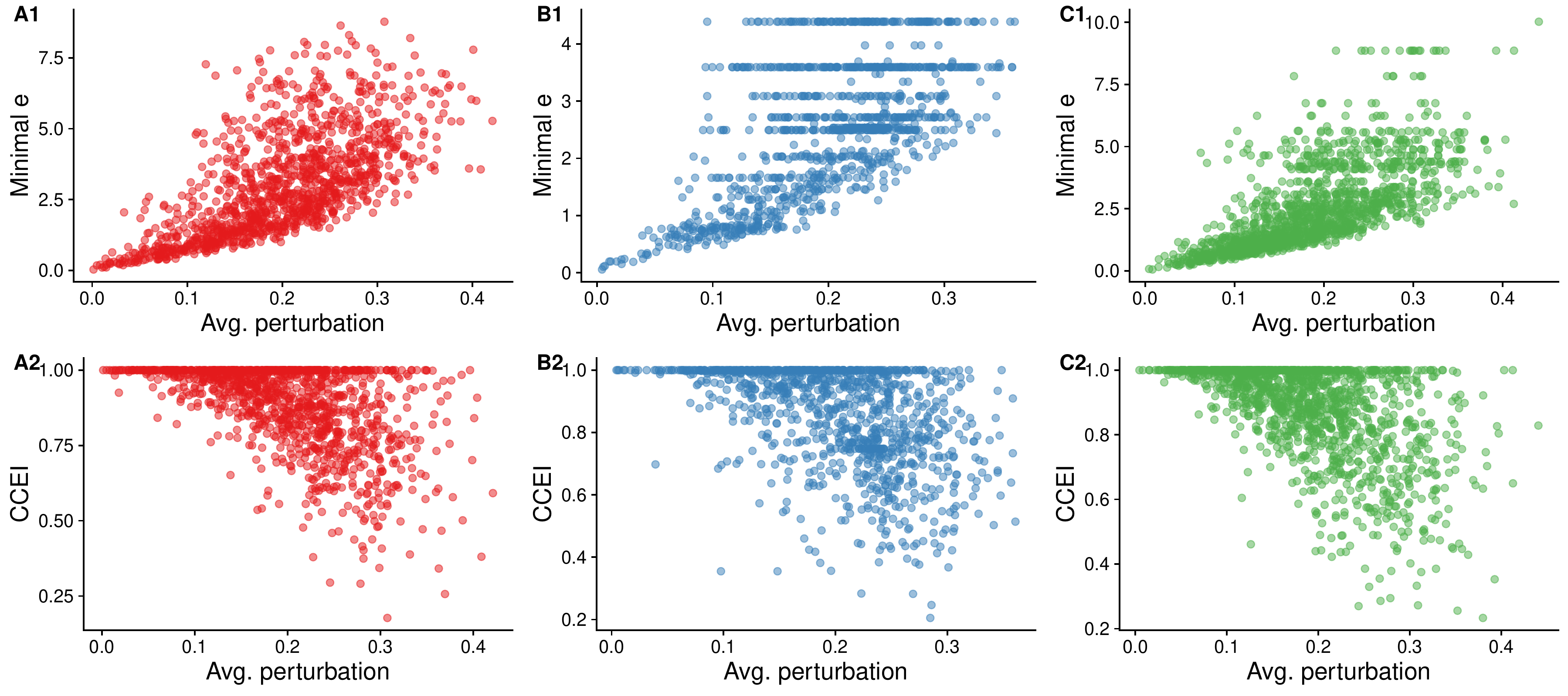} 
\caption{Correlation between $\bar{e}$ and $e_*$ (top panels) and $\bar{e}$ and CCEI (bottom panels). Panels: (A) CKMS, (B) CMW, (C) CS.}
\label{fig:corr_ptb_e_ccei}
\end{figure}

\begin{figure}[!p]
\centering 
\includegraphics[width=\textwidth]{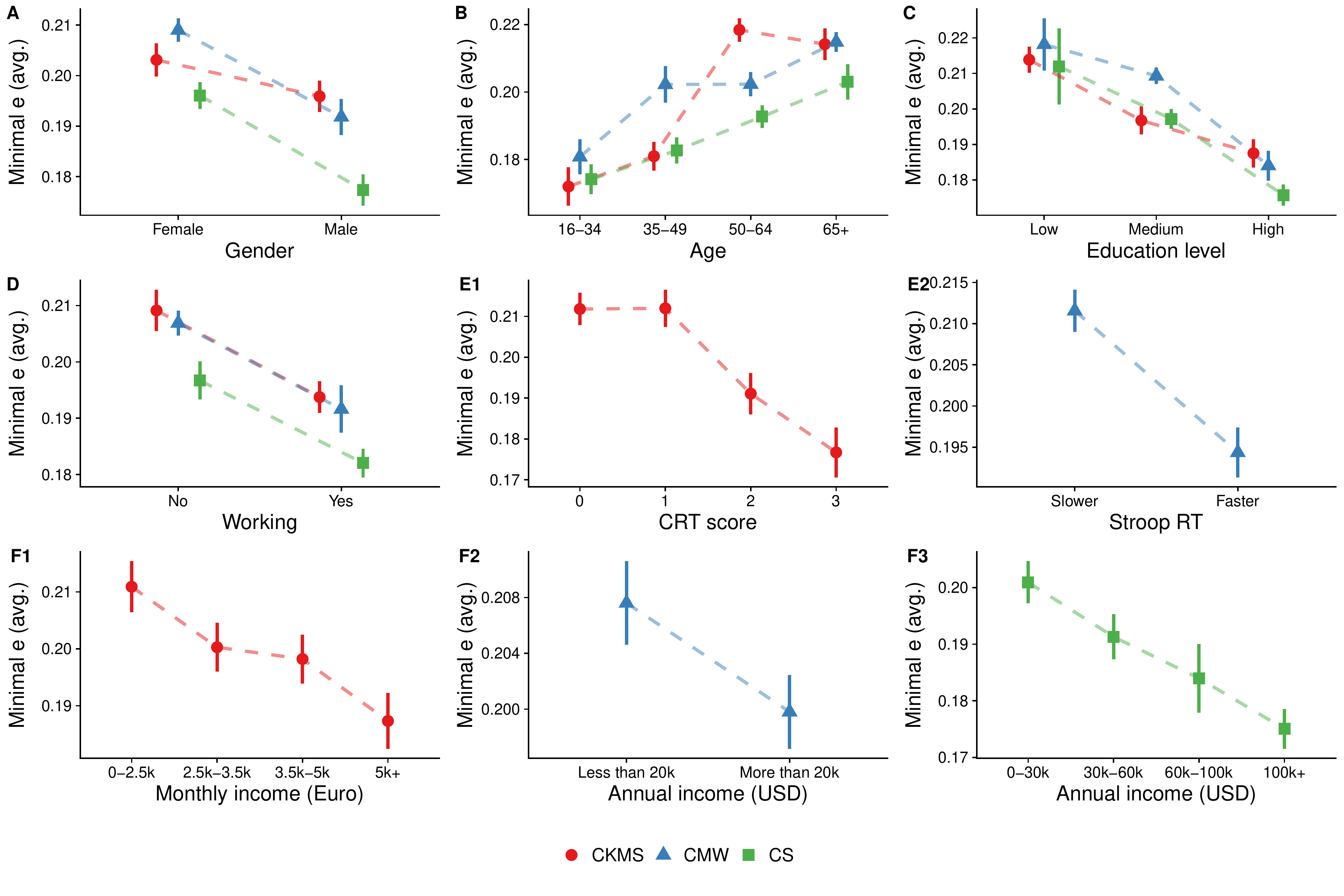} 
\caption{$\bar{e}$ and demographic variables.}
\label{fig:ptb_demographics}
\end{figure}

\clearpage 
\subsection{Properties of $e_*$} 
\label{appendix:supplementary_analysis_property}

\paragraph{$e_*$ from observed and simulated choices.} 
The statistical approach described in Section~\ref{sec:implementation_perturbation} is one way to assess ``how big'' the observed $e_*$'s are. Another way is to simulate choice data assuming some behavioral model and calculate $e_*$ on the simulated dataset. 
Following \citeSOMpapers{bronars1987}, we randomly select an allocation from each budget line. 
Since subjects in CKMS and CS faced a randomly selected set of budgets, we first randomly select one set of budgets (from the observed sets of budgets) and then randomly choose allocations on these budgets. We then calculate $e_*$, as well as CCEI, using the simulated choices. 
We repeat this 10,000 times for each of the three datasets. 

Figure~\ref{fig:obs-sim_e-ccei_cdf} compares the observed and simulated $e_*$. The distribution of observed $e_*$ locates left of simulated $e_*$ (all differences are statistically significant, according to two-sample Kolmogorov-Smirnov test). 
The actual subjects' behavior is thus closer to OEU rationality compared to completely random behavior (even though complete random is unrestrictive and may not be the best benchmark). 

Figure~\ref{fig:obs-sim_e-ccei_corr} looks at the correlation between $e_*$ and CCEI and compares the pattern in observed and simulated datasets (panels A-C in the top row are same as Figure~\ref{fig:minimal_e_vs_ccei}). 

\clearpage 
\begin{figure}[p]
\centering 
\includegraphics[width=\textwidth]{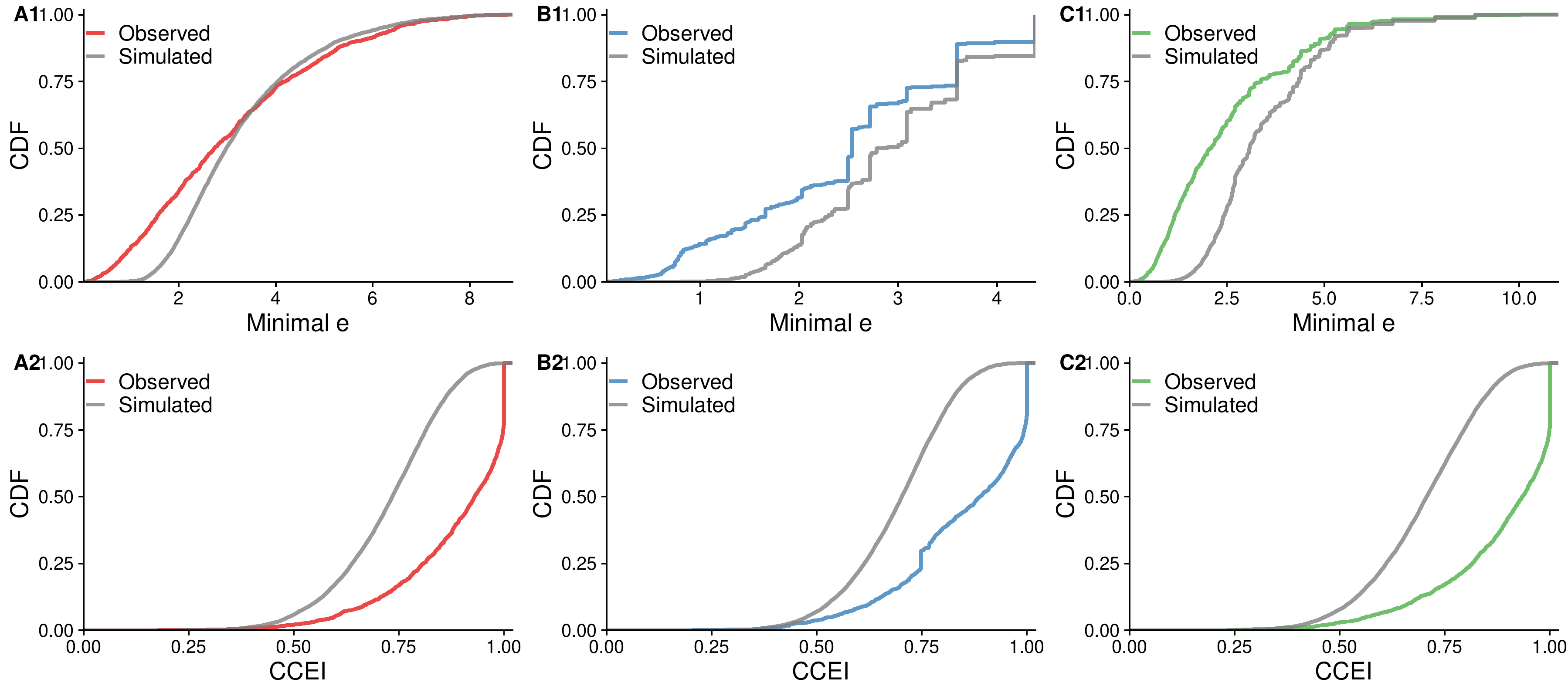}
\caption{Comparison between observed and simulated $e_*$ (top panels) and CCEI (bottom panels). Panels: (A) CKMS, (B) CMW, (C) CS.}
\label{fig:obs-sim_e-ccei_cdf}
\end{figure}

\begin{figure}[p]
\centering 
\includegraphics[width=\textwidth]{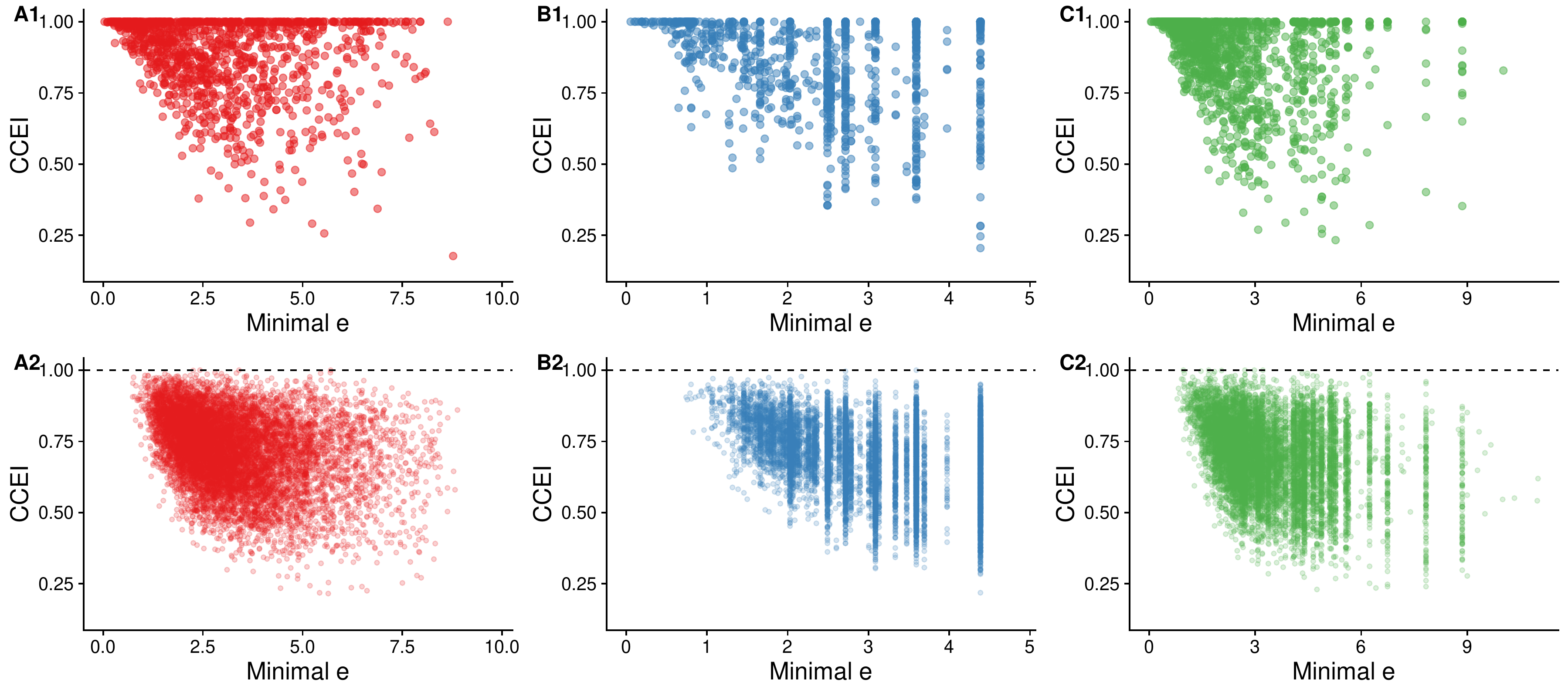}
\caption{Comparison between observed (top panels) and simulated (bottom panels) $e_*$ and CCEI. Panels: (A) CKMS, (B) CMW, (C) CS. {\em Notes}: Top panels are identical to those in Figure~\ref{fig:minimal_e_vs_ccei}.}
\label{fig:obs-sim_e-ccei_corr}
\end{figure}

\clearpage 
\paragraph{Bound of $e_*$.} 

The value of $e_*$ depends on the structure of the budgets an agent faces. 
In particular, it is clear from $e$-PSAROEU that $1 + e_*$ is bounded by the maximum ratio of risk-neutral prices: 
  \[ 1 + e_* \leq \max_{k, \in K, s, t \in S} \frac{\rho^{k}_{s}}{\rho^{k}_{t}} . \]
Since CKMS, CMW, and CS experiments all used two equally-likely states, the ratio of risk-neutral prices is equal to the ratio of prices. 
Figure~\ref{fig:bound_e-ccei} shows the observed $e_*$ and (participant-specific) upper bound. 
(Since all subjects faced the same set of budgets in the CMW study, there is only one vertical line.) 
About 13\% of the subjects ($475/3719$ in merged data; $221/1182$ in CKMS; $114/1116$ in CMW; $140/1421$ in CS) have their $e_*$ exactly at the upper bound. 

\begin{figure}[h]
\centering 
\includegraphics[width=\textwidth]{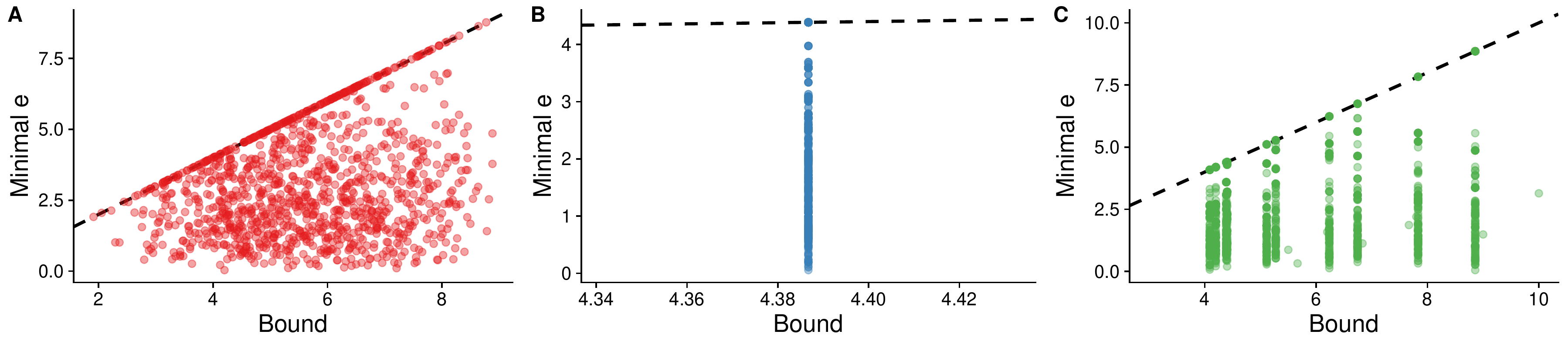} 
\caption{Bound of $e_*$. The $x$-axis in each plot is the upper bound of $e_*$, given by $\max_{k, s, t} p_s^k / p_t^k - 1$. {\em Notes}: There is no variation in bounds in the CMW data (panel~B) since all subjects faced the same set of budgets. In the CS data (panel~C), the $x$-axis is cut at~10 for better visualization. There are~22 additional observations in the data with the bounds ranging from~11 to~48.} 
\label{fig:bound_e-ccei}
\end{figure}

\clearpage
\subsection{Illustration of $e$-Perturbed OEU} 
\label{appendix:example_e-price-perturbation} 

In Figure~\ref{fig:choice_pattern_ccei_1}, we present typical choice patterns from selected subjects with $\text{CCEI} = 1$ and varying degrees of~$e_*$. Panels~A-F plot observed choices and panels~a-f plot the relationship between $\log (x_2 / x_1)$ and $\log (p_2 / p_1)$, which shows how much the dataset conforms to the downward-sloping demand. The measure $e_*$, roughly speaking, captures the degree of deviation from the downward-sloping demand. 

Consider an observed dataset $(x^k, p^k)_{k=1}^K$ and a perturbed dataset $(x^k, \tilde{p}^k)_{k=1}^K$, where $\tilde{p}_s^k = p_s^k \varepsilon_s^k$ and $\varepsilon_s^k \geq 0$ for all $s \in S$ and $k \in K$. 
Since we fix the chosen bundle $(x^k)_{k=1}^K$ and rotate the budget lines around them, price perturbation ``moves'' points in panels~a-f horizontally. 

To make the dataset $e$-price-perturbed OEU rational (Definition~\ref{def:e-price-perturbed_oeu}), we need to move the points horizontally so that they satisfy the downward-sloping demand. 
Note that the horizontal distance for each observation $k$ is given by 
\begin{equation*}
\begin{aligned}
\log \left( \frac{\tilde{p}_2^k}{\tilde{p}_1^k} \right) - \log \left( \frac{p_2^k}{p_1^k} \right) 
= \log \left( \frac{\tilde{p}_2^k / p_2^k}{\tilde{p}_1^k / p_1^k} \right) 
= \log \left( \frac{\varepsilon_2^k}{\varepsilon_1^k} \right) . 
\end{aligned}
\end{equation*}
We thus need to look at the maximal horizontal adjustment among observations, and the measure $e_*$ is obtained by minimizing it. 

Figure~\ref{fig:e-price-perturbation_example} shows the idea behind calculation of $e_*$ using price perturbation. It plots the same six subjects as in Figure~\ref{fig:choice_pattern_ccei_1}. 
In panels~A-F, red dotted lines represent the original budgets and blue solid lines represent perturbed budgets. 
In panels~a-f, green circles represent the original dataset and blue triangles represent the perturbed dataset. Red arrows connect points that correspond to the maximal adjustment. 
The figure shows that $e_*$-perturbed datasets satisfy the downward-sloping demand.\footnote{Perturbed dataset in each panel is based on one particular set of $(\varepsilon_s^k)_{s \in S, k \in K}$ returned by \texttt{Matlab}. There are small deviations from the downward-sloping demand (e.g., in panels~C and~E), but it is possible to correct for these numerical deviations without influencing the value of $e_*$.} 

We can draw several observations about the practical aspect of $e_*$. 
First, observe that the ``cheapest'' way for correcting choices violating FOSD-monotonicity is to perturb budgets corresponding to these observations so that $\tilde{p}_1^k = \tilde{p}_2^k$. 
Second, the figure provides an intuitive explanation of why $e_*$ can be large for choice patterns like panel~D. Since clicking on the point exactly on the 45-degree line is a challenging task, choices would scatter around the 45-degree line, occasionally falling in the region of FOSD-monotonicity. No matter how small these deviations from the 45-degree line are, $e$-price perturbation requires horizontal adjustments to achieve the downward-sloping demand. If the necessary adjustment is applied on a relatively extreme budget line, $e_*$ for such a subject can be very high.

\begin{figure}[!hp]
\centering 
\includegraphics[width=0.85\textwidth]{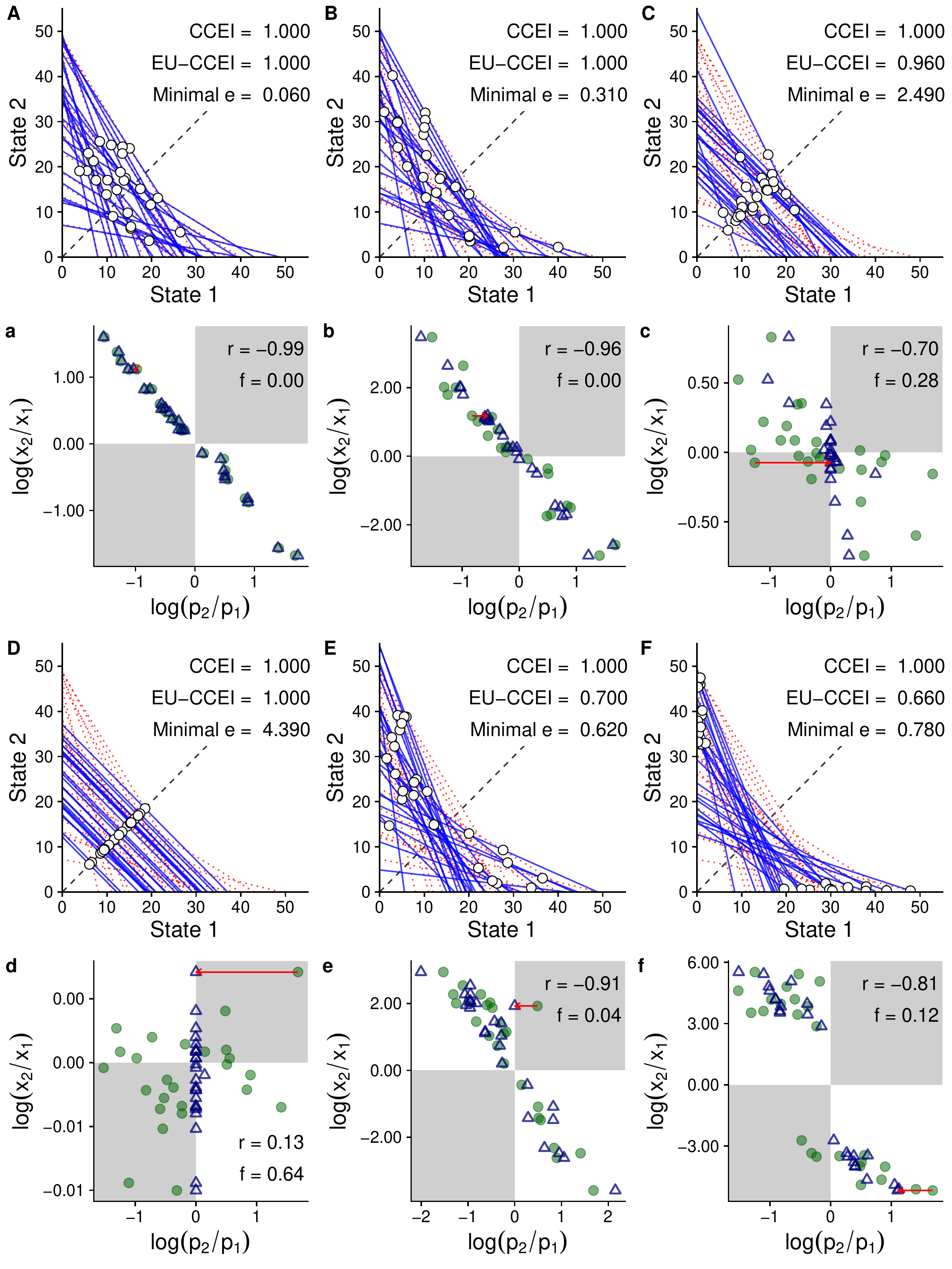}
\caption{Illustration of $e$-price-perturbed OEU rationalization. (A-F) Perturbed budgets (blue solid lines) and the original budgets (red dotted lines). (a-f) The relation between $\log (x_2 / x_1)$ and $\log (p_2 / p_1)$ (green circles), and $\log (x_2 / x_1)$ and $\log (\tilde{p}_2 / \tilde{p}_1)$ (blue triangles). Red arrows indicate observations requiring the largest adjustment.} 
\label{fig:e-price-perturbation_example}
\end{figure}

\clearpage
\subsection{Comparing Measures} 
\label{appendix:comparing_measures}

We calculate CCEI at which a subject is consistent with a given model, stochastically monotone utility maximization \citep{nishimura2017comprehensive}, EU, and concave EU, using the GRID method developed in \citeSOMpapers{polissonquahrenou17som}.\footnote{A stohastically monotone utility function gives strictly higher utility to bundle $x$ compared to another bundle $y$ if $x$ first-order stochastically dominates $y$ and gives them the same utility if two bundles are stochastically equivalent. In the environment we consider (two states with equally likely objective probabilities), a utility function is stochastically monotone if and only if it is symmetric and strictly increasing. 

\citeSOMpapers{choi2014som} also discuss a similar idea. They propose additional measure, which jointly captures the extent of GARP violations and violations of stochastic dominance, by combining the observed data and its ``mirror-image''. More precisely, they assume that if an allocation $(x_1, x_2)$ is chosen under the budget constraint $p_1 x_1 + p_2 x_2 = 1$, then $(x_2, x_1)$ would have been chosen under the mirror-image budget constraint $p_2 x_1 + p_1 x_2 = 1$. They then re-calculate CCEI for the ``combined'' data consisting of~50 ($25 \text{ budgets} \times 2$) choices.} 
We call these measures F-GARP, EU-CCEI, and cEU-CCEI. 
For a given dataset, the measures are ordered as 
\[ \text{cEU-CCEI} \leq \text{EU-CCEI} \leq \text{F-GARP} \leq \text{CCEI} , \]
since models we look at are nested in this order. 
Note that \citeSOMpapers{polissonquahrenou17som} calculated and reported CCEI, F-GARP, EU-CCEI, and cEU-CCEI for the CKMS dataset but not for the CMW and the CS datasets. 

Figures~\ref{fig:compare_measures_ckms}-\ref{fig:compare_measures_cs} compare $e_*$, CCEI, and these three additional measures.\footnote{We did not compute cEU-CCEI for 23 subjects (8 in CMW, and 15 in CS) since the code spent significantly long computation time. (\citeSOMpapers{polissonquahrenou17som} used a high-performance computing facility.) We also treated cEU-CCEI for six subjects in CS as missing values, since the code incorrectly returned $\text{cEU-CCEI} = 0$. Note that F-GARP and EU-CCEI for these 29 subjects are included in Figures~\ref{fig:compare_measures_ckms}-\ref{fig:compare_measures_cs}.} 
Panels on the diagonal show the distribution of each measure. Pairwise scatter plots are presented below diagonal, and their Spearman's correlation coefficients are shown above the diagonal (all $p < 0.001$; uncorrected for multiple comparison). 

The first column in each figure shows the relationship between $e_*$ and other measures. The second and the fourth panels in this column ($e_*$ vs. CCEI and $e_*$ vs. EU-CCEI) are identical to those presented in Figure~\ref{fig:minimal_e_vs_ccei}. As we discussed in Section~\ref{section:oeu_application_results} of the paper, we see that there are a significant number of subjects whose CCEI and EU-CCEI are close to one but their $e_*$'s are widely dispersed and further away from zero. 

This observation is not specific to CCEI and EU-CCEI. In the third and the fifth panels of the same column, we can see a similar pattern between $e_*$ and F-GARP as well as $e_*$ and EU-CCEI. The pattern is a general feature that distinguishes the idea behind the measures: $e_*$ is based on rotating budget lines while the other measures, which are all variants of CCEI, are based on shrinking budget sets.

\clearpage 
\begin{figure}[ht!]
\centering 
\includegraphics[width=\textwidth]{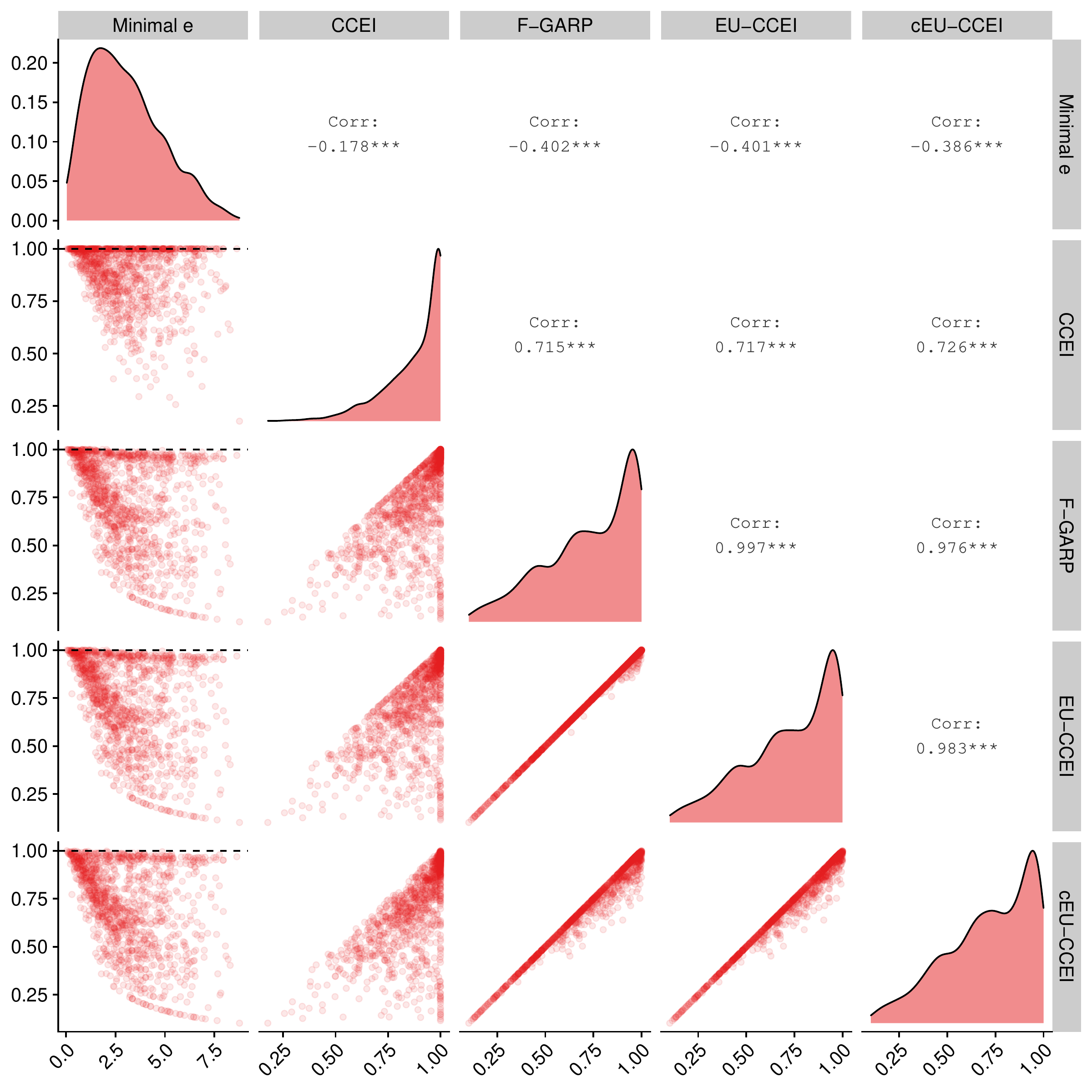} 
\caption{Comparing measures of rationality in the CKMS data.}
\label{fig:compare_measures_ckms}
\end{figure}

\begin{figure}[ht!]
\centering 
\includegraphics[width=\textwidth]{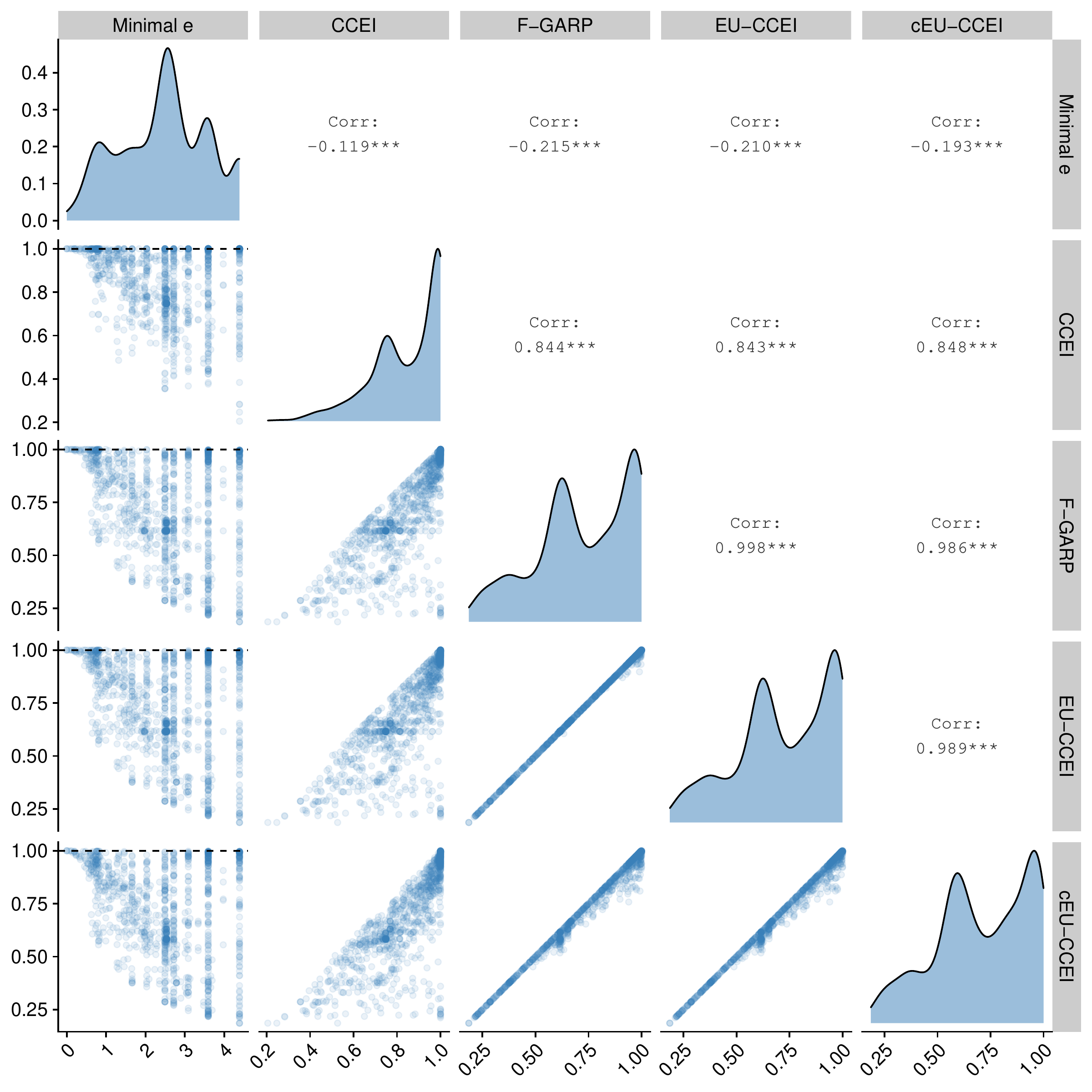} 
\caption{Comparing measures of rationality in the CMW data.}
\label{fig:compare_measures_cmw}
\end{figure}

\begin{figure}[ht!]
\centering 
\includegraphics[width=\textwidth]{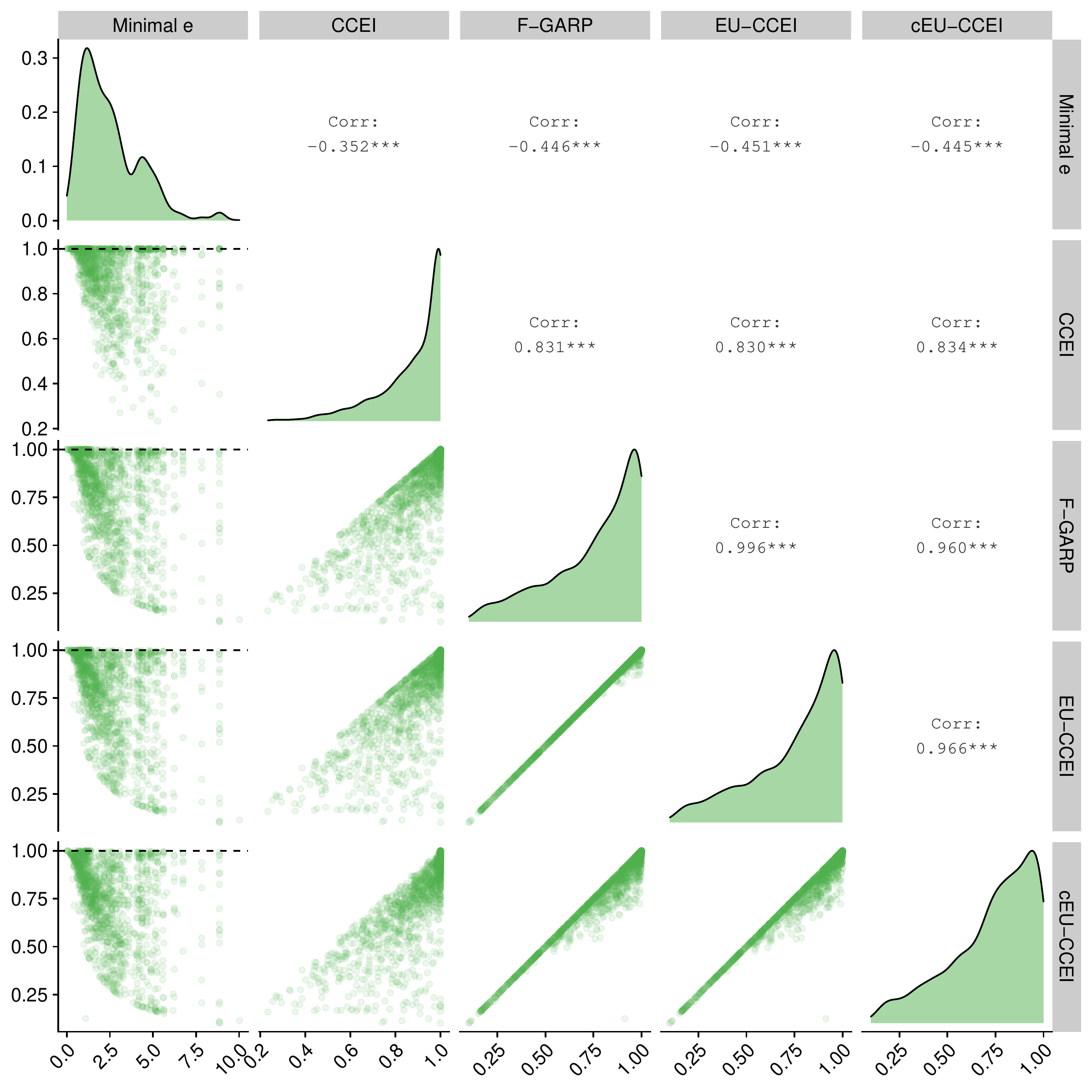} 
\caption{Comparing measures of rationality in the CS data.}
\label{fig:compare_measures_cs}
\end{figure}

\clearpage
\subsection{Choice Pattern: Additional Examples} 
\label{appendix:more_choice_pattern}

Choice data from four subjects presented in Section~\ref{section:oeu_application_results}, Figure~\ref{fig:choice_pattern_ccei_1}, are not meant to be representative of the entire dataset consisting of more than 3,000 subjects. 
In this section, we present more examples to understand the similarity and differences between $e_*$, CCEI, and EU-CCEI. 

We pick subjects from the CMW experiment, where all the subjects faced with the same set of 25 budget lines. This feature of the design makes the variation of $e_*$ smaller than in the other datasets (we observe several ``jumps'' in the empirical CDF of $e_*$ in Figure~\ref{fig:minimal_e_cdf}), but the comparison across choice patterns becomes easier. 

Figure~\ref{fig:cmw_map_sample_subjects} is the scatterplot of $e_*$ and EU-CCEI in the CMW data. Dashed lines represent the 25th, 50th, and 75th percentiles of $e_*$ and EU-CCEI. Two shaded areas represent combinations of $e_*$ and EU-CCEI that ``disagree'', in the sense that one measure says the subject is close to EU (relative to the median subject) but the other measure says the same subject far from EU (again, relative to the median subject). 
Each subject's choice pattern is shown below. 

\begin{figure}[ht!]
\centering 
\includegraphics[width=\textwidth]{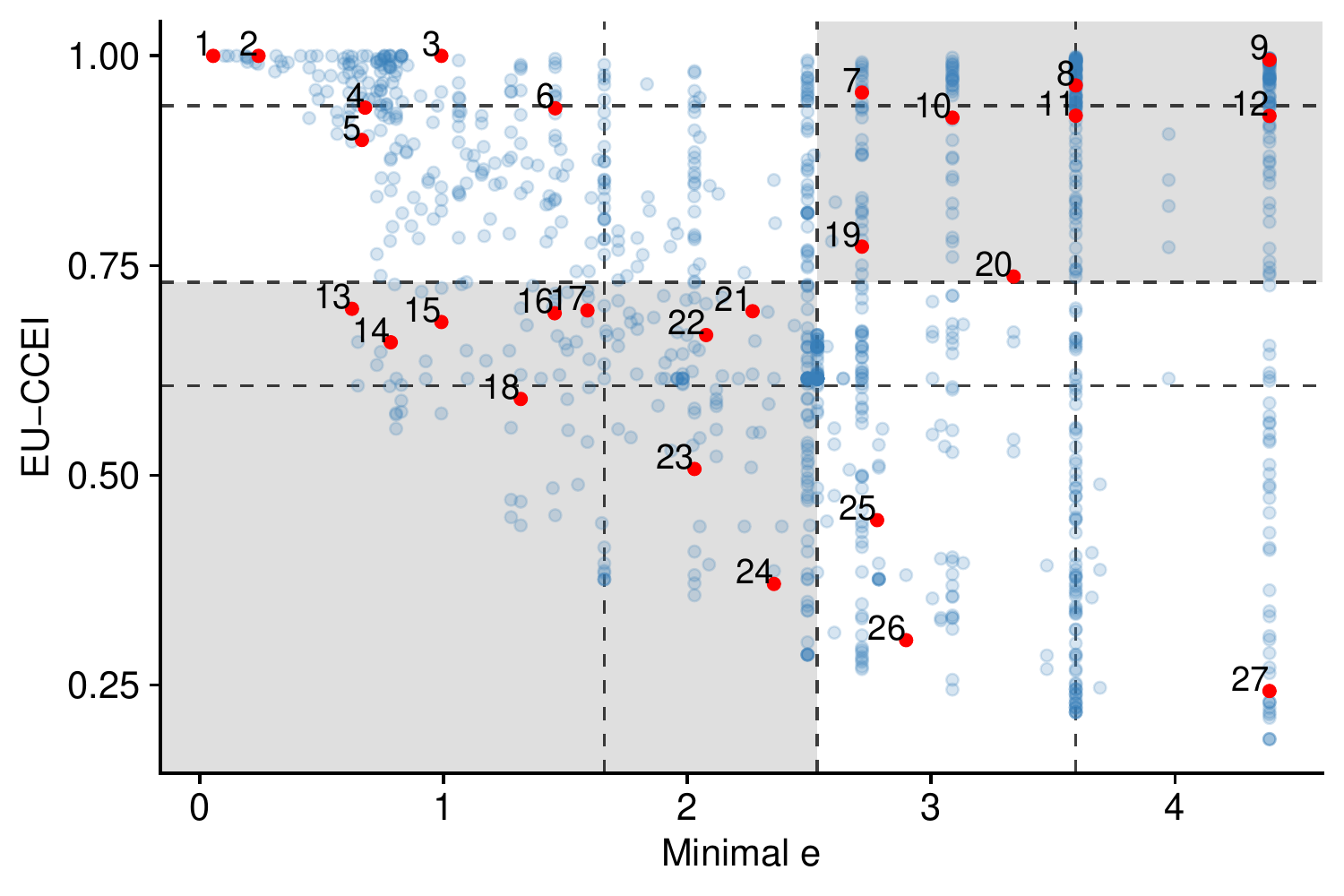} 
\caption{$e_*$ and EU-CCEI in CMW data. {\em Notes}: Vertical dashed lines represent the 25th, 50th, and 75th percentiles of $e_*$. Horizontal dashed lines represent the 25th, 50th, and 75th percentiles of EU-CCEI.}
\label{fig:cmw_map_sample_subjects}
\end{figure}


\begin{figure}[ht!]
\centering 
\includegraphics[width=0.9\textwidth]{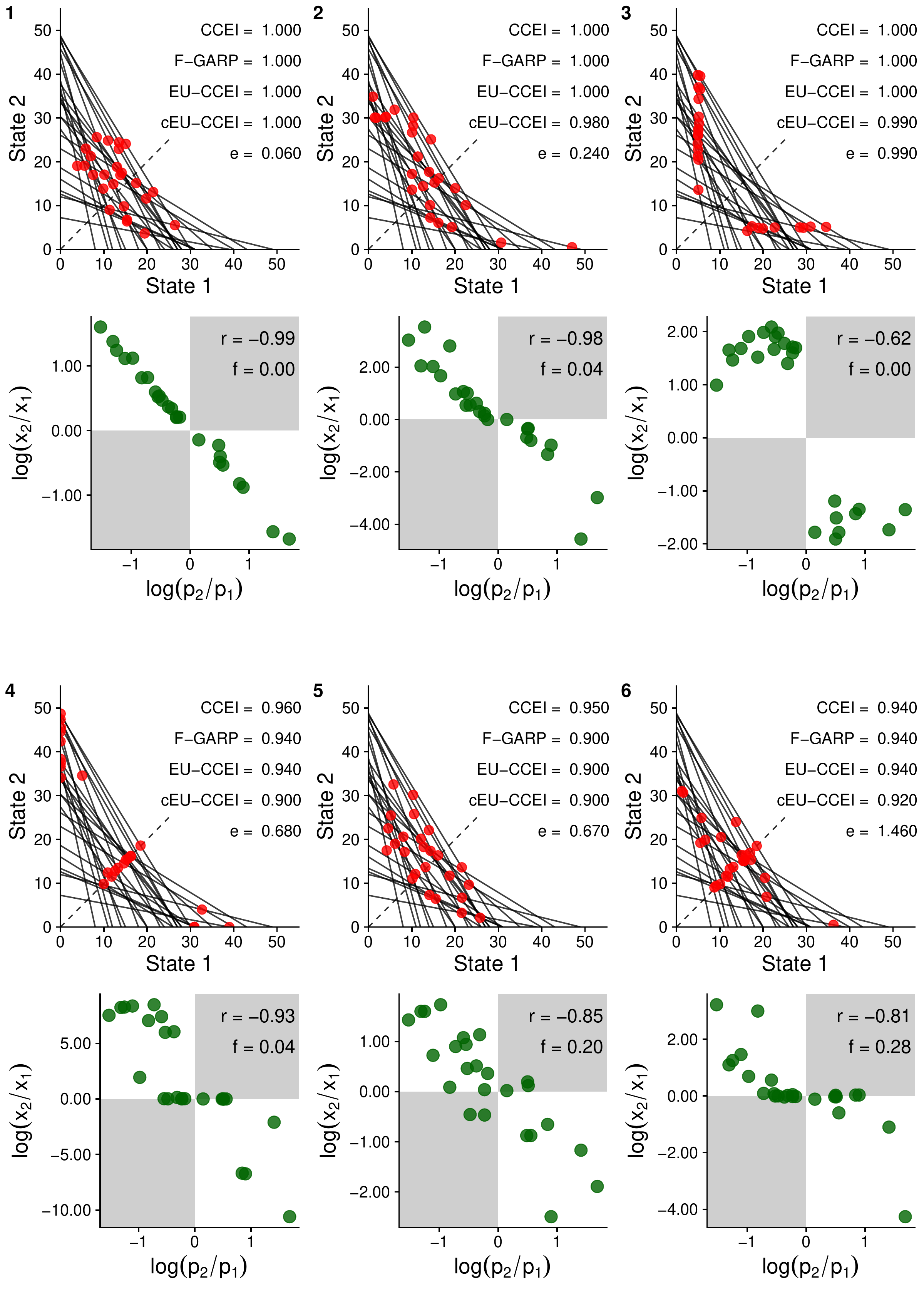} 
\end{figure}

\begin{figure}[ht!]
\centering 
\includegraphics[width=0.9\textwidth]{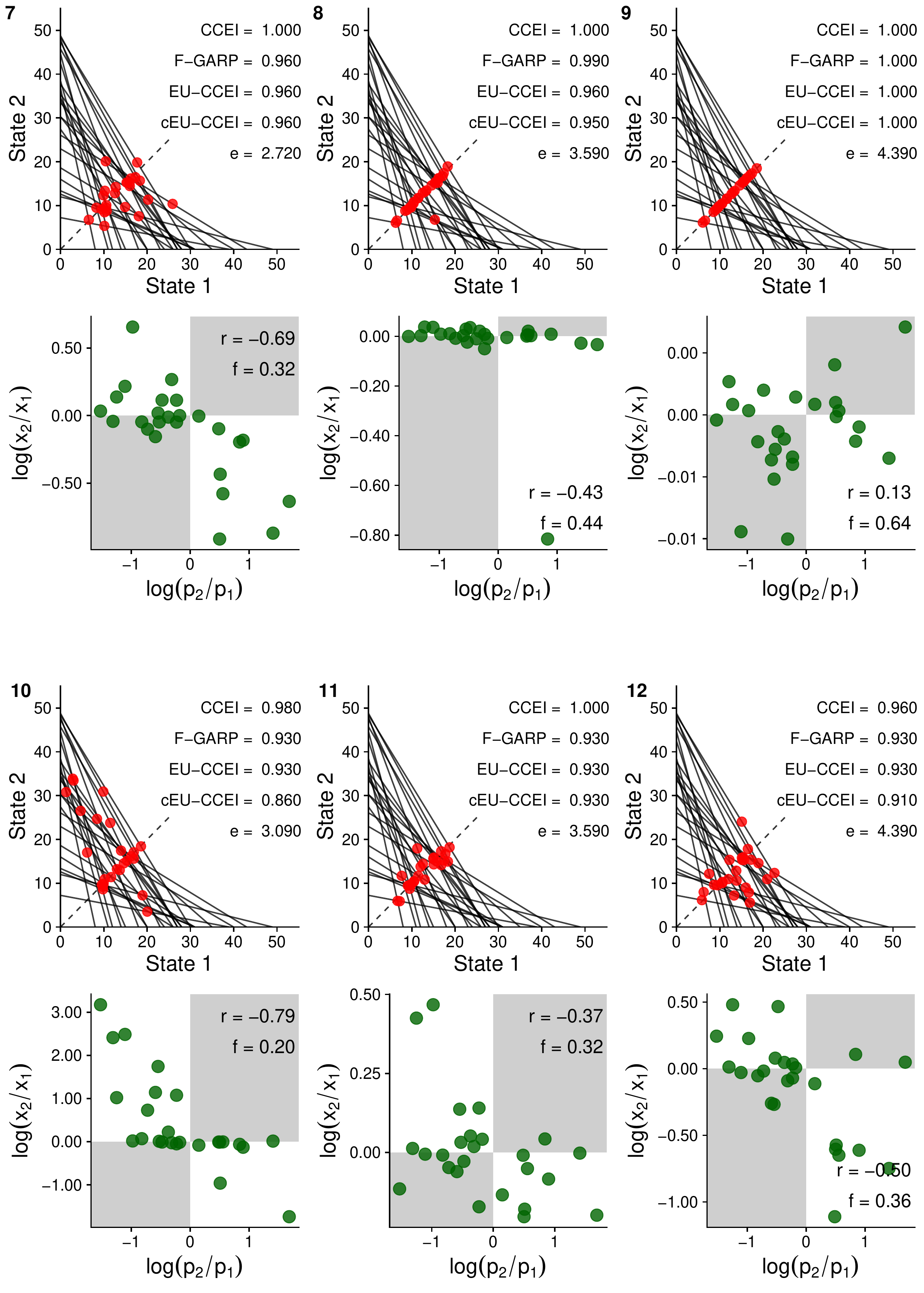} 
\end{figure}

\begin{figure}[ht!]
\centering 
\includegraphics[width=0.9\textwidth]{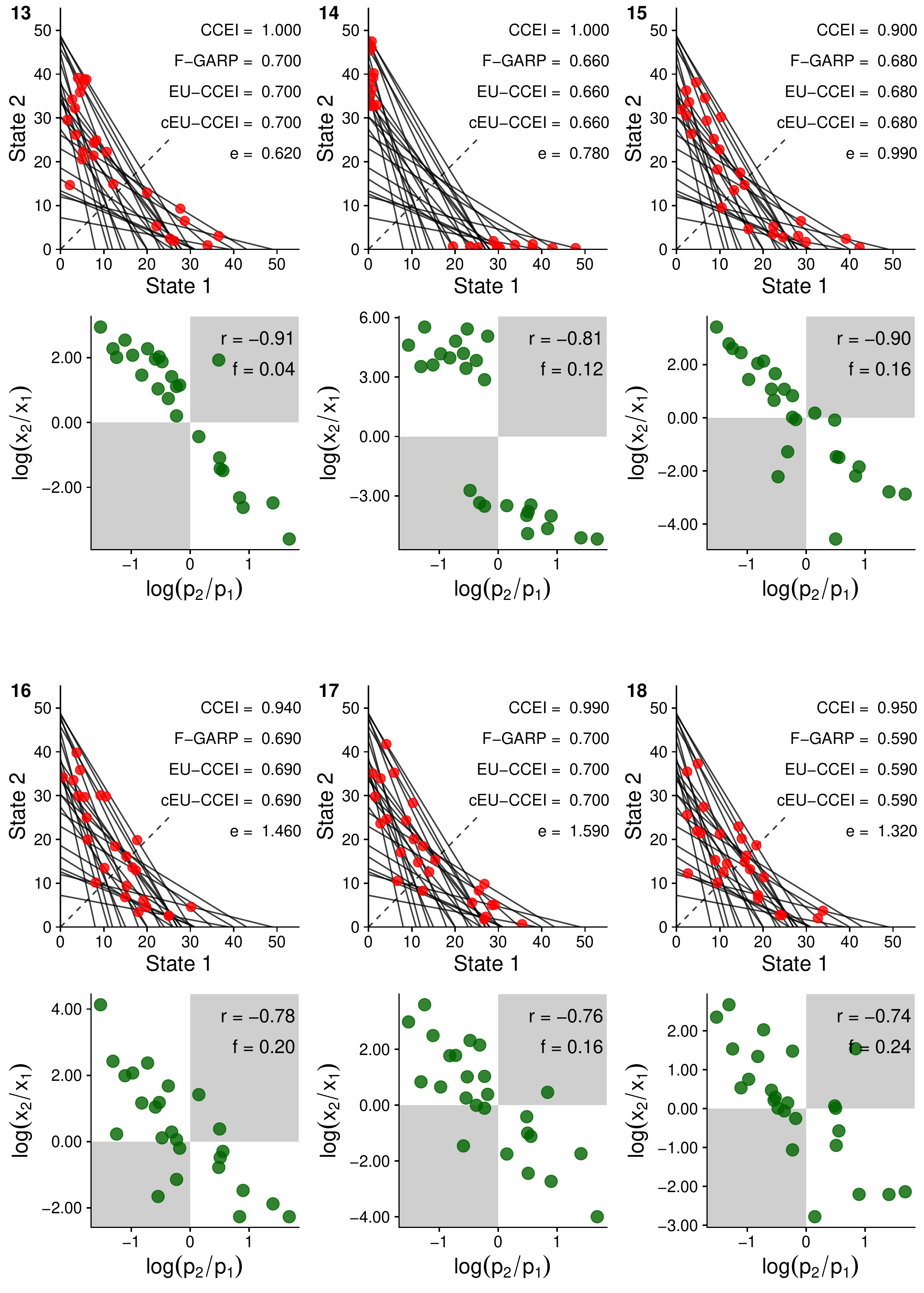} 
\end{figure}

\begin{figure}[ht!]
\centering 
\includegraphics[width=0.9\textwidth]{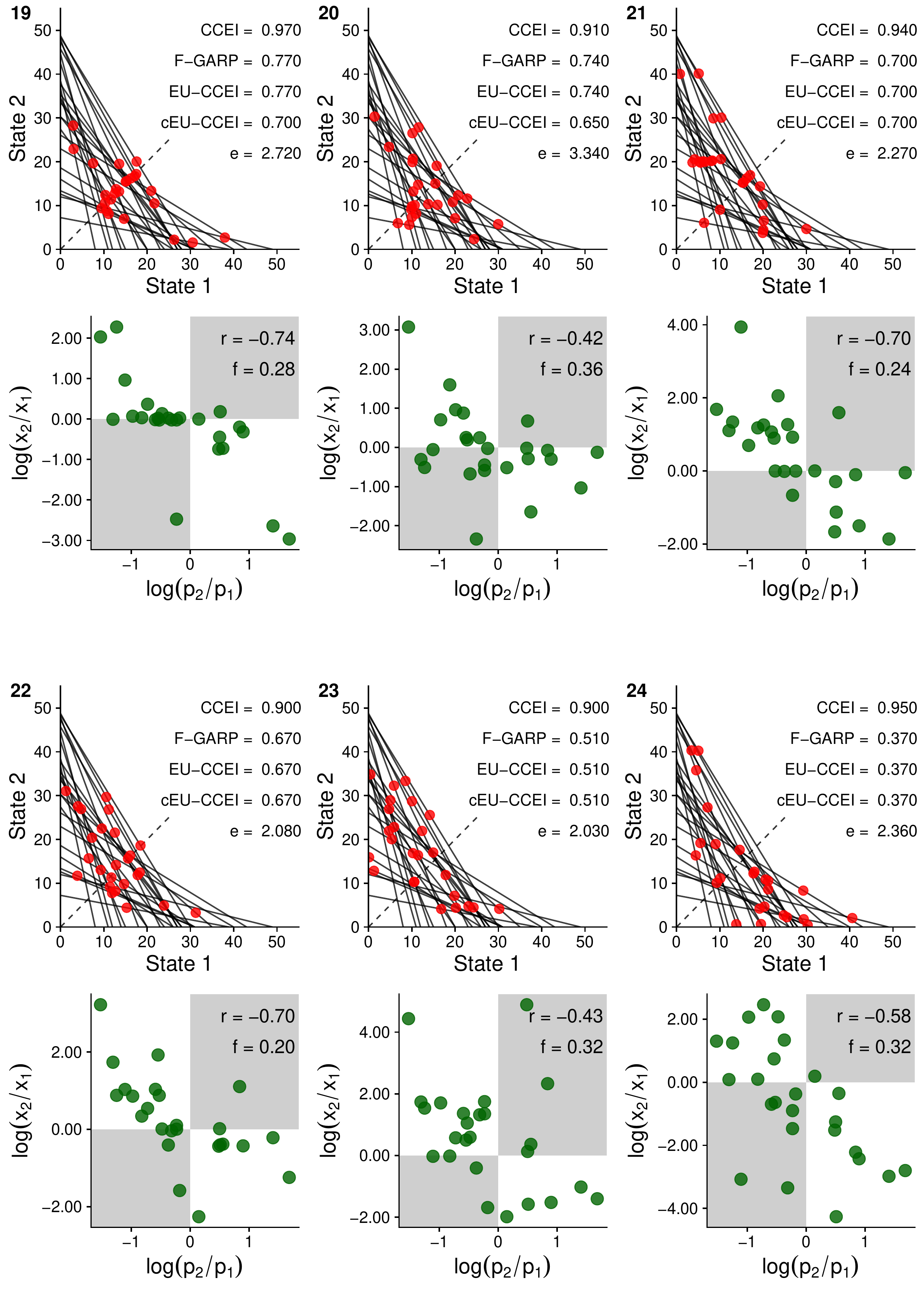} 
\end{figure}

\begin{figure}[ht!]
\centering 
\includegraphics[width=0.9\textwidth]{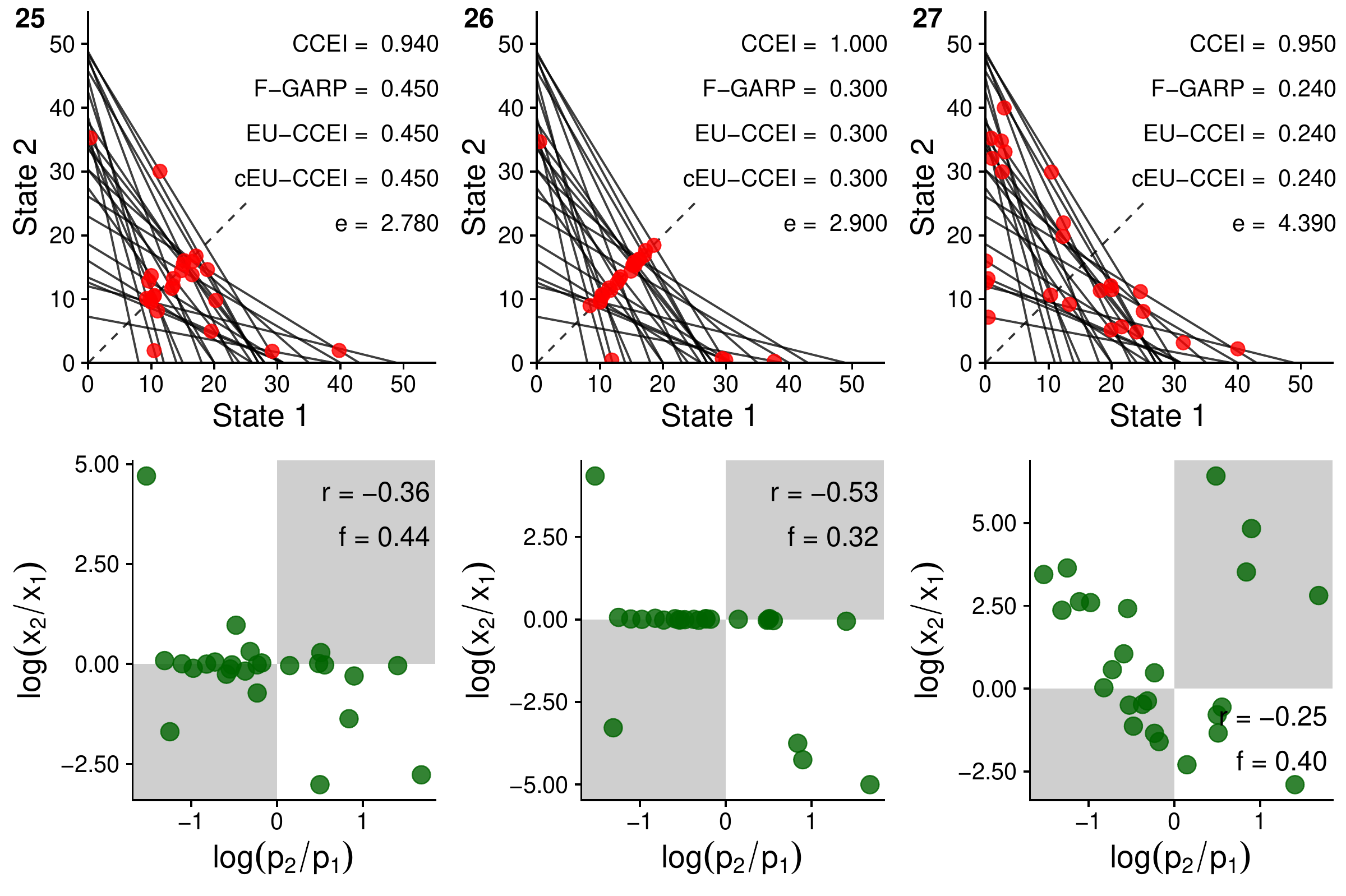} 
\end{figure}

\clearpage 
\renewcommand{\baselinestretch}{0}
\bibliographystyleSOMpapers{ecta}
\bibliographySOMpapers{approximate_eu}

\end{document}